\documentclass[11pt]{article}
\usepackage{amsthm,amssymb,amsmath}
\usepackage{graphicx,algorithmic,algorithm}
\usepackage{enumerate}
\usepackage{tikz,tikz-cd}
\usepackage{tabularx}
\usepackage{multirow}
\usepackage{makecell}
\usepackage{csquotes}

\tikzcdset{scale cd/.style={every label/.append style={scale=#1},
    cells={nodes={scale=#1}}}}
\usetikzlibrary{shapes,fit}
\usepackage[letterpaper, margin=1in]{geometry}
\usepackage{thmtools}
\usepackage{hyperref}
\usepackage[capitalize]{cleveref}

\usepackage{todonotes}
\usepackage[normalem]{ulem}
\hypersetup{
	pdftitle=   {},
	pdfauthor=  {}
}
\usepackage{authblk}

\usepackage{caption}
\usepackage{subcaption}
\usepackage{commath}

\crefname{figure}{Figure}{Figures}
\crefname{equation}{Equation}{Equations}

\renewcommand{\leq}{\leqslant}

\renewcommand{\le}{\leqslant}
\renewcommand{\ge}{\geqslant}

\newtheorem{theorem}{Theorem}[section]
\newtheorem{lemma}[theorem]{Lemma}
\newtheorem{proposition}[theorem]{Proposition}
\newtheorem{corollary}[theorem]{Corollary}

\newtheorem{claim}[theorem]{Claim}

\newtheorem*{thm:main1}{Theorem~\ref{thm:main1}}
\newtheorem*{thm:main2}{Theorem~\ref{thm:main2}}
\newenvironment{clproof}{\begin{list}{}{
			\setlength{\leftmargin}{3mm}
		} \item {\it Proof.} }{\hfill$\lozenge$\end{list}}

\newcommand\extrafootertext[1]{
    \bgroup
    \renewcommand\thefootnote{\fnsymbol{footnote}}
    \renewcommand\thempfootnote{\fnsymbol{mpfootnote}}
    \footnotetext[0]{#1}
    \egroup
}

      \newcommand{\tikzwall}[6]{

            \begin{scope}
                \pgfmathtruncatemacro{\c}{ #1 * 2 - 1 }
                \pgfmathtruncatemacro{\r}{ #2 - 1}
                \pgfmathtruncatemacro{\basex}{#3}
                \pgfmathtruncatemacro{\basey}{#4}
                \pgfmathtruncatemacro{\bm}{Mod(\basex+\basey,2)}
                \pgfmathtruncatemacro{\bmi}{Mod(\basex+\basey+1,2)}
                
                \ifthenelse{\c > 1 \and \r > 0}{
                    \foreach \x in {0, ..., \c}{
                        \ifthenelse{0 < \x \and \x < \c}{
                            \pgfmathtruncatemacro{\start}{0}
                            \pgfmathtruncatemacro{\t}{\r}
                        }{
                            \ifthenelse{\x = 0}{
                                \pgfmathtruncatemacro{\start}{\bmi}
                                \pgfmathtruncatemacro{\t}{\r - Mod(\r+\bm,2)}
                            }{
                                \pgfmathtruncatemacro{\start}{\bm}
                                \pgfmathtruncatemacro{\t}{\r - Mod(\r+\bmi,2)}
                            }
                        }
                        \foreach \y in {\start,...,\t} {
                            \ifthenelse{\start < \t}{
                                \node at (\basex+\x,\basey+\y) [#5] (v\x-\y) {};
                            }{}
                        }
                    }
                    \foreach \x in {0, ..., \c}{
                        \foreach \y in {0,...,\r} {
                            \pgfmathtruncatemacro{\nextx}{\x+1}
                            \pgfmathtruncatemacro{\nexty}{\y+1}
                            \pgfmathtruncatemacro{\sum}{\x+\y+\bm}
                            \pgfmathtruncatemacro{\p}{Mod(\r+\bmi,2)}
                            \pgfmathtruncatemacro{\q}{Mod(\r+\bm,2)}
                            
                            \ifthenelse{\isodd{\sum} \and \y < \r}{
                                \draw [#6] (v\x-\y) edge (v\x-\nexty);
                            }{}
                            
                            \ifthenelse{\x < \c}{
                                \ifthenelse{0 < \y \and \y < \r}{
                                    \draw [#6] (v\x-\y) edge (v\nextx-\y);
                                }{
                                    \ifthenelse{\y = 0 \and \the\numexpr \bmi - 1 \relax < \x \and \x < \the\numexpr \c - \bm \relax}{
                                        \draw [#6] (v\x-\y) edge (v\nextx-\y);
                                    }{}
                                    \ifthenelse{\y = \r \and \the\numexpr \q - 1 \relax < \x \and \x < \the\numexpr \c - \p \relax}{
                                        \draw [#6] (v\x-\y) edge (v\nextx-\y);
                                    }{}
                                    
                                }
                            }{
                            }
                        }
                    }
                }{}
            \end{scope}
        }
        
\newcommand{\Nn}{\mathbb{N}}
\newcommand{\Bc}{\mathcal{B}}

\newcommand{\Pc}{\mathcal{P}}
\DeclareMathOperator{\rk}{rk}

\newcommand{\redG}{\mathcal{R}}

\newcommand{\dist}{\mathsf{dist}}
\newcommand{\Ball}{\mathsf{Ball}}
\DeclareMathOperator{\diam}{diam}
\DeclareMathOperator{\stw}{stw}
\DeclareMathOperator{\str}{str}
\DeclareMathOperator{\conv}{conv}
\newcommand{\reduced}[1]{{#1}^{\downarrow}}

\DeclareMathOperator{\maxleaf}{ml}
\DeclareMathOperator{\compmaxleaf}{cml}
\DeclareMathOperator{\trace}{trace}
\newcommand{\pairs}{\mathsf{Endpair}}

\newcommand{\pow}[1][]{\mathsf{pow}^{#1}}
\newcommand{\blow}[1][]{\mathsf{blow}^{#1}}

\title{Moderately Beyond Clique-Width: \\ Reduced Component Max-Leaf and Related Parameters}
	\author[1]{\'Edouard Bonnet}
	\author[2]{Yeonsu Chang}
	\author[3]{Julien Duron}
	\author[4]{Colin Geniet}
	\author[5]{O-joung Kwon}

    \affil[1]{Univ Lyon, CNRS, ENS de Lyon, Universit\'{e} Claude Bernard Lyon 1, France}
	\affil[2,5]{Department of Mathematics and Research Institute for Natural Sciences, Hanyang University, Seoul, South Korea}
    \affil[3]{Institute of Informatics, University of Warsaw, Poland}
	\affil[4,5]{Discrete Mathematics Group, Institute for Basic Science (IBS), Daejeon, South Korea}

	\date\today

\begin{document}

\maketitle
\extrafootertext{Y. Chang and O. Kwon are supported by the National Research Foundation of Korea (NRF) grant funded by the Ministry of Science and ICT (No. RS-2023-00211670). C. Geniet and O. Kwon are supported by the Institute for Basic Science (IBS-R029-C1).}

	\extrafootertext{E-mail addresses: \texttt{edouard.bonnet@ens-lyon.fr} (\'E. Bonnet), 
    \texttt{yeonsu@hanyang.ac.kr} (Y. Chang),
    \texttt{j.duron@uw.edu.pl} (J. Duron),
    \texttt{colin@ibs.re.kr} (C. Geniet),
    \texttt{ojoungkwon@hanyang.ac.kr} (O. Kwon),
 }

\begin{abstract}
  Reduced parameters [BKW, JCTB '26; BKRT, SODA '22] are defined via contraction sequences.
  Based on this framework, we introduce the reduced component max-leaf, denoted by $\compmaxleaf^\downarrow$, where component max-leaf is the maximum number of leaves in any spanning tree of any connected component.
  Reduced component max-leaf is strictly sandwiched between clique-width and reduced bandwidth, it is bounded in unit interval graphs, and unbounded in planar graphs. 
  We design polynomial-time algorithms for problems such as \textsc{Maximum Independent Set}, \textsc{Maximum Clique}, \textsc{Maximum Induced $d$-Regular Subgraph}, and \textsc{Induced Disjoint Paths} in graphs given with a~contraction sequence witnessing low $\compmaxleaf^\downarrow$, unifying and extending tractability results for classes of bounded clique-width and unit interval graphs.

  We get the following collapses in sparse classes of bounded $\compmaxleaf^\downarrow$: bounded maximum degree implies bounded treewidth, whereas $K_{t,t}$-subgraph-freeness implies strongly sublinear treewidth; we show the latter, more generally, for classes of bounded reduced cutwidth.
  We establish the former result by showing that graphs with bounded $\reduced{\compmaxleaf}$ admit balanced separators dominated by a~bounded number of vertices.
  In contrast, there are graphs $G$ of arbitrarily large girth and treewidth $\Theta(|V(G)|^{1/2})$ such that $\compmaxleaf^\downarrow(G) \leqslant 3$.

  We then showcase an application of the reduced parameters to establishing non-transducibility results. 
  We prove that for most reduced parameters $p^\downarrow$ (including reduced bandwidth), the family of classes of bounded $p^\downarrow$ is closed under first-order transductions.
  We then answer a question of [BKW '26] by showing that the 3-dimensional grids have unbounded reduced bandwidth.
  As the class of planar graphs (or any class of bounded genus) has bounded reduced bandwidth [BKW '26], this reproves a recent result [GPP, LICS '25] that planar graphs do not first-order transduce the 3-dimensional grids.
\end{abstract}

\section{Introduction}
Contraction sequences\footnote{See~\cref{sec:prelim} for definitions and notation.} associate $n$-vertex graphs with sequences of~$n$ graphs of ever smaller vertex count, called \emph{red graphs}.
\emph{Twin-width}~\cite{twin-width1} is the first parameter to be defined via contraction sequences, as the largest maximum degree of the red graphs, minimized over all contraction sequences.
The so-called \emph{reduced parameters}, introduced in~\cite{reduced-bdw,twin-width6}, generalize the link between maximum degree $\Delta$ and twin-width:
For any graph parameter~$p$, \emph{reduced-$p$}, denoted by $p^\downarrow$, is the largest~$p$ over all red graphs of the contraction sequence, minimized among all contraction sequences.
Every graph admits a~contraction sequence where each red graph consists of a~star and isolated vertices.
So, $p^\downarrow$ is only relevant when bounded $p$ implies bounded maximum degree.
The interesting reduced parameters are thus less general than twin-width.\footnote{As an abuse of language, we say that $p$ is less general than $q$ if classes of bounded $p$ have bounded $q$.}
We denote it by $p \sqsubset q$ if every class of bounded $q$ has bounded $p$, but not vice versa.  

If $p = \star$ is the size of a~largest connected component, then $p^\downarrow$ is tied to clique-width~\cite{twin-width6,BarilCL26}.
Thus, by the Courcelle--Makowsky--Rotics theorem~\cite{CMR00} monadic second-order logic (with modular counting) can be efficiently model-checked in classes of bounded $\star^\downarrow$, whereas first-order model checking is tractable in classes of effectively\footnote{We say that a~class has \emph{effectively bounded $p^\downarrow$} if there is a~constant~$c$ and a~polynomial-time algorithm that, given a~graph~$G$ from the class, outputs a~contraction sequence of~$G$ such that for every red graph $H$ of the sequence, $p(H) \leqslant c$ holds.} bounded $\Delta^\downarrow$, i.e., twin-width~\cite{twin-width1}.
This suggests the following line of work:
\begin{quote}
  For parameters $p$ with $\Delta \sqsubset p \sqsubset \star$, design polynomial algorithms for a~strict superset of the problems definable in first-order logic on graphs of effectively bounded $p^\downarrow$.
\end{quote}
Arguably the easiest-to-formulate graph problem that \emph{cannot} be expressed in first-order logic is \textsc{Maximum Independent Set}, which on graph $G$ seeks $\max_{X \subseteq V(G)} |X|$ subject to the formula $\forall x, y \in X, \neg E(x,y)$.
It is thus reasonable to expect that the superset contains this problem. 

There are two main families of graphs with bounded twin-width on which \textsc{Maximum Independent Set} remains NP-hard: planar graphs~\cite{twin-width1,HlinenyJ25} and $f(n)$-subdivisions of $n$-vertex graphs with $2 \log n \leqslant f(n) \leqslant n^{O(1)}$~\cite{BergeBonnetDepres,twin-width2}.
In the work explicitly introducing the reduced parameters~\cite{reduced-bdw}, the focus is on reduced bandwidth.
The main result is that planar graphs have bounded reduced bandwidth; so do $(\geqslant n)$-subdivisions of $n$-vertex graphs.
Since many problems, like \textsc{Maximum Independent Set}, remain NP-hard on planar graphs, the parameter $p$ we consider cannot be as general as bandwidth.

An attempt is \emph{stretch-width}~\cite{BonnetDuron23}, which, while strictly speaking not a~reduced parameter, is strongly inspired by this framework.
Classes of bounded stretch-width generalize those of bounded clique-width, and fittingly do not contain the class of planar graphs or of polynomial subdivisions.
However it is currently open whether \textsc{Maximum Independent Set} can be solved in polynomial time in every class of bounded stretch-width (even with a~witness of low stretch-width).
Partial results comprise a~$2^{\Tilde O(n^{4/5})}$-time algorithm, and a~polynomial-time algorithm when the maximum degree is also bounded. 

In this paper, we introduce the \emph{component max-leaf} parameter, denoted by $\compmaxleaf$, the least integer~$k$ such that every spanning tree of every connected component has at~most~$k$ leaves, and its reduced counterpart $\compmaxleaf^\downarrow$.
Connected graphs have bounded component max-leaf if and only if they are subdivisions of graphs of bounded size (see for instance \cite[Theorem 4.7]{ZickfeldThesis}). 
We have $\text{bandwidth} \sqsubset \compmaxleaf \sqsubset \star$, which sets $\compmaxleaf^\downarrow$ in the ``right'' regime.

Since graphs of twin-width at~most~2 have bounded $\compmaxleaf^\downarrow$, unit interval graphs---which have twin-width at~most~2~\cite{twin-width3} and unbounded clique-width---testify that $\compmaxleaf^\downarrow$ is \emph{strictly} more general than clique-width.
We actually prove that unit interval graphs have unbounded stretch-width, which is consistent with the conjecture that classes of bounded stretch-width have at~most logarithmic clique-width~\cite{BonnetDuron23} (as there are $n$-vertex unit interval graphs with clique-width~$\Omega(\sqrt n)$).

We show in \cref{thm:cml-dominated-balanced-sep} that every graph $G$ with $\reduced{\compmaxleaf}(G)=k$ admits a balanced separator dominated by at most $k(k+2)$ vertices.
This implies that every graph~$G$ with maximum degree~$\Delta$ and $\reduced{\compmaxleaf}(G) = k$ has treewidth $O(\Delta k^2)$.  
Since planar graphs of arbitrarily large treewidth admit arbitrarily large subdivided walls or their line graphs as induced subgraphs~\cite[Theorem 1.1]{AboulkerAKST21}, this establishes the following fact.

\begin{theorem}\label{thm:in-planar-graphs}
  A class of planar graphs has bounded treewidth if and only if it has bounded $\compmaxleaf^\downarrow$.
\end{theorem}

We exhibit weakly sparse (i.e., without $K_{t,t}$ subgraph for some fixed integer~$t$) classes of bounded $\compmaxleaf^\downarrow$ and polynomially large treewidth.
More precisely, for every positive integer $g$, there is an infinite family of $n$-vertex graphs $G$ of girth at~least~$g$ and treewidth $\Theta_g(\sqrt n)$ such that $\compmaxleaf^\downarrow(G) \leqslant 3$.
On the other hand, we show that weakly sparse classes of bounded $\compmaxleaf^\downarrow$ (even of bounded reduced cutwidth) have polynomial expansion.
Thus, their $n$-vertex subgraphs have treewidth $O(n^\beta)$ for some constant $\beta < 1$. This answers a question in~\cite{reduced-bdw}.

\begin{theorem}\label{thm:weaklysparse}
  Weakly sparse classes of bounded reduced $\compmaxleaf^\downarrow$ (even of bounded reduced cutwidth) have polynomial expansion.
\end{theorem}

Exponential subdivisions of grids have unbounded $\compmaxleaf^\downarrow$ by~\cref{thm:in-planar-graphs}, but bounded stretch-width~\cite{BonnetDuron23}.
This gives the other direction for the incomparability of $\compmaxleaf^\downarrow$ and stretch-width.
We show that these parameters are also incomparable to \emph{mim-width}~\cite{VatshelleThesis}, another (moderate) generalization of clique-width lending itself to algorithmic meta-theorems~\cite{BuixuanTV2013,BergougnouxK2021,Bergougnoux23}.

\begin{theorem}\label{thm:incomparable}
  Reduced component max-leaf, stretch-width, and mim-width are all incomparable.
\end{theorem}

We then design polynomial-time algorithms on graphs of effectively bounded $\compmaxleaf^\downarrow$ for a~family of problems containing \textsc{Maximum Independent Set}. 
Fix some finite set $\sigma$ of non-negative integers.
Let \textsc{$\sigma$-Neighborhood} be the problem, given a~graph $G$, of finding a~maximum-cardinality subset $S \subseteq V(G)$ such that every $v \in S$ has a~number of neighbors in~$S$ that is in $\sigma$.
Setting $\sigma = \{0\}$ defines \textsc{Maximum Independent Set}, while $\sigma = \{1\}$ defines \textsc{Maximum Induced Matching}.
This framework was first proposed by Telle and Proskurowski~\cite{TelleP1997}.

\begin{theorem}\label{thm:alg}
  There is an algorithm that, given an $n$-vertex graph $G$ and a~contraction sequence witnessing that $\compmaxleaf^\downarrow(G) \leqslant t$, solves \textsc{$\sigma$-Neighborhood} on~$G$ in time $n^{O(t\alpha^2)}$ with $\alpha=\max(\sigma)$. 
\end{theorem}

Bonnet et al.~\cite{twin-width3} proved that \textsc{Maximum Independent Set} can be solved in polynomial time on graphs of twin-width at most $2$, provided that a~contraction sequence of width~$2$ is given as part of the input.
Since graphs of twin-width at most $2$ have reduced component max-leaf at~most~$2$, \Cref{thm:alg} generalizes this result.
Bui-Xuan, Telle, and Vatshelle~\cite{BuixuanTV2013} showed that the \textsc{$\sigma$-Neighborhood} problem is polynomial time solvable on classes of effectively bounded mim-width, for any finite or cofinite set $\sigma$.

We also consider a~connectivity problem, namely the \textsc{Induced Disjoint Paths} problem, where, given a set of pairs, the task is to find a~set of paths connecting the pairs such that no two paths share a~vertex or have adjacent vertices.
This problem is known to be polynomial-time solvable on classes of effectively bounded mim-width~\cite{JaffkeKT2020}.
We obtain a~polynomial-time algorithm for classes of effectively bounded $\compmaxleaf^\downarrow$.
Using this algorithm together with a simple reduction (see~\cite{JaffkeKT2020}), one can solve in polynomial time \textsc{$H$-Induced Subdivision Containment}, that is, the problem of finding an induced subdivision of a~fixed graph $H$.

\begin{theorem}\label{thm:alg2}
  There is an algorithm that, given an $n$-vertex graph $G$ and a~contraction sequence witnessing that $\compmaxleaf^\downarrow(G) \leqslant t$, solves \textsc{Induced Disjoint Paths} on~$G$ in time $n^{O(t)}$.  
\end{theorem}

In the conclusion, we raise open questions on problems parameterized by reduced $\compmaxleaf$.

A~by-product of reduced parameters is that they provide a~new tool to show that a~class~$\mathcal C$ does not first-order transduce another class~$\mathcal D$.
This is identified as a~current gap~\cite{Gajarsky25}:
\enquote{\emph{the biggest bottleneck to obtaining a fine understanding of the first-order transduction quasi-order is the lack of a robust toolbox for proving non-transducibility results that extends beyond commonly studied transduction invariant properties.}}
We show that for most reduced parameters $p^\downarrow$, the family of classes of bounded $p^\downarrow$ is closed under first-order transductions.

\begin{theorem}\label{thm:transduction-intro}
  For every monotone parameter $p$ closed under powers and blowups, the family of classes of bounded $p^\downarrow$ is closed under first-order transductions.
\end{theorem}

Then if $\mathcal C$ has bounded $p^\downarrow$ but~$\mathcal D$ has unbounded $p^\downarrow$, this indeed proves that~$\mathcal C$ does not transduce~$\mathcal D$.
For example, we reprove the result of~\cite{Gajarsky25} with $p = \text{bandwidth}$, $\mathcal C=\{\text{planar graphs}\}$ ($\mathcal C=\{\text{graphs of genus}~\leqslant g\}$), and $\mathcal D=\{\text{3-dimensional grids}\}$ since planar graphs (graphs of bounded genus) have bounded reduced bandwidth~\cite{reduced-bdw}, but we show that 3-dimensional grids have unbounded reduced bandwidth (\Cref{thm:3d-grids-rbandw}).

\Cref{fig:hasse-diag} summarizes the context and content of the paper.
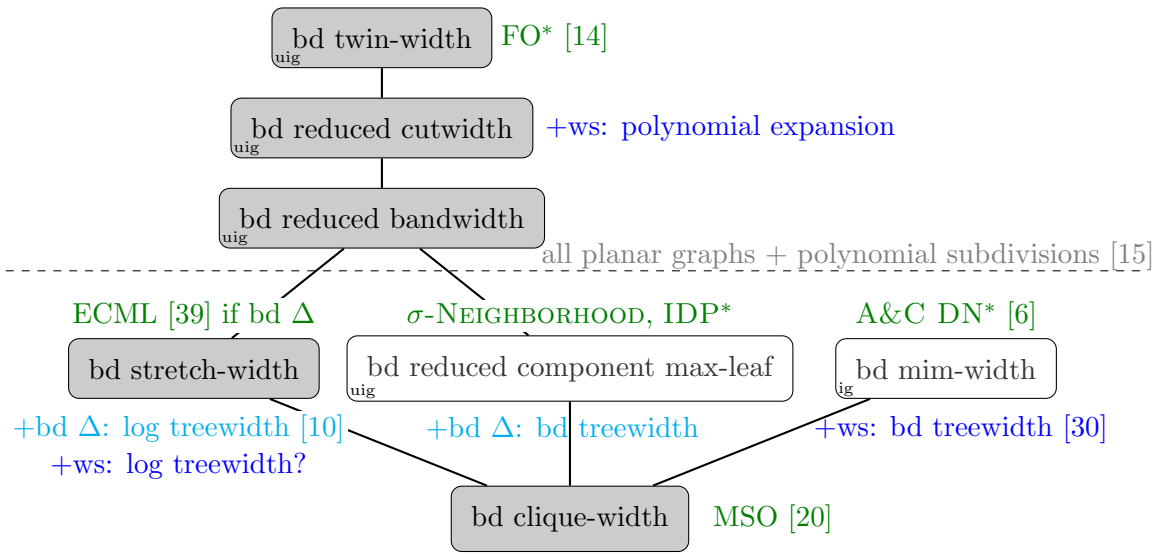
\begin{figure}[!ht]
  \centering
    \begin{tikzpicture}
\def\s{1}

\foreach \i/\j/\l/\t/\b/\x in {0/0/clique-width/cw/black/, -5/2/stretch-width/stw/black/, 0/2/{reduced component max-leaf}/rcml/white/uig, 5/2/mim-width/mimw/white/ig, -2.5/4/reduced bandwidth/rbdw/black/uig, -2.5/5.2/reduced cutwidth/rcutw/black/uig, -2.5/6.4/twin-width/tww/black/uig}{
\node (\t) at (\i * \s,\j * \s) {bd \l} ;
\node[draw, rounded corners, fill=\b, fill opacity=0.2, fit = (\t)] (B\t) {} ;
\node at (B\t.south west) [anchor=south west, inner sep=1pt] {\tiny{\x}} ;
}

\foreach \i/\j in {cw/stw, cw/rcml, cw/mimw, stw/rbdw, rcml/rbdw, rbdw/rcutw, rcutw/tww}{
\draw[thick] (B\i) -- (B\j) ;
}

\foreach \i/\j/\l in {-5.2/0.7/{+ws: log treewidth?}, 2/5.2/{+ws: polynomial expansion}}{
\node[fill=white, inner sep=1pt] at (\i * \s,\j * \s) {\textcolor{blue}{\l}} ;
}

\foreach \i/\j/\l in {5.2/1.2/{+ws: bd treewidth~\cite{Gurski00}}}{
\node at (\i * \s,\j * \s) {\textcolor{blue}{\l}} ;
}

\foreach \i/\j/\l in {-5.2/1.2/{+bd $\Delta$: log treewidth~\cite{BonnetDuron23}}}{
\node[fill=white, inner sep=1pt] at (\i * \s,\j * \s) {\textcolor{cyan}{\l}} ;
}

\foreach \i/\j/\l in {-0.1/1.2/{+bd $\Delta$: bd treewidth}}{
\node at (\i * \s,\j * \s) {\textcolor{cyan}{\l}} ;
}

\foreach \i/\j/\l in {-5/2.7/{ECML~\cite{Pilipczuk11}~if bd $\Delta$}, 0/2.7/{\textsc{$\sigma$-Neighborhood}, \textsc{IDP}$^\ast$}, 5/2.7/{\textsc{A\&C DN}$^\ast$~\cite{Bergougnoux23}}, -0.2/6.4/{FO$^\ast$~\cite{twin-width1}}, 2.7/0/{MSO~\cite{CMR00}}}{
\node[fill=white, inner sep=1pt] at (\i * \s,\j * \s) {\textcolor{green!50!black}{\l}} ;
}

\draw[thin,dashed] (-7.5,3.3) -- (7.6,3.3) ;
\node at (3.7,3.5) {\textcolor{gray}{all planar graphs + polynomial subdivisions~\cite{reduced-bdw}}} ;
\end{tikzpicture}
  \caption{Hasse diagram of the handled families of classes.
    The shaded families are those closed under first-order transductions.
    We use the following abbreviations: ``bd'' $\rightarrow$ bounded,  ``ws'' $\rightarrow$ weakly sparse, ``log treewidth'' $\rightarrow$ $n$-vertex graphs have treewidth $O(\log n)$, ``uig'' $\rightarrow$ contains the class of unit interval graphs, ``ig'' $\rightarrow$ contains the class of interval graphs (at the bottom-left corner of the boxes).
    We indicate in green the problems solvable in polynomial time, and refer the reader to the corresponding papers for the definition of the logics.
  The asterisks mean that a~witness for low parameter is assumed by the algorithm.}
  \label{fig:hasse-diag}
\end{figure}

Admittedly, the only reduced parameter that does not fulfill the preconditions of~\cref{thm:transduction-intro} (nor its conclusion) is the one we introduce in the current paper: reduced component max-leaf.
However this is also precisely what makes it suited to unify the polynomial-time algorithms for \textsc{Maximum Independent Set} and \textsc{Maximum Induced Matching} in classes of bounded clique-width and unit interval graphs.
Indeed, unit interval graphs first-order transduce the class of grid graphs (induced subgraphs of grids), on which \textsc{Maximum Induced Matching} is NP-hard~\cite{Ekim13}.

Bridging efficient first-order and monadic second-order model checking is an active line of work. 
We explore here the class-first side (similarly to~\cite{Bergougnoux23,BonnetDuron23}); for the logic-first side, we refer the reader to~\cite{PilipczukSSTV22,SchirrmacherSST24,FominGSST25,Bojanczyk25} and the references therein.

\section{Preliminaries}\label{sec:prelim}

Let $\mathbb{N}$ denote the set of positive integers.
For an integer $n$, let $[n]$ denote the set of all positive integers at~most~$n$. Let $\triangle$ denote the symmetric difference.  All graphs in this paper are finite, undirected, and simple.

Let $G$ be a graph.
We denote by $V(G)$ the vertex set and $E(G)$ the edge set of $G$.
For a~set $S \subseteq V(G)$, let $G[S]$ denote the subgraph of $G$ induced by $S$.
For $v \in V(G)$, let $N_G(v) := \{u \in V(G) : uv \in E(G)\}$ be the \emph{neighborhood} of~$v$ in~$G$.
For $S \subseteq V(G)$, let $N_G(S)$ be the set of vertices in $V(G) \setminus S$ having a neighbor in~$S$.
Let $N_G[S] = N_G(S) \cup S$. 
For graphs $G$ and $H$, the \emph{union} of~$G$ and~$H$ is the graph ${G\cup H:=(V(G)\cup V(H),E(G)\cup E(H))}$.

Given two graph parameters $p, q$, we write $p \sqsubseteq q$ if there is a~function~$f$ such that for every graph $G$, $p(G) \leqslant f(q(G))$.
Then, $p, q$ are \emph{tied} if $p \sqsubseteq q$ and $q \sqsubseteq p$, and \emph{incomparable} if both $p \sqsubseteq q$ and $q \sqsubseteq p$ fail.
We write $p \sqsubset q$ when $p \sqsubseteq q$ holds but $q \sqsubseteq p$ does not, and may say that $q$ is \emph{larger} than~$p$.

For two disjoint subsets $X$ and $Y$ of vertices of a graph $G$, 
\begin{itemize}
    \item $X$ is \emph{complete} to $Y$ in $G$ if for all $x\in X$, $y\in Y$, $x$ is adjacent to $y$,
    \item $X$ is \emph{anti-complete} to $Y$ in $G$ if for all $x\in X$, $y\in Y$, $x$ is not adjacent to $y$,
    \item $X$ and $Y$ are \emph{homogeneous} in $G$ if $X$ is complete or anti-complete to $Y$ in $G$.
\end{itemize}
A \emph{clique} of a graph~$G$ is a set of pairwise adjacent vertices of~$G$, and an \emph{independent set} of~$G$ is a set of pairwise non-adjacent vertices of~$G$.

For two vertices $v,w$ in $G$, 
a \emph{${(v,w)}$-path} in $G$ is a path from $v$ to $w$ in $G$.

Given a~graph $G$ and a~weight function $w: V(G) \to \mathbb R_+$, a~\emph{balanced separator} of weighted graph $(G,w)$ is a~set $X \subseteq V(G)$ such that for every connected component $C$ of $G-X$, $w(V(C)) \leqslant \frac{2}{3} w(V(G))$, where $w(S) := \sum_{v \in S} w(v)$ for every $S \subseteq V(G)$.
A~\emph{balanced separator} of an unweighted graph $G$ is a~balanced separator of the weighted graph where every vertex of~$G$ is assigned the same positive weight.

A \emph{separation} in~$G$ is a pair~$(A,B)$ of subsets of~$V(G)$ so that $V(G) = A \cup B$, and there is no edge between~$A \setminus B$ and $B \setminus A$.
In other words, it corresponds to a choice of any bipartition of the connected components of $G - A \cap B$ into $A \setminus B$ and $B \setminus A$.

  \begin{figure}[!t]
    \centering
        \begin{tikzpicture}
            [scale=0.6]
            \tikzset{vx/.style = {circle, draw, fill=black!0, inner sep=0pt, minimum width=4pt}}
            \tikzset{vx2/.style = {circle, draw, fill=black!50, inner sep=0pt, minimum width=4pt}}
            
            \tikzwall{7}{7}{0}{0}{vx2}{black}
            
        \end{tikzpicture}
        \caption{A $(7,7)$-wall. }
        \label{fig:wall}
    \end{figure}
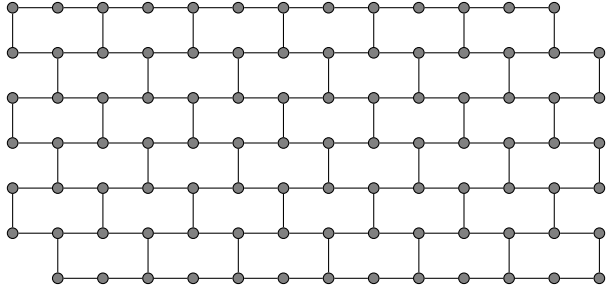

For an integer $n\ge 3$, the \emph{$n\times n$-wall} is the graph obtained from the graph on the vertex set~${[2n] \times [n]}$ whose edge set is 
\begin{align*}
    &\left\{ (i,j) (i+1,j) \, \colon \, i \in [2n-1],\, j \in [n] \right\} \\
    &\cup 
    \left\{ (i,j) (i,j+1) \, \colon \, i \in [2n],\, j \in [n-1],\, i+j \textnormal{ is odd} \right\}
\end{align*}
by deleting both degree-$1$ vertices. 
See Figure~\ref{fig:wall} for an illustration.
For a subdivision $G$ of a wall~$W$, a vertex of $G$ contained in $W$ is called a \emph{branching vertex} of~$G$.  For ${j \in [n]}$, the \emph{$n$-th row} of the $n\times n$-wall $W$ is the path~${W \big[ \big\{ (i, j)\in V(W) \, \colon \, i \in [2n] \big\} \big]}$. 
For ${i \in [n]}$, the \emph{$i$-th column}~$C_i$ of~$W$ is the path~${W \big[ \big\{ (i',j)\in V(W) \, \colon \, i' \in \{ 2i - 1, 2i \}, \, j \in [n] \big\} \big]}$. 

A linear ordering of a finite set $S$ is a bijection $\psi:S\to [|S|].$  For a graph $G$, a linear ordering of $V(G)$ is also referred to as a linear ordering of $G$.

\subsection{Contraction sequence and reduced-\texorpdfstring{$f$}{f}}

A \emph{trigraph} is a triple $G = ( V, E, R )$ where $V$ is a finite set, and $E$ and $R$ are disjoint subsets of $\binom{V}{2}$. Elements of $V$ are \emph{vertices}, and elements of $E\cup R$ are \emph{edges}. Edges in $E$ are called \emph{black}, and edges in $R$ are called \emph{red}. 

The \emph{red graph}
$\mathcal{R}(G)$ of $G$ is $(V, \emptyset, R)$.

For a vertex $v$ of $G$, let $B_G(v)$ denote the set of vertices that are adjacent to $v$ by black edges, and let $R_G(v)$ denote the set of vertices that are adjacent to $v$ by red edges. We also write $B_G[v]:=B_G(v)\cup \{v\}$ and $R_G[v]:=R_G(v)\cup \{v\}$.
We extend basic terminologies for graphs to trigraphs. 

A subtrigraph $H$ of a trigraph $G$ is called a~\emph{red-induced subtrigraph} if $H$ is an induced subtrigraph of $G$ and all edges of $H$ are contained in $\mathcal{R}(G)$.

For a graph $G$ and a vertex partition $\mathcal{P}$ of $G$, let $G/\mathcal{P}$ be the trigraph on the vertex set $\mathcal{P}$ such that for two parts $X,Y\in \mathcal{P}$,
\begin{itemize}
    \item $X$ and $Y$ are adjacent by a black edge if $X$ is complete to $Y$ in $G$, 
    \item $X$ and $Y$ are adjacent by a red edge if $X$ and $Y$ are not homogeneous, and 
    \item $X$ and $Y$ are not adjacent if $X$ is anti-complete to $Y$ in $G$.
\end{itemize}
For a set of parts $\mathcal{Q}\subseteq\mathcal{P}$, we write $\bigcup \mathcal{Q}:=\bigcup_{Z\in \mathcal{Q}}Z$.

For a graph $G$, a sequence $\mathcal{S}=\mathcal{P}_n, \ldots, \mathcal{P}_1$ of partitions of the vertex set of $G$ is a \emph{contraction sequence} of $G$ if 
\begin{itemize}
    \item $\mathcal{P}_n=\{\{v\}:v\in V(G)\}$ and $\mathcal{P}_1=\{V(G)\}$, and
    \item each $\mathcal{P}_i$ is obtained from $\mathcal{P}_{i+1}$ by merging two parts of $\mathcal{P}_{i+1}$.
\end{itemize} 

For any graph parameter $f$, let \emph{reduced-$f$} be the graph parameter~$\reduced{f}$, where for every graph~$G$,  \emph{$\reduced{f}(G)$} is the minimum $k\in\mathbb{N}$ such that there is a contraction sequence $\mathcal{P}_n, \ldots, \mathcal{P}_1$ of $G$ where $f(\mathcal{R}({G/\mathcal{P}_i}))\leq k$ for each $i\in\{1,\dots,n\}$.

\subsection{Max-leaf and component max-leaf}

The \emph{max-leaf} of a connected graph $G$, denoted by $\maxleaf(G)$, is the minimum $k$ such that every spanning tree of $G$ has at most $k$ leaves.
Although we will not use this fact directly, it follows from a result of Zickfeld \cite[Theorem~4.7]{ZickfeldThesis} that a connected graph~$G$ has bounded max-leaf if and only if~$G$ has bounded maximum degree and a bounded number of vertices with degree different from~2,
or in other words if and only if~$G$ is a subdivision of a graph of bounded size.

The \emph{component max-leaf} of $G$, denoted by $\compmaxleaf(G)$, is the minimum $k$ such that every connected component of $G$ has max-leaf at most $k$.
We consider that the 1-vertex graph counts as a tree with one leaf, so that when~$G$ is an edgeless graph, $\compmaxleaf(G) = 1$.

Our algorithms in \cref{sec:algoapplication} use dynamic programming over the collection of all \emph{connected} induced subgraphs of a graph with small component max-leaf.
The next two properties are crucial to establish the runtime of these algorithms.
\begin{lemma}\label{lem:thenumberofredneighbor}
    Let $G$ be a graph with component max-leaf $t$, and let $H$ be a connected subgraph of~$G$. Then $|N_G(V(H))|\le t$.
\end{lemma}
\begin{proof}
    For the sake of contradiction, assume that $|N_G(V(H))| \ge t+1$. Let $T$ be a spanning tree of $H$. Let $T'$ be a subgraph obtained from $T$ by adding, for each vertex $v \in N_G(V(H))$, an edge between $v$ and its neighbor in $T$. 
    Since each vertex in $N_G(V(H))$ is attached to $T$ by a single edge, $T'$ is still a tree. Moreover, each vertex in $N_G(V(H))$ is a leaf in $T'$. Hence, $T'$ has at least $t+1$ leaves. Finally~$T'$ can be extended to a spanning tree of the connected component of $G$ containing~$H$, which must still have at least $t+1$ leaves, a contradiction.
\end{proof}

Note that when taking~$H$ to be a single vertex,
\cref{lem:thenumberofredneighbor} shows that a graph~$G$ with $\compmaxleaf(G) = t$ has maximum degree at most~$t$.

\begin{lemma}\label{lem:numberofinducedsubgraphs}
    Let $t$ be a positive integer, and let $G$ be a graph with component max-leaf $t$.
    Then the number of non-empty connected induced subgraphs of $G$ is at most $|V(G)|^{t}$.
\end{lemma}
\begin{proof}
    Firstly, when~$G$ is not connected, assuming the statement holds for each connected component $H_1,\dots,H_k$ of~$G$, the number of non-empty connected induced subgraphs of~$G$ is at most
    \[ |V(H_1)|^t + \dots + |V(H_k)|^t \le |V(G)|^t. \]
    We may thus assume that~$G$ is connected.

    Consider a connected induced subgraph~$H$ of~$G$.
    Then the set $S = N_G(V(H))$ has size at most~$t$ by \cref{lem:thenumberofredneighbor}.
    Clearly~$H$ is a~connected component of~$G-S$.
    Furthermore, for any fixed $v \in S$, $H$ must contain a neighbor of~$v$.
    By \cref{lem:thenumberofredneighbor}, the degree of~$v$ is at most~$t$, hence for this fixed choice of~$S$,
    there can be at most~$t$ valid choices of~$H$.
    As an edge case, $S$ could be empty, in which case the only non-empty connected subgraph with neighborhood exactly~$S$ is $H=G$.
    So to describe~$H$, it suffices to give~$S$ of size at most~$t$,
    and indicate which of the connected components of~$G-S$ is~$H$,
    with only~$t$ relevant choices for the second step by the previous argument.
    This immediately gives that there are at most $t \cdot |V(G)|^t$ non-empty connected induced subgraphs.

    To reach the claimed bound of $|V(G)|^t$ instead,
    denote by $S^n_t = \sum_{i=0}^t \binom{n}{i}$ the number of subsets of size at most~$t$ of $\{1,\dots,n\}$.
    The previous paragraph proved that there are at most $t S^n_t$ non-empty connected induced subgraphs in~$G$.
    When $t \ge 3$, we show the following bound, which directly implies the main statement
    \begin{equation}\label{eq:sumofbinom}
        \text{for $3 \le t \le n$,} \quad S^n_t \le \frac{n^t}{t}.
    \end{equation}
    Indeed, for $t=3$, one can directly check that $S^n_t = \frac{n^3 + 5n+6}{6} \le \frac{n^3}{3}$,
    and for $t \ge 4$, we have
    \[ S^n_t = S^n_{t-1} +\binom{n}{t} \le \frac{n^{t-1}}{t-1} + \frac{n^t}{t!} \le \frac{2\cdot n^t}{t(t-1)} \le \frac{n^t}{t}. \]

    Finally, the cases $t=1$ or $t=2$ (where \cref{eq:sumofbinom} fails) correspond respectively to edgeless graphs, and to disjoint unions of paths and cycles. It is easy to check the statement in both cases.
\end{proof}

\subsection{Bandwidth and cutwidth}

The \emph{bandwidth} of a graph $G$ is the minimum $k\in\mathbb{N}\cup \{0\}$ such that there is an ordering $v_1,\dots,v_n$ of $V(G)$ satisfying $|i-j|\leq k$ for every edge $v_iv_j\in E(G)$. 
The \emph{cutwidth} of a graph $G$ is the minimum $k\in\mathbb{N}\cup \{0\}$ such that there is an ordering $v_1,\dots,v_n$ of $V(G)$ satisfying that for each $i\in [n-1]$, there are at most $k$ edges between $\{v_j:j\in [i]\}$ and $V(G)\setminus \{v_j:j\in [i]\}$.

\subsection{Stretch-width}

We will compare reduced parameters and stretch-width in \Cref{sec:comparewidth}.

An \emph{ordered graph} is a pair $(G, \psi)$ of a graph $G$ and a linear ordering $\psi$ of $V(G)$. 
Let $(G, \psi)$ be an ordered graph, and $X \subseteq V(G)$.
The \emph{convex closure} or \emph{span} of $X$ is 
\[\conv(X) := \{v \in V(G) : \min\{\psi(w):w\in X\} \leq \psi(v) \leq \max\{\psi(w):w\in X\}\}.\] 
Two sets $X, Y \subseteq V(G)$ are \emph{in conflict}, or $X$ \emph{conflicts with} $Y$, if $\conv(X) \cap \conv(Y) \neq \emptyset$. Note that this does not imply that $X$ and $Y$ intersect.

Let $\mathcal{P}$ be a partition of $V(G)$ and $X \in \mathcal{P}$. We say that $Y \in \mathcal{P} \setminus \{X\}$ \emph{interferes} with $X$ if $Y$ conflicts with $\bigcup R_{G/\mathcal{P}}[X]$. Note that it is possible that $Y$ interferes with $X$, but not vice versa. The \emph{stretch} of the part $X \in \mathcal{P}$, denoted by $\str(X)$, is defined as the number of parts in $\mathcal{P}$ interfering with $X$. The \emph{stretch} of $\mathcal{P}$ is the maximum over all parts $Z \in \mathcal{P}$ of $\str(Z)$. The \emph{stretch-width} of the ordered graph $(G, \psi)$, denoted by $\stw(G, \psi)$, is the minimum among all contraction sequences $\mathcal{P}_n, \ldots, \mathcal{P}_1$ of $G$, of $\max_{i \in [n]}\str(\mathcal{P}_i)$. Finally the \emph{stretch-width} of $G$, denoted by $\stw(G)$, is the minimum of $\stw(G, \psi)$ taken among all linear orderings $\psi$ on $V(G)$.

\begin{lemma}\label{lem:largereddegree}
	Let $d$ be a positive integer. 
	Let $(G, \psi)$ be an ordered graph, and let $\mathcal{P}$ be a partition of $V(G)$. If the maximum red degree of $G/\mathcal{P}$ is $d$, then the stretch of $\mathcal{P}$ in $(G, \psi)$ is at least $d$.
\end{lemma}
\begin{proof}
    By the definition, every $Y \in R_{G/\mathcal{P}}(X)$ interferes with $X$ for each $X \in \mathcal{P}$. Hence, $\str(X) \ge |R_{G/\mathcal{P}}(X)|$ for all $X \in \mathcal{P}$, and therefore, \[\max\limits_{Z \in \mathcal{P}}\str(Z) \ge \max\limits_{Z \in \mathcal{P}}|R_{G/\mathcal{P}}(Z)| = d.\qedhere\]
\end{proof}

\section{Reduced parameters on subdivided walls and 3-dimensional grids}\label{sec:lowerbound}

In~\cref{subsec:lowercml}, we prove that planar graphs have unbounded reduced component max-leaf.
To prove it, we show that every weighted graph $(G,w)$ with $\reduced{\compmaxleaf}(G) \le k$ admits a balanced separator dominated by at most~$k(k+2)$ vertices of~$G$.
This implies that large subdivided walls or their line graphs have unbounded $\reduced{\compmaxleaf}$.

In~\cref{subsec:reducedbw3dimgrid}, we prove that $3$-dimensional grids have unbounded reduced bandwidth.

\subsection{Balanced separators and reduced component max-leaf}\label{subsec:lowercml}
In a weighted graph $(G,w)$, for $0 < a < 1$, say that a subset $S \subseteq V(G)$ is an $a$-balanced separator
if each connected component~$C$ in $G-S$ has weight $w(C) \le a \cdot w(G)$.
Thus a $\frac{2}{3}$-balanced separator is the usual notion of balanced separator.

  \begin{figure}[!t]
    \centering
\begin{tikzpicture}[scale=0.7,
  nd/.style={circle, draw=black, line width=2pt,
             minimum size=22pt, inner sep=0pt, fill=white}
]

\draw[rounded corners, line width=1.5pt] (0,1) rectangle (13,12);

\draw[line width=1.5pt] (0,7.5) -- (13,7.5);

\draw[line width=1.5pt]
  (13.4,11.7) -- (13.7,11.7) -- (13.7,7.8) -- (13.4,7.8);
\node[font=\large] at (14.35,9.75) {$B_i$};

\node[nd] (A) at ( 4.0, 11.0) {};   
\node[nd] (B) at ( 2.5,  9.5) {};   
\node[nd] (C) at ( 5.5,  9.5) {};   
\node[nd] (D) at ( 8.5, 11.0) {};   
\node[nd] (E) at (11.0, 11.0) {};   
\node[nd] (F) at (10.5,  9.2) {};

\node[nd] (G) at (2.5, 5.5) {};     
\node[nd] (H) at (4.5, 5.5) {};     
\node[nd] (I) at (2.5, 3.5) {};     
\node[nd, minimum size=26pt] (J) at (4.5, 3.5) {};  

\node[nd] (K) at ( 8.5, 5.5) {};    
\node[nd] (L) at (10.5, 5.5) {};

\draw[line width=2pt] (A) -- (B);
\draw[line width=2pt] (A) -- (C);
\draw[line width=2pt] (D) -- (E);
\draw[line width=2pt] (E) -- (F);

\draw[red, line width=1.5pt] (A) -- (F);   
\draw[red, line width=1.5pt] (C) -- (D);

\draw[red, line width=1.5pt] (B) -- (G);
\draw[red, line width=1.5pt] (C) -- (H);

\draw[red, line width=1.5pt] (F) -- (K);

\draw[red, line width=1.5pt] (G) -- (H);
\draw[red, line width=1.5pt] (G) -- (I);
\draw[red, line width=1.5pt] (H) -- (J);
\draw[red, line width=1.5pt] (I) -- (J);

\draw[red, line width=1.5pt] (K) -- (L);

\draw[dashed, line width=2pt]
  (1.3, 2.2) rectangle (5.8, 6.5);

\node[black] at (3.55, 1.45){fully red component};

\end{tikzpicture}
        \caption{A fully red component in \cref{lem:cml-dominated-sep}. }
        \label{fig:fullyred}
    \end{figure}
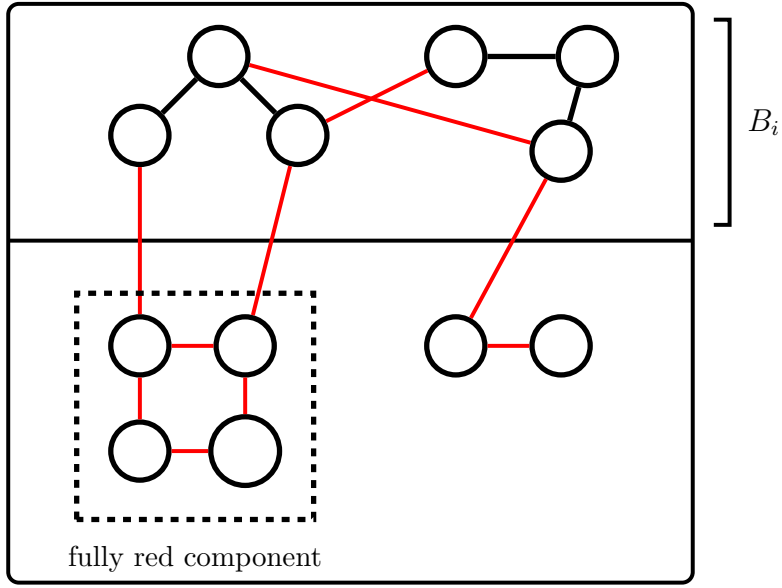

\begin{lemma}\label{lem:cml-dominated-sep}
    For any weighted graph~$(G,w)$ with $\reduced{\compmaxleaf}(G) \le k$, there is a $(1-\frac{1}{2k+3})$-balanced separator~$S$ of~$(G,w)$ such that~$S$ is contained in the open neighborhood of at most~$k$ vertices of~$G$.
\end{lemma}
\begin{proof}
    Fix a contraction sequence $\Pc_n,\dots,\Pc_1$ witnessing that $\reduced{\compmaxleaf}(G) = k$.
    For each $i \in [n]$, let $B_i \subseteq \Pc_i$ be the set of parts that have a black neighbor in~$G/\Pc_i$,
    and $R_i = \Pc_i \setminus B_i$ the parts that only have red neighbors.
    The point is that any part $P \in B_i$ is contained in $N_G(v)$ for some vertex $v \in V(G)$, namely take~$v$ to be any vertex in the black neighbor of~$P$.
    We call \emph{fully red component} a connected component of $G/\Pc_i - B_i$. See \Cref{fig:fullyred} for an illustration.

    Let $\lambda = \frac{1}{2k+3}$, and consider the first partition~$\Pc_i$ (i.e., the maximum~$i$) such that there is
    \begin{enumerate}
        \item \label{case:big-part} either a part $P \in B_i$ with weight $w(P) \ge \lambda w(G)$,
        \item \label{case:big-fully-red-component} or a fully red component~$C$ in $G/\Pc_i - B_i$ with $w(\bigcup V(C)) \ge \lambda w(G)$.
    \end{enumerate}
    In Case~\labelcref{case:big-part}, $P$ itself trivially is a $(1-\lambda)$-balanced separator since $w(V(G) \setminus P) \le (1-\lambda)w(G)$,
    and $P$~is in the neighborhood of a single vertex since it is in~$B_i$.
    We can thus focus on Case~\labelcref{case:big-fully-red-component}.

    We first bound the weight of the fully red component~$C$.
    \begin{claim}
        We have $\lambda w(G) \le w(\bigcup V(C)) \le (1-\lambda) w(G)$.
    \end{claim}
    \begin{clproof}
        The first inequality is just the assumption on~$C$.

        To prove the second one, we need to also consider the previous step~$\Pc_{i+1}$ in the contraction sequence.
        Say that~$\Pc_i$ was obtained by merging $P_1,P_2 \in \Pc_{i+1}$ into $P := P_1 \cup P_2 \in \Pc_i$.
        We are interested in the new fully red parts in~$\Pc_i$, meaning $R_i \setminus R_{i+1}$.
        Any black edge in $G/\Pc_{i+1}$ that was not incident to~$P_1$ or~$P_2$ remains in $G/\Pc_i$.
        It follows that any part in $R_i \setminus R_{i+1}$ is either~$P$ itself, or adjacent (by a red edge) to~$P$ in~$G/\Pc_i$.
        In particular, $|R_i \setminus R_{i+1}| \le k+1$.

        Now define $X = V(C) \setminus R_{i+1} \subseteq R_i \setminus R_{i+1}$ the set of new red parts used in~$C$.
        While~$P$ is not necessarily in~$X$,
        the previous argument shows that $X \cup \{P\}$ is connected in $\redG(G/\Pc_i)$.
        Thus by \cref{lem:thenumberofredneighbor}, $X \cup \{P\}$ has at most~$k$ neighbors in~$\redG(G/\Pc_i)$.
        Let~$C_1,\dots,C_\ell$ be the connected components of $C - X$.
        Since~$C$ is connected in $\redG(G/\Pc_i)$, $X$ has a neighbor in each of~$C_1,\dots,C_\ell$,
        and thus~$\ell \le k$.
        Finally, each~$C_j$ for $j \in [\ell]$ is contained in $C \setminus X \subseteq R_{i+1}$.
        Thus~$C_j$ was already a connected subset of~$R_{i+1}$, hence is contained in a fully red component of~$G/\Pc_{i+1}$.

        To summarize, $C$ consists of (1) the at most~$k+1$ parts of~$X$, and (2) at most~$k$ components~$C_1,\dots,C_\ell$,
        each of which is contained in a fully red component of~$G/\Pc_{i+1}$.
        The choice of~$\Pc_i$ implies that each part of~$\Pc_{i+1}$, and each fully red component of~$G/\Pc_{i+1}$ has weight at most~$\lambda w(G)$.
        Thus each of $C_1,\dots,C_\ell$ has weight at most~$\lambda w(G)$, and so does each part of~$X$,
        except possibly~$P$ (the only part in $\Pc_i \setminus \Pc_{i+1}$).
        The latter is the union of two parts of~$\Pc_{i+1}$, hence $w(P) \le 2\lambda w(G)$.
        Summing up, we obtain $w(\bigcup V(C)) \le (2k+2)\lambda w(G) = (1-\lambda) w(G)$.
    \end{clproof}

    This claim directly implies that~$N_G(\bigcup V(C))$ is a $(1-\lambda)$-balanced separator in~$(G,w)$.
    To conclude, it suffices to prove that it has a small dominator:
    \begin{claim}
        The neighborhood $N_G(\bigcup V(C))$ is dominated by at most~$k$ vertices in~$G$.
    \end{claim}
    \begin{clproof}
        In~$G/\Pc_i$, since~$C$ is a fully red component, the only neighbors of~$C$ are red neighbors.
        Since~$C$ is connected, \cref{lem:thenumberofredneighbor} gives that $N_{G/\Pc_i}(V(C))$ consists of at most~$k$ parts of~$\Pc_i$.
        Each of these~$k$ parts is in~$B_i$, hence is dominated by a single vertex.
        Thus the neighborhood of~$\bigcup V(C)$ is dominated by at most~$k$ vertices.
    \end{clproof}
    Thus~$N_G(\bigcup V(C))$ is a $(1-\lambda)$-balanced separator dominated by~$k$ vertices, as desired.
\end{proof}

While \cref{lem:cml-dominated-sep} only produces a~`slightly balanced' separator, we can apply this statement multiple times to obtain a~$\frac{2}{3}$-balanced separator:
\begin{theorem}\label{thm:cml-dominated-balanced-sep}
    For any weighted graph $(G,w)$ with $\reduced{\compmaxleaf}(G) \le k$, there is a $\frac{2}{3}$-balanced separator~$S$ of $(G,w)$ such that~$S$ is contained in the open neighborhood of at most~$k(k+2)$ vertices of~$G$.
\end{theorem}
\begin{proof}
    Define $a := 1 - \frac{1}{2k+3}$.
    We construct a sequence of balanced separators $S_1,\dots,S_\ell$.
    Suppose they have already been constructed up to~$S_i$ (possibly $i=0$).
    Let~$C_i$ be the connected component of $G - \bigcup_{j \le i} S_j$ of maximum weight.
    Clearly any connected component but~$C_i$ has weight at most $\frac{1}{2}W$.
    If $w(C_i) \le \frac{2}{3}W$, then $\bigcup_{j \le i} S_j$ is a $\frac{2}{3}$-balanced separator, and we stop the process here.
    Otherwise, let~$S_{i+1}$ be an $a$-balanced separator in $C_i$ with the restriction of~$w$ as weight
    (note that~$C_i$ also has reduced component max-leaf at most~$k$), and continue the process.

    For $i < \ell$, since~$S_{i+1}$ is $a$-balanced in~$C_i$, we have $w(C_{i+1}) \le a \cdot w(C_i)$.
    On the other hand, we have $w(C_i) > \frac{2}{3}W$ (otherwise we would stop with $\ell=i$),
    hence 
    \[ w(C_{i+1}) \le w(C_i) - \frac{2}{3}(1-a)W = w(C_i) - \frac{2}{3(2k+3)}W. \]
    It follows that for any $i \le \ell$, we have
    \[ w(C_i) \le W \left(1 - \frac{2i}{3(2k+3)} \right).\]
    For $i = k+2$, the right hand side is less than $\frac{2}{3}W$, hence the process must stop with $\ell \le k+2$.
    Thus we obtain that $\bigcup_{i \le \ell} S_i$ is a $\frac{2}{3}$-balanced separator, contained in the open neighborhood of at most $k\ell \le k(k+2)$ vertices.
\end{proof}

If the maximum degree is bounded, then \cref{thm:cml-dominated-balanced-sep} actually provides balanced separators of bounded size,
proving that treewidth and reduced component max-leaf are functionally equivalent on bounded degree graphs:
\begin{corollary}\label{cor:cml-degree-tw}
    Any graph~$G$ with maximum degree~$\Delta$ and $\reduced{\compmaxleaf}(G) = k$ has treewidth at most $4\Delta k(k+2)$.
\end{corollary}
\begin{proof}
    By \cref{thm:cml-dominated-balanced-sep}, for any weight function, $G$ has a balanced separator~$S$ contained in the open neighborhood of at most~$k(k+2)$ vertices.
    The degree bound then gives $|S| \le \Delta k(k+2)$.
    Thus~$G$ has $\frac{2}{3}$-balanced separators of size $\Delta k(k+2)$ for any weight function.
    By a classical result of Robertson and Seymour~\cite{graph-minors2}, this implies that $G$ has treewidth at most $4\Delta k(k+2)$.
\end{proof}

It is known that planar graphs of arbitrarily large treewidth admit arbitrarily large subdivided walls or their line graphs as induced subgraphs, see~\cite[Theorem 1.1]{AboulkerAKST21} (where wall subdivisions are simply called walls). By~\cref{cor:cml-degree-tw}, large subdivided walls or their line graphs have unbounded $\reduced{\compmaxleaf}$.
Hence \cref{thm:in-planar-graphs} follows.

\Cref{cor:cml-degree-tw} implies that strict subdivisions of $K_n$ have unbounded $\reduced{\compmaxleaf}$,
since they contain strict subdivisions of the $\sqrt{n/2} \times \sqrt{n/2}$-wall as induced subgraph.
We give here a~simpler proof with a~better lower bound of~$n-4$.

\begin{lemma}
    \label{lem:cml-subdivided-clique}
    Let~$G$ be a strict subdivision of~$K_n$.
    Then $\reduced{\compmaxleaf}(G) \ge n-4$.
\end{lemma}
\begin{proof}
    In a given contraction sequence, let~$\Pc$ be the first partition in which some branching vertex~$u$ is contained in some part $P \in \Pc$ of size at least~$2$.
    The part~$P$ may contain at most one other branching vertex, which we ignore.
    Thus the remaining branching vertices~$v_1,\dots,v_{n-2}$ satisfy $\{v_i\} \in \Pc$.
    
    Let~$S_i$ be the path from~$u$ to~$v_i$ resulting from the subdivision of~$uv_i$, with vertices
    \[ u = x_{i,0}, x_{i,1}, \dots, x_{i,\ell_i} = v_i. \]
    Furthermore call~$P_{i,k}$ the part of~$\Pc$ containing~$x_{i,k}$.
    Note that the parts $P = P_{i,0}$ and~$P_{i,1}$ are either equal or connected by a red edge,
    except in the very special case $P = \{u, x_{i,2}\}$.
    This special case can only happen for a single choice of~$i$ among~$1,\dots,n-2$, say $i=n-2$.
    We ignore it and proceed with $i \in \{1,\dots,n-3\}$.
    
    \begin{claim}\label{clm:big-makes-red}
        If~$P_{i,k}$ and~$P_{i,k+1}$ are neither the same part nor connected by a red edge,
        then~$P_{i,k}$ does not contain any vertex of $S_j \setminus \{u\}$ for $j \neq i$.
    \end{claim}
    \begin{clproof}
        The vertex~$x_{i,k+1}$ has no neighbor in~$S_j \setminus \{u\}$.
    \end{clproof}
    Pick~$k$ minimal such that~$P_{i,k}$ and~$P_{i,k+1}$ are neither the same part nor connected by a red edge, if such a~$k$ exists, and choose~$Q_i$ to be~$P_{i,k}$.
    Note that~$k$ cannot be~0, and thus~$Q_i$ contains a~vertex of~$S_i \setminus \{u\}$.
    When no $k$ as above exists, we pick $Q_i = \{v_i\}$.
    Either way, the following conditions hold:
    \begin{enumerate}
        \item the part~$Q_i$ contains a vertex of $S_i \setminus \{u\}$, and is disjoint from $S_j \setminus \{u\}$ for $j \neq i$, and
        \item in~$\redG(G/\Pc)$, $P$ and~$Q_i$ are connected by a path using only vertices among $P_{i,1},\dots,P_{i,\ell_i}$.
    \end{enumerate}
    In particular, the first condition gives that the parts~$Q_i$ are pairwise distinct.
    Also, for any~$k>0$, the part~$P_{i,k}$ contains a vertex of~$S_i \setminus \{u\}$.
    Thus in~$\redG(G/\Pc)$, $P$ and~$Q_i$ are connected by a path whose internal vertices are not~$Q_j$ for any~$j \neq i$.
    It may be that $P = Q_i$ for some~$i$, say $i=n-3$, which we ignore.
    For the remaining $i \in \{1,\dots,n-4\}$, we have that~$P$ and~$Q_i$ are connected in~$\redG(G/\Pc)$ by a~path disjoint from the other~$Q_j$, witnessing that the connected component of~$\redG(G/\Pc)$ containing~$P$ has max-leaf at least~$n-4$.
\end{proof}

\subsection{Reduced bandwidth of 3-dimensional grids}\label{subsec:reducedbw3dimgrid}

In this section, we relate the expansion of graphs to the structure of the red components in their contraction sequences.
For any three functions $\tau_{\min}, \tau_{\max}, f$, we say that an $n$-vertex graph $G$ is a~\emph{$(\tau_{\min}, \tau_{\max}, f)$-expander} if
\begin{equation*}
    \text{for every vertex subset $U \subseteq V(G)$ with $\tau_{\min}(n) \leqslant |U| \leqslant \tau_{\max}(n)$, we have $|N_G(U)| \geqslant f(|U|)$.}
\end{equation*}
Note that we do not require the maximum degree of $G$ to be bounded.
We may implicitly use $n$ as the argument of $\tau_{\min}, \tau_{\max}, f$.

We start with a~simple lemma on the growth of integer sequences.
\begin{lemma}\label{lem:growth_fct}
 Consider $0 < a < 1$ and $\lambda > 0$, and $(n_k)_{k \in \mathbb N}$ a sequence of positive reals such that for any $k$, $n_{k+1} \ge n_k + \lambda \cdot n_k^a$.
Then $n_k \ge c \cdot (n_0^{1-a} + k)^{\frac{1}{1-a}}$ for some constant $0 < c := c(a, \lambda) < 1$, provided that $n_0$ is sufficiently large.
\end{lemma}

\begin{proof}
Let $m := n_0^{1-a}$.
We prove the result by induction.
For $k = 0$, the result is clear.
Assume the formula holds for some $k \ge 0$.
Note that it is sufficient to prove that 
\[
c \cdot(m + k + 1)^{\frac{1}{1-a}} \le c \cdot (m + k)^{\frac{1}{1-a}} + \lambda \cdot c^a \cdot (m+k)^{\frac{a}{1-a}}.
\]
Dividing both sides by $c \cdot (m + k)^{\frac{1}{1-a}}$ and taking the logarithm, one obtains the following equivalent inequality.
\[
\frac{1}{1-a} \log\left(1+\frac{1}{m+k}\right) \leqslant \log\left(1 + \frac{\lambda c^{a-1}}{m+k}\right).
\]
And since $\log(1 + x) = x + o(x)$ around 0, we indeed have $\log(1 + \frac{\lambda c^{a-1}}{m+k}) \ge \frac{1}{1-a} \log(1 + \frac{1}{m+k})$ when $c^{a-1}$ is sufficiently large, and so when $c$ is sufficiently small.
\end{proof}

The next lemma has a~rather technical statement, mainly because we want to use it in various settings.
However, it can be simply summarized as follows: if for any $U \subseteq V(G)$ we have $|N_G(U)| \geqslant |U|^a$ for some $0 < a < 1$, then one can find, in any contraction sequence of bounded red degree, an arbitrary large red subgraph $H$ such that $|V(H)| \geqslant \diam(H)^{\frac{a}{1-a}}$.

\begin{lemma}\label{lem:from-growth-2-red-diameter}
Let $f$ be such that $f(n) = \Omega(n^a)$ for some $0 < a < 1$, and $\tau_{\min}$ and $\tau_{\max}$ satisfy $\tau_{\max}(n)^{1-a} \ge (\log n + \tau_{\min}(n)) \log n$.
Let $r$ and $d$ be integers, and let $G$ be a $(\tau_{\min}, \tau_{\max}, f)$-expander without $K_{r, r}$ subgraph.
If $\mathcal{P}_n, \ldots, \mathcal{P}_1$ is a contraction sequence of $G$ of twin-width at most $d$,
then at least one red graph $\mathcal{R}(G/\mathcal{P}_i)$ admits a~connected subgraph $H$ with $|V(H)| = \Omega(\diam(H)^{\frac{a}{1-a}})$ and $|V(H)| = \Omega((\log n)^{a})$.
\end{lemma}

\begin{proof}
Let $t := \lceil\max(\log n, \tau_{\min}(n))\rceil$ where $n$ is the number of vertices of~$G$, which we assume large enough.
Consider a contraction sequence $\mathcal{P}_n, \ldots, \mathcal{P}_1$ of $G$ of witnessing twin-width at~most~$d$.
Let $s$ be the largest index such that $\mathcal P_s$ contains a~part $P$ that itself contains at least $t$ vertices of $G$.
In particular, any part of $\mathcal{P}_{s+1}$ contains less than $t$ vertices, so any part of $\mathcal{P}_s$ contains at most $2t$ vertices of $G$. Let $G_s:=G/\mathcal{P}_s$.

Our goal is to show that the red component of $G_s$ containing $P$ locally has \emph{large growth} around~$P$, that is, the sizes of the red neighborhoods at a~small distance from $P$ are large.
This fact will be sufficient to exhibit the subgraph $H$ from $G_s$.

We define a sequence  $Y_0, Y_1, \ldots $ of vertex sets of $G_s$ as follows:
\begin{itemize}
    \item $Y_0 := \{P\}$ where $P$ is the part of $\mathcal{P}_s$ with $|P| \geqslant t$.
    \item  $Y_{k+1} := \{ Q \in N_{G_s}(Y_k) : |Q| \ge r\} \cup Y_k$.
\end{itemize}
That is, $Y_{k+1}$ contains all parts of size at least $r$ that are in the neighborhood of $Y_k$ in $G_s$, plus~$Y_k$.
In particular, all the parts in $Y_k$ have at least $r$ vertices.
This fact implies that the trigraph $G_s[Y_k]$ contains only red edges: indeed, a~black edge between two vertices of $Y_k$ would imply the existence of a~biclique subgraph in~$G$ between the two corresponding parts, hence a~$K_{r, r}$-subgraph.
Note that the radius of $G_s[Y_k]$ is at~most~$k$.

Our goal is now to analyze the growth of $|Y_k|$.
Let $g_k := |Y_k|$ and $n_k := |\bigcup_{Q\in Y_k}Q|$.
Recall that, by the maximality of $s$, every part in $\mathcal{P}_s$ has at most $2t$ vertices.
Thus,
\begin{equation}\label{eq:ineq_nk_gk}
    n_k \leqslant 2t g_k.
\end{equation}

Consider the set $A_k := \{ Q \in N_{G_s}(Y_k) : |Q| < r\}$.
We clearly have that $|\bigcup_{Q\in A_k}Q| < r \cdot |A_k|$.
Furthermore, the parts in $A_k$ can be split into two sets: the ones that are red neighbors of a~vertex of $Y_k$ in $G_s$, which we denote by $A_k^{\text{red}}$ and the others, which we denote by $A_k^{\text{black}}$.

Since the maximum red degree of $G_s$ is upper bounded by $d$, each vertex of $Y_k$ is adjacent by a red edge to at most $d$ vertices of $A_k^{\text{red}}$, which implies that $|A_k^{\text{red}}| \le d \cdot |Y_k|$.
In addition, the black neighborhood of a vertex of $Y_k$ in $G_s$ contains less than $r$ vertices, since otherwise $G$ would contain a $K_{r, r}$ subgraph. Thus $|A_k^{\text{black}}| \le r \cdot |Y_k|$.
Combining these bounds, one obtains the following inequality
\[
\left|\bigcup_{Q\in A_k}Q\right| \le (r+d)r|Y_k| \le r(r+d)g_k.
\]
This directly implies that
\[
n_{k+1} \ge n_k + \left|\bigcup_{Q\in N_{G_s}(Y_k)}Q \right| - r(r+d)g_k
\ge n_k + \left|N_G\left(\bigcup_{Q\in Y_k }Q\right)\right| - r(r+d)g_k.
\]
Our goal is now to use the expansion properties of $G$ to lower bound $n_k$, which will translate into large lower bounds for the integers $g_k$ by \Cref{eq:ineq_nk_gk}.

Since $G$ is a $(\tau_{\min}, \tau_{\max}, f)$-expander,
for any $k$ such that $\tau_{\min}(n) \leq n_k \leq \tau_{\max}(n)$,
we have $\left|N_G\left(\bigcup_{Q\in Y_k }Q\right)\right| \ge f(n_k)$.
This implies that $n_{k+1} \geqslant n_k + f(n_k) - r(r+d)g_k$ as long as
\begin{equation}\label{eq:nk_bound}
    n_k \leqslant \tau_{\max}(n).
\end{equation}
In addition, $f$ is at least $x \mapsto 2 \lambda x^a$ for some $\lambda > 0$.
Hence, as long as 
\begin{equation}\label{eq:f-bound}
     r(r+d)g_k \leqslant \lambda n_k^a 
\end{equation}
holds, we have 
\begin{equation}\label{eq:nk_growth}
    n_{k+1} \geqslant n_k + \lambda \cdot n_k^a.
\end{equation}

Consider $k_{\max}$ the first integer $k$ such that either \Cref{eq:nk_bound} or \Cref{eq:f-bound} does not hold (if such an integer does not exist, we set $k_{\max} = \infty$).
Let $m := t^{1-a}$.
By \Cref{lem:growth_fct}, there exists a~constant $c$ such that $n_k \ge c \cdot (m + k)^{\frac{1}{1-a}}$ for any $k \le k_{\max}$.
Note that, since $k \mapsto c \cdot (m+k)^{\frac{1}{1-a}}$ is unbounded, $k_{\max}$ is actually always finite.
Hence, either \Cref{eq:nk_bound} or \Cref{eq:f-bound} fails for~$k_{\max}$.

We claim that for any~$k$, \Cref{eq:f-bound} implies \Cref{eq:nk_bound},
meaning that it must in fact be \cref{eq:f-bound} which fails for~$k_{\max}$.
Indeed, if we have $r(r+d)g_k \leqslant \lambda n_k^a$,
then \cref{eq:ineq_nk_gk} gives $\lambda n_k^a \ge \frac{r(r+d)}{2t}n_k$, or equivalently
\[ \frac{2\lambda}{r(r+d)} t \ge n_k^{1-a}. \]
Note that $\frac{2\lambda}{r(r+d)}$ is a smaller than~1, so in particular $t \ge n_k^{1-a}$.
Further, by assumption on~$\tau_{\max}$ and choice of~$t$, we have $\tau_{\max}(n)^{1-a} \ge t$.
This gives $\tau_{\max}(n)^{1-a} \ge n_k^{1-a}$, implying \cref{eq:nk_bound}.

Hence we know that \Cref{eq:f-bound} does not hold for $k_{\max}$, which gives
\[
g_{k_{\max}} > \frac{\lambda c^a(m + k_{\max})^{\frac{a}{1-a}}}{r(r+d)}.
\]
In particular $g_{k_{\max}} \geqslant \frac{\lambda c^a}{r(r+d)} m^{\frac{a}{1-a}}$, so $|Y_{k_{\max}}| = \Omega(m^{\frac{a}{1-a}}) = \Omega(t^{a}) = \Omega((\log n)^{a})$ and $|Y_{k_{\max}}| = \Omega(k_{\max}^{\frac{a}{1-a}}) = \Omega( \diam(G_s[Y_{k_{\max}}])^{\frac{a}{1-a}})$.
\end{proof}

The following is a corollary of the isoperimetric inequality of Loomis and Whitney \cite{Loomis-Whitney}.
\begin{lemma}\label{lem:3d-grids-expand}
The 3-dimensional grid is a~$(1, x^{1/3} - 1, x^{2/3})$-expander.
That is, in the $n \times n \times n$-grid, any set of vertices $P$ of size less than
$n$ satisfies $|N(P)| \ge |P|^{2/3}$.
\end{lemma}
\begin{proof}
Let $G$ be the $n \times n \times n$ 3-dimensional grid.
For a vertex $v = (v_x, v_y, v_z) \in V(G)$, we define the $z$-line of $v$ as $\{(v_x, v_y, i)~:~i \in [n]\}$.
Consider a set $P \subseteq V(G)$ of at most $n - 1$ vertices.
We define the three projections $XY(P)$, $XZ(P)$ and $YZ(P)$ of $P$, where $XY(P)$ is the set of all $(x, y)$ values such that $P$ contains a point in $\{(x, y, i)~:~i \in [n]\}$, and $XZ(P)$ and $YZ(P)$ are defined symmetrically.

By the Loomis--Whitney inequality \cite{Loomis-Whitney}, $|P| \le \sqrt{|XY(P)| \cdot |XZ(P)| \cdot |YZ(P)|}$. Hence, one of $XY(P)$, $XZ(P)$ or $YZ(P)$ has size at least $|P|^{2/3}$.
Assume up to symmetry that it is $XY(P)$.
For each $(x, y) \in XY(P)$ the $z$-line of $(x, y, 0)$ contains both a vertex of $P$ and a vertex of $V(G)\setminus P$ since $k < n$.
In particular, it contains a vertex of $N_G(P)$.
Since all the $z$-lines of points $(x, y, 0)$ are disjoint, we obtain $|N_G(P)| \ge |XY(P)| \ge k^{2/3}$.
\end{proof}

\begin{theorem}\label{thm:3d-grids-rbandw}
The reduced bandwidth of 3-dimensional grids is unbounded.
\end{theorem}
\begin{proof}
  By~\cref{lem:3d-grids-expand}, the 3-dimensional grid is a~$(1, \tau_{\max}, f)$-expander of bounded maximum degree for $f(x) = x^{2/3}$ and $\tau_{\max}(x) = x^{1/3}-1$.
  Thus \Cref{lem:from-growth-2-red-diameter} yields connected subgraphs $H$ in at least one red graph within any contraction sequence of the 3-dimensional grids (of increasing sizes) such that \[|V(H)| = \Omega\left(\diam(H)^{\frac{2/3}{1-2/3}}\right) = \Omega\left(\diam(H)^2\right).\]
  However, graphs of bounded bandwidth have linear diameter in their number of vertices.
\end{proof}

\section{Weakly sparse collapse of bounded reduced parameters}\label{sec:ws}

Within classes of bounded twin-width, it was shown that weak sparsity and bounded degeneracy are equivalent~\cite{twin-width2}.
Furthermore, the treewidth of bounded-degree $n$-vertex graphs of bounded twin-width can be as large as $\Theta(n)$, as there are classes of cubic expanders with bounded twin-width~\cite{twin-width2}.
In contrast, it was proven in~\cite{reduced-bdw} that no weakly sparse class of bounded reduced bandwidth can be an expander class.

In~\cref{sec:ws-rcutw}, we will lift this result to reduced cutwidth by showing that every weakly sparse class of bounded reduced cutwidth has polynomial expansion---as defined by Ne{\v s}et{\v r}il and Ossona de Mendez~\cite{Sparsity}.
It was shown~\cite{DvorakN16} that a subgraph-closed class has polynomial expansion if and only if it has \emph{strongly sublinear separators}, i.e., all its $n$-vertex subgraphs have treewidth $O(n^{1-\varepsilon})$ for some $\varepsilon > 0$ depending on the class only.
We first prove a~more quantitative analogous statement for bounded reduced component max-leaf.

\subsection{Weakly sparse classes of bounded reduced component max-leaf}\label{sec:ws-rcml}

We will rely on this iterative construction of balanced separators.

\begin{lemma}\label{lem:protrusion}
  Let $G$ be an $n$-vertex graph and $\alpha \in (0,1)$ be a~real (allowed to depend on~$n$) such that for every (induced) subgraph $H$ of~$G$ there is a~separator $S \subseteq V(H)$ and a~union $Y$ of connected components of~$H-S$ of size at~most~$\frac{2}{3} n$ with $|S|/|V(Y)| \leqslant \alpha$.
  Then $G$ has a~balanced separator of size at~most~$\alpha n$.
\end{lemma}

\begin{proof}
  We define a~sequence of separators and induced subgraphs of~$G$ as follows.
  Set $G_1 := G$ and let $S_1 \subseteq V(G_1), Y_1$ satisfy the precondition of the lemma.
  Set $i \leftarrow 1$.
  While $G_i - (S_i \cup V(Y_i))$ has more than $2n/3$ vertices, set $G_{i+1} := G_i - (S_i \cup V(Y_i))$, let $S_{i+1}, Y_{i+1}$ satisfy the assumption of the lemma on $G_{i+1}$, and increment $i$ by~1.

  Let $q$ be the index of the last defined pair $G_q, S_q$.
  By construction, $\bigcup_{i \in [q]} S_i$ is a~balanced separator of~$G$.
  Note that the family $\{V(Y_i)\}_{i \in [q]}$ has pairwise disjoint sets, thus $\bigcup_{i \in [q]} |V(Y_i)| \leqslant n$.
  Therefore $\sum_{i \in [q]} |S_i| \leqslant \bigcup_{i \in [q]} \alpha |V(Y_i)| \leqslant \alpha n$.
\end{proof}

We first show a~general statement where the $\reduced{\compmaxleaf}$ upper bound and the size of the excluded biclique subgraph may be superconstant. 

\begin{lemma}\label{lem:ws-rcml}
  Let $0 < \varepsilon, \beta \leqslant \delta < \gamma < 1$ be reals, and $n, k, t$ be integers such that $k, t \leqslant n^\beta$, and $n$ is sufficiently large.
  Let $G$ be an $n$-vertex graph with $\reduced{\compmaxleaf}(G) \leqslant k$ and no $K_{t,t}$ subgraph.
  Then $G$ has a~balanced separator of size at~most~$n^z$ with $z := \max(1+\beta+\delta-\gamma+\varepsilon,~1-\delta+\beta+\varepsilon,~\gamma+\beta+\varepsilon)$.
\end{lemma}

\begin{proof}
  We take $n$ sufficiently large so that the following inequalities hold:
  \begin{itemize}
  \item $n^\gamma - n^\delta > 0.5 n^\gamma$, and
  \item $n^\varepsilon > 56$.
 \end{itemize}

In a~fixed contraction sequence witnessing $\reduced{\compmaxleaf}(G) \leqslant k$, let $\mathcal P$ be the first partition with a~part $P$ of size at~least~$n^\gamma$.

We build a~separator $S$ much smaller than at~least~one connected component of $G-S$.
This can be used iteratively to exhibit a~strongly sublinear separator of $G$, and of all its induced subgraphs (thus of all its subgraphs).

Let $\mathcal F$ be the set of parts $P' \in \Pc$ of size at~most~$n^\delta$ such that there is a~$(P, P')$-path in $\redG(G/\Pc)$ where every part but $P'$ has size larger than $n^\delta$.
It follows from \cref{lem:thenumberofredneighbor} that $|\mathcal F| \leqslant k \leqslant n^\beta$, and thus $|\bigcup \mathcal F| \leqslant n^{\beta+\delta}$.
Let $C$ be the connected component of $\redG(G/\Pc)-\mathcal F$ containing~$P$.
Each part in~$V(C)$ has size more than $n^\delta$, so for $q := |V(C)|$, we have $q < n^{1-\delta}$.

For any part $P' \in V(C)$, since $|P'| > n^\delta \geqslant n^\beta \geqslant t$,
the union $B(P')$ of the parts that are black neighbors of $P'$ in~$G/\mathcal P$ is of size at~most~$t-1$.
We set $X := \bigcup \mathcal F \cup \bigcup_{P' \in V(C)} B(P')$.
Let $D$ be the subgraph of $G$ induced by the union of the parts in $C$.
Note that $D$ is the union of a set of connected components of $G-X$.

We have $|X| \leqslant n^{\beta+\delta}+q(t-1)$ and $|V(D)| \geqslant (q-1)n^\delta+n^\gamma$.

Hence,
\[\frac{|X|}{|V(D)|} \leqslant \frac{n^{\beta+\delta} + q(t-1)}{n^\gamma-n^\delta+q n^\delta} \leqslant \max\left(\frac{n^{\beta+\delta}}{n^\gamma-n^\delta},\frac{t-1}{n^\delta}\right) \leqslant
\max(2 n^{\beta+\delta-\gamma}, n^{\beta-\delta}).\]
If $|V(D)|<\frac{2}{3} n$, then $X, V(D)$ satisfy the condition of~\cref{lem:protrusion} on~$G$.
We recurse on~$G-(X \cup V(D))$ and invoke~\cref{lem:protrusion} if ``$|V(D)|<\frac{2}{3} n$'' keeps happening.
We would then get a~balanced separator of size smaller than~$\frac{1}{2} n^{\max(1+\beta+\delta-\gamma+\varepsilon,~1-\delta+\beta+\varepsilon)}$ since $n^\varepsilon > 4$.

We now assume that $|V(D)| \geqslant \frac{2}{3} n$.
We will now find a~balanced separator of~$G$ of size at~most~$\frac{1}{2} n^z$ and conclude.
Inside~$C$, consider an inclusion-wise maximal connected subgraph~$A$ with $\left|\bigcup V(A)\right| \le \frac{1}{2}n$ and $P \in V(A)$.
Note that the choice of the partition~$\Pc$ ensures that every part in~$\Pc$ has size less than~$n^\gamma$, except for~$P$ which has size less than~$2n^{\gamma}$.
Since~$C$ is connected and $|\bigcup V(C)| \ge |V(D)| \ge \frac{2}{3}n$,
this inclusion-wise maximal subgraph~$A$ must satisfy $\frac{1}{2}n - n^{\gamma} < \left|\bigcup V(A)\right| \le \frac{1}{2}n$.
In particular, $\bigcup V(A)$ has size between $\frac{1}{3}n$ and $\frac{2}{3}n$.

Define now $S := \bigcup N_{\redG(G/\Pc)}(V(A)) \cup \bigcup_{P' \in V(A)} B(P')$, which is exactly the union of all parts that are neighbors (red or black) of~$A$ in~$G/\Pc$.
Thus $S$ separates~$\bigcup V(A)$ from the rest of the graph.
Since~$A$ has size between $\frac{1}{3}n$ and $\frac{2}{3}n$, $S$ is a balanced separator.

By \cref{lem:thenumberofredneighbor}, $A$ has at most~$k$ neighbors in~$\redG(G/\Pc)$.
On the other hand, since each part of $V(A) \subseteq V(C)$ has size more than~$n^\delta$, we have $|V(A)| < n^{1-\delta}$.
Using that every part in~$\Pc$ except~$P$ has size at most~$n^\gamma$, and that $|B(P')| < t$ for any $P' \in V(A) \subseteq V(C)$, we obtain
\[ |S| \le k \cdot n^\gamma + n^{1-\delta}t \le n^{\beta+\gamma} + n^{1-\delta+\beta}
    \le \frac{1}{2} n^{\max(\gamma+\beta+\varepsilon,~1-\delta+\beta+\varepsilon)}. \]

Thus we obtain a~balanced separator of~$G$ of size at~most~$\frac{1}{2} n^z$.
We add this separator to the union (whose size is also upper bounded by $\frac{1}{2} n^z$) of the separators defined when ``$|V(D)| < \frac{2}{3}n$'' and conclude.
\end{proof}

From~\cref{lem:ws-rcml} we get that weakly sparse classes of $\reduced{\compmaxleaf}$ at~most~$n^\beta$ with $\beta<1/3$ admit strongly sublinear separators. By a result of Dvo\v{r}\'{a}k and
                  Norin~\cite{DvorakN16}, such classes have polynomial expansion.

More precisely:

\begin{corollary}\label{cor:ws-rcml-1}
  For every sufficiently small $\varepsilon > 0$ and sufficiently large $n$, every $n$-vertex graph $G$ with $\reduced{\compmaxleaf}(G) \leqslant n^{\frac{1}{3}-2\varepsilon}$ and no $K_{n^{\frac{1}{3}-2\varepsilon}, n^{\frac{1}{3}-2\varepsilon}}$ subgraph admits a~balanced separator of size $n^{1-\varepsilon}$. 
\end{corollary}
\begin{proof}
  We apply~\cref{lem:ws-rcml} with $\beta=\frac{1}{3}-2\varepsilon$, $\delta=\frac{1}{3}$, $\gamma=\frac{2}{3}$.
  This gives a~balanced separator of size $n^z$ with $z = \max(1+\frac{1}{3}-2\varepsilon+\frac{1}{3}-\frac{2}{3}+\varepsilon,~1-\frac{1}{3}+\frac{1}{3}-2\varepsilon+\varepsilon,~\frac{2}{3}+\frac{1}{3}-2\varepsilon+\varepsilon)=1-\varepsilon$.
\end{proof}

Another consequence of~\cref{lem:ws-rcml} is that weakly sparse classes of bounded $\reduced{\compmaxleaf}$ have treewidth $O(n^{\frac{2}{3}+\varepsilon})$ for any $\varepsilon > 0$.

\begin{corollary}\label{cor:ws-rcml-2}
  Let $\mathcal C$ be a~weakly sparse class of bounded $\reduced{\compmaxleaf}$.
  Then for every $\varepsilon > 0$, there is a~constant $c := c_{\mathcal C, \varepsilon}$ such that every $n$-vertex graph of $\mathcal C$ has treewidth at~most~$c n^{\frac{2}{3}+\varepsilon}$. 
\end{corollary}
\begin{proof}
  Let $k, t$ be integers such that $\mathcal C$ has no $K_{t,t}$ subgraph nor graphs of $\reduced{\compmaxleaf}$ more than~$k$.
  In~\cref{lem:ws-rcml}, one can then make ``$\beta+\varepsilon$'' arbitrarily close to 0.
  Thus, for every sufficiently small~$\varepsilon > 0$, every $n$-vertex induced subgraph of a~graph of $\mathcal C$ admits a~balanced separator of size $O(n^z)$ with $z := \max(1+\delta-\gamma+\varepsilon,~1-\delta+\varepsilon,~\gamma+\varepsilon)$.
  We then choose $\delta=\frac{1}{3}$ and $\gamma=\frac{2}{3}$, so that $z = \frac{2}{3}+\varepsilon$.
  We conclude since separation number and treewidth are linearly tied~\cite{DvorakN19}.
\end{proof}

Given a~positive integer $n$, the $n \times n$ Pohoata--Davies grid~\cite{Pohoata14,Davies22} is the graph obtained from the disjoint union of $n$ paths on $n$ vertices each, by adding for each $i \in [n]$, a~vertex adjacent to the $i$-th vertex of each path; see~\cref{fig:pohoata-grid}.
\begin{figure}[!t]
	\centering
	\begin{tikzpicture}[vertex/.style={draw,circle,inner sep=0.06cm}]
		\def\k{7}
		\pgfmathtruncatemacro\km{\k - 1}
		\def\sv{0.5}
		\def\sh{0.8}

		\foreach \i in {1,...,\k}{
			\foreach \j in {1,...,\k}{
				\node[vertex] (a\i\j) at (\i * \sh, \j *\sv) {};
			}
		}
		\foreach \i in {1,...,\k}{
			\node[vertex] (s\i) at (\i * \sh, \k * \sv + \sv) {};
		}

		\foreach \i [count = \ip from 2] in {1,...,\km}{
			\foreach \j in {1,...,\k}{
				\draw (a\i\j) -- (a\ip\j) ;
			}
		}

		\foreach \i in {1,...,\k}{
			\foreach \j in {1,...,\k}{
				\draw (s\i) to [bend left = 24] (a\i\j) ;
			}
		}
	\end{tikzpicture}
	\caption{The $7 \times 7$ Pohoata--Davies grid.}
        \label{fig:pohoata-grid}
\end{figure}
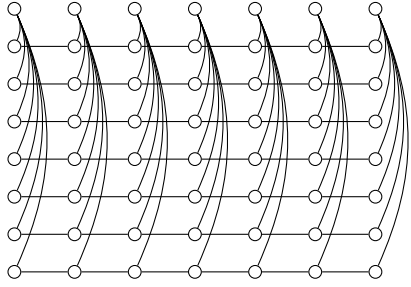
For every integer $g \geqslant 0$, let $G_{k,g}$ be the graph obtained from the $k \times k$ Pohoata--Davies by subdividing each edge of the $k$ paths $g$ times.
Note that $G_{k,g}$ has girth $2(g+1)+4$.
It is not difficult to establish that the treewidth of $G_{k,g}$ is~$k$.
Setting $n := |V(G_{k,g})|$, we have $n < k^2(1+g)$, so the treewidth of $G_{k,g}$ is larger than $\frac{1}{(1+g)^{1/2}} \sqrt n$.
The reduced component max-leaf of $G_{k,g}$ is at~most~3 due to the following contraction sequence, in which the only non-singleton connected component in each red graph is a~path, plus optionally a~pendant vertex attached to the penultimate vertex of the path.
For each $j$ going from 2 to $k$, for each $i$ going from 1 to $k+g(k-1)$, contract the $i$-th vertex of the first path with the $i$-th vertex of the $j$-th path.
Then contract (in any order) vertices outside the red path with their black neighbor on the red path.
Finally contract adjacent pairs of vertices along the red path, until it is reduced to a~single vertex.

This establishes the following fact, in contrast with~\cref{cor:ws-rcml-2}.

\begin{proposition}\label{prop:pd-grid}
  For every positive integer $g$, there is an infinite family of graphs $G$ of girth at~least~$g$ and treewidth $\Omega_g(|V(G)|^{1/2})$ with $\reduced{\compmaxleaf}(G) \leqslant 3$.
\end{proposition}

The next result, also based on the Pohoata--Davies grid, shows that there are weakly sparse classes of bounded reduced component max-leaf whose subgraph closure has unbounded reduced component max-leaf.
In contrast, although \emph{bounded twin-width} is not a~subgraph-closed property (because, for example, of cliques), it is indeed monotone within weakly sparse classes~\cite{twin-width2}.

\begin{proposition}\label{thm:non-fo-closure}
For any positive integers $g, n$, there is a~graph of girth at~least~$g$ and reduced component max-leaf at~most~3, such that $G$ has a~subgraph of reduced component max-leaf at least~$n$.
\end{proposition}
\begin{proof}
  Consider the subdivided Pohoata--Davies grids $G_{(n+4)^2,g}$.
  We already showed that \[\reduced{\compmaxleaf}(G_{(n+4)^2,g}) \leqslant 3,\] and observed that $G_{(n+4)^2,g}$ has girth $2(g+1)+4 \geqslant g$.
  It can be seen that $G_{(n+4)^2,g}$ admits a~strict subdivision of $K_{n+4}$ as a~subgraph: use $n+4$ high-degree vertices as branching vertices and reserve one distinct path of the $(n+4)^2$ paths to connect each pair of branching vertices. 
  By~\cref{lem:cml-subdivided-clique}, this subgraph has reduced component max-leaf at~least~$n$.
\end{proof}

\subsection{Weakly sparse classes of bounded reduced cutwidth}\label{sec:ws-rcutw}

We now need the original definition of \emph{polynomial expansion}.
This relies on the notion of \emph{shallow minors} introduced by Plotkin, Rao and Smith~\cite{PRS}.

We say that a~graph $H$ is a~$d$-shallow minor of~$G$ if there are pairwise disjoint connected subgraphs $B_1, \dots B_{|V(H)|}$ of $G$, called \emph{branch sets}, such that the radius of $B_i$ is at~most~$d$, and $V(B_i)$ and $V(B_j)$ are adjacent in $G$ if $v_iv_j \in E(H)$, with $V(H)=\{v_1, \ldots, v_{|V(H)|}\}$.
The family $\{B_1, \ldots, B_{|V(H)|}\}$ is then called a~\emph{minor model} (or simply \emph{model}) of~$H$ in~$G$.
This model is \emph{minimal} if there is no model $\{B'_1, \ldots, B'_{|V(H)|}\}$ of~$H$ where each $B'_i$ is a~subgraph of $B_i$, one of~which a~strict subgraph.

Note that if $H$ is a $d$-shallow minor of~$G$, then there is a~minimal model of~$H$ in~$G$ with all parts of radius at most $d$, by simply keeping, in each branch set, a~shortest-path tree from a~central vertex (i.e., one realizing the radius of the branch set).

For a~class of graphs $\mathcal C$, we say that $H$ is a $d$-shallow minor of $\mathcal C$ if it is a $d$-shallow minor of some $G \in \mathcal C$.
Plotkin, Rao and Smith showed that the absence of large cliques as shallow minors implies small(er) balanced separators.

\begin{theorem}[Plotkin, Rao and Smith~\cite{PRS}]\label{thm:PRS}
If $G$ is an $n$-vertex graph not containing $K_t$ as a~$d$-shallow minor, then $G$ has a~balanced separator of size $O(d t^2 + \frac{n \log n}{d})$.
\end{theorem}

In particular, for any constant $c>0$, if a~class $\mathcal C$ does not contain $K_{d^c}$ as $d$-shallow minors for all~$d$, \cref{thm:PRS} implies that $\mathcal C$ has strongly sublinear balanced separators, by setting $d := n^{\frac{1}{2(c+1)}}$.

A~class $\mathcal C$ has \emph{expansion (at most) $f$} if for any $H$ that is a $d$-shallow minor of $\mathcal C$, $\frac{|E(H)|}{|V(H)|} \le f(d)$, and $\mathcal C$ has \emph{polynomial expansion} if $f$ can further be taken polynomial.
A~class has \emph{$\omega$-expansion $f$} if no clique on $f(d)$ vertices is a~$d$-shallow minor of the class.
We will rely on the following result.

\begin{theorem}[Dvo\v{r}\'{a}k~\cite{Dvorak2018}]\label{thm:poly-exp-poly-clique}
A graph class has polynomial expansion if and only if it has polynomial $\omega$-expansion.
\end{theorem}

We are now ready to prove that any class of bounded reduced cutwidth has polynomial expansion.
The key idea is that classes with large expansion (or $\omega$-expansion) contain as subgraphs ``short'' subdivisions of any cubic graphs, and in particular of cubic expanders.
These graphs are not proper expanders, but are $(\tau_{\min}, \tau_{\max}, f)$-expanders for suitable functions $\tau_{\min}, \tau_{\max}$ and $f$.
Hence, by \Cref{lem:from-growth-2-red-diameter}, we know that any contraction sequence of these graphs will eventually realize red graphs with small diameter compared to their number of vertices, which is sufficient to get a~lower bound on their cutwidth.

\begin{lemma}\label{lem:Kt-to-subexp}
Let $G$ be a graph containing $K_t$ as a $d$-shallow minor.
Then $G$ contains as subgraph a~$(\le 4d)$-subdivision of any $t$-vertex cubic graph.
\end{lemma}
\begin{proof}
Let $H$ be a~cubic graph with $V(H)=[t]$. 
Since $H$ is a~subgraph of $K_t$, it is a $d$-shallow minor of $G$, and therefore, there exists a minimal model $\mathcal M = \{B_1, \ldots, B_t\}$ of~$H$ in $G$ such that the radius of each $B_i$ is at~most~$d$.
By minimality, each $B_i$ is a tree with at most three leaves.
Since $\mathcal M$ is a model of $H$, for each edge $uv \in E(H)$, there is an edge of $G$ between $B_u$ and $B_v$.
We denote this edge by $\mathcal M(uv)$ (if there are more than one such edge, choose any of them).

Let $G'$ be the subgraph of $G$ obtained by adding to $\bigcup_{i \in [t]} B_i$ the edges $\mathcal M(uv)$ for $uv \in E(H)$.
Each part $V(B_v)$ contains exactly one vertex of degree 3 in $G'$.

We denote this vertex by $c(v)$.
For each edge $uv \in E(H)$, there is a~unique path $P(uv)$ of $G'[V(B_u) \cup V(B_v)]$ going from $c(u)$ to $c(v)$.
Since the radius of any branch set is at~most~$d$, $P(uv)$ contains at~most $4d+1$ edges.
Moreover for any two edges $uv, wx \in E(H)$, the paths $P(uv)$ and $P(wx)$ are internally vertex-disjoint.
Hence the vertices $\{c(u):u\in V(H)\}$ and the paths $\{P(uv):uv \in E(H)\}$ form a~$(\le 4d)$-subdivision of~$H$.
\end{proof}

A~function $g$ from $\mathbb R_+$ to $\mathbb R_+$ is \emph{sub-homogeneous} if for every $x \geqslant 1, y > 0$, $g(xy) \leqslant x g(y)$.
Note that this implies that for every $x \leqslant 1, y > 0$, it holds that $g(xy) \geqslant x g(y)$.

\begin{lemma}\label{lem:subdiv-expander}
If $G$ is a $(\tau_{\min}, \tau_{\max}, f)$-expander of maximum degree $\Delta$, such that $\tau_{\min}, \tau_{\max}$ and $f$ are non-decreasing and sub-homogeneous, then any $(\leqslant t)$-subdivision of $G$ is a $(c \cdot \tau_{\min}, \frac{\tau_{max}}{c}, \frac{f}{c \Delta^2})$-expander with $c := \Delta t + 1$.
\end{lemma}
\begin{proof}
Let $G'$ be a $(\leqslant t)$-subdivision of~$G$. Let $n=|V(G)|$.
Note that $|V(G')| \le c n$.
We consider the partition $\mathcal{P} = \{P_u~:~u \in V(G)\}$ of $V(G')$ where each vertex $v$ of $G'$ belongs to the part $P_u$ with $u$ the vertex of $G$ with minimum distance to~$v$, breaking ties arbitrarily.
In particular, for any $u$, $|P_u| \le c$.

Let $S \subseteq V(G')$ be such that $c \cdot \tau_{\min}(|V(G')|) \leqslant |S| \leqslant \tau_{\max}(|V(G')|)/c$.
In particular, $c \cdot \tau_{\min}(n) \le |S| \le \tau_{\max}(cn)/c$.
Let $Q$ be the set of vertices $u \in V(G)$ such that $P_u$ contains a vertex of $S$.
Clearly, since $|S|/c \le |Q| \le |S|$, we have by the monotonicity and sub-homogeneity of $\tau_{\min}$ and $\tau_{\max}$ that 
\[
\tau_{\min}(n) \le 
\frac{c \cdot \tau_{\min}(n )}{c} \le
|Q| \le
\frac{\tau_{\max}(c n)}{c} \le
\tau_{\max} (n).
\]
Hence, because $G$ is a $(\tau_{\min}, \tau_{\max}, f)$-expander, we have $|N_G(Q)| \ge f(|Q|) \ge f(|S|/c) \ge f(|S|)/c$ (the last inequality coming from sub-homogeneity).
Since none of the parts $P_u, u \in N_G(Q)$ contains a vertex of $S$, each pair $P_u, P_v$ with $u \in Q$ and $v \in N_G(u) \setminus Q$ contains a vertex in $N_{G'}(S)$.
These vertices are disjoint whenever the parts involved are disjoint.
Hence, since the maximum degree of $G$ is $\Delta$, we have that a vertex of $w$ can participate in at most  $\Delta^2$ such pairs.
This implies that $N_{G'}(S) \ge f(|S|)/(c \Delta^2)$.
\end{proof}

\newcommand{\cst}{\mathcal{\mu}}

\begin{corollary}\label{cor:super-poly-exp-2-exp}
If $\mathcal C$ is of superpolynomial expansion, then there exist functions $\tau_{\min}(n) = n^{o(1)}$, $\tau_{\max}(n) = n^{1 - o (1)}$ and $f(n) = n^{1 - o(1)}$ such that $\mathcal C$ contains an infinite family of $(\tau_{\min}, \tau_{\max}, f)$-expanders as subgraphs.
\end{corollary}
\begin{proof}
Since $\mathcal C$ is of superpolynomial expansion, \Cref{thm:poly-exp-poly-clique} implies there exist a function $g : d \mapsto d^{\omega(1)}$ and a family of graphs $(G_d)_{d \in \Nn}$ such that $G_d$ contains $K_{g(d)}$ as a $d$-shallow minor.
Since $g(d) = d^{\omega(1)}$, there exists a superpolynomial function $h(d) = d^{\omega(1)}$ such that $g(d) = h(d)^{\omega(1)}$.

For each $d$, we fix some cubic $(1, n/\cst, \cst \cdot n)$-expander $\Gamma_d$ on $f(d)=d^{1-o(1)}$ vertices, with $\mu > 0$ a~universal constant~(see for instance~\cite{Hoory06}).
By applying \Cref{lem:Kt-to-subexp} to each $G_d$, we obtain a family of graphs $(H_d)_{d \in \Nn}$, where $H_d$ is a~subgraph of~$G_d$ and a~$(\le 4d)$-subdivision of $\Gamma_d$.
By~\Cref{lem:subdiv-expander}, $H_d$ is in particular a~$(\tau_{\min}, \tau_{\max}, f)$-expander where
$\tau_{\min} (|V(H_d)|) = c'$,
$\tau_{\max} (|V(H_d)|) = |V(H_d)|/c'$ and
$f(x) = x/c'$
with 
\[
c' := 100 \mu \cdot d\ge
9\cst \cdot (3d + 1) = 
\cst \cdot (\Delta(\Gamma_d)d + 1)\cdot \Delta(\Gamma_d)^2.
\]
where we weakened the outcome of~\Cref{lem:subdiv-expander} to ease the computations.
It remains to adjust a~bit these functions: indeed, as is, $\tau_{\min}$ is not sufficient large to allow the required expansion.

Consider now the function $\tau_{\min}' := h \circ \tau_{\min}$.
For any $x \ge \tau'_{\min}(|V(H_d)|) \ge d^{\omega(1)}$, we have that
\[
f(x) \ge x/(100 \cst \cdot d) \ge x^{1 - 1/\omega(1)} \ge x^{1 - o(1)},
\]
which allows us to consider a function $f'$ of the form $x \mapsto x^{1 - o(1)}$ such that, when $d$ tends to infinity, $x \ge \tau'_{\min}(|V(H_d)|)$ implies $f(x) \ge f'(x)$.
Finally, we have $\tau'_{\min}(|V(H_d)|)^{\omega(1)} \le \tau_{\max}(|V(H_d)|)$.
Hence, the family $(H_d)_{d \in \Nn}$ is a family of~$(\tau'_{\min}, \tau_{\max}, f')$-expanders that satisfy the lemma.
\end{proof}

The next two lemmas give an upper bound on the number of vertices of a~graph of bounded diameter by a~polynomial function of its cutwidth.
We first show that for trees.

\begin{lemma}[Folklore]\label{lem:cutw2diamtree}
Let $T$ be a~tree with cutwidth at~most~$c$ and diameter at~most~$p$.
Then, $|V(T)| \leqslant (p+1)^c$.
\end{lemma}
\begin{proof}
We prove the result by induction of the cutwidth of~$T$.
Since a connected graph has cutwidth~1 if and only if it is a path, and a path of diameter~$p$ has~$p+1$ vertices, the lemma holds for~$c = 1$.

Assume the lemma holds for graphs of cutwidth at most~$c$, and let~$T$ be a tree of cutwidth~$c+1$.
Let~$\prec$ be a total order of~$V(T)$ witnessing the cutwidth of~$T$.
Let~$P = v_1 \cdots v_k$ be a shortest path of~$T$ connecting the smallest vertex~$v_1$ to the largest vertex~$v_k$ along $\prec$.
Since~$P$ is a shortest path, we have~$k \le p + 1$.
Now, for any~$i \in [k]$, let $C_i$ be the vertex set of the connected component of~$T-E(P)$ containing~$v_i$.

The diameter of~$T[C_i]$ is at most~$p$ since~$T[C_i]$ is a~subtree of~$T$.
Furthermore, the cutwidth of~$T[C_i]$ is at most~$c-1$ along the order~$\prec$:
indeed, any cut of~$G$ along~$\prec$ contains an edge from~$P$.
Hence~$|C_i| \le (p+1)^c$ by the induction hypothesis.
Since~$|V(T)| = \sum_{i \in [k]}|C_i|$, we have~$|V(T)| \le k \cdot (p+1)^c \le (p+1)^{c+1}$.
\end{proof}

\begin{lemma}\label{cor:cutw2diam}
Let $G$ be a~connected graph with cutwidth at~most~$c$ and diameter at~most~$p$.
Then, $|V(G)| \leqslant (2p+1)^c$.
\end{lemma}
\begin{proof}
Let~$v$ be a vertex of~$G$, and let~$T$ be a~breadth-first search (spanning) tree rooted at~$v$.
As the diameter of~$G$ is at most~$p$, the radius of~$T$ is at most~$p$, thus its diameter is at~most~$2p$.
Since~$T$ is a subgraph of~$G$, the cutwidth of~$T$ is at most~$c$, too.
Hence by \Cref{lem:cutw2diamtree}, we have that  $|V(G)|=|V(T)|\le (2p+1)^c$.
\end{proof}

\begin{theorem}
Weakly sparse graphs of bounded reduced cutwidth have polynomial expansion.
\end{theorem}
\begin{proof}
We prove the result by contradiction.
Assume that $\mathcal C$ is a class of $K_{t, t}$-subgraph-free graphs of superpolynomial expansion with reduced cutwidth at most $k$.
By \Cref{cor:super-poly-exp-2-exp}, there exist $\tau_{\min}:n \mapsto n^{o(1)}$, $\tau_{\max}: n \mapsto n^{1 - o(1)}$ and $f:n \mapsto n^{1 - o(1)}$
such that $\mathcal{C}$ contains an infinite family of graphs $(G_d)_d$ such that $G_d$ contains a $(\tau_{\min}, \tau_{\max}, f)$-expander $H_d$ on at least $d$ vertices as a~subgraph.
Note that the graphs $G_d[V(H_d)]$ are also $(\tau_{\min}, \tau_{\max}, f)$-expanders on at least $d$ vertices.

Let $a := 1 - 1/(1+2^k)$. We have $f(n) = \Omega(n^a)$, and $\tau_{\max}(n)^{1 - a} \ge (\log n + \tau_{\min}(n)) \log n$. Hence, by \Cref{lem:from-growth-2-red-diameter} we know that any contraction sequence of reduced cutwidth at~most~$k$ of $G_d[V(H_d)]$ eventually creates a~red connected graph $R_d$ of cutwidth at most $k$ with $|V(R_d)| = \Omega(\diam(R_d)^{\frac{a}{1-a}})$ and $|V(R_d)| = \Omega((\log d)^{a})$.
In particular, since the maximum degree of $R_d$ is at~most~$2k$, the diameter of $R_d$  

is an unbounded function of~$d$.

However, \Cref{cor:cutw2diam} implies that $|V(R_d)| \le (2 \diam(R_d) + 1)^k$.
Hence, $(2 \diam(R_d) + 1)^k = \Omega(\diam(R_d)^{\frac{a}{1-a}}) = \Omega(\diam(R_d)^{2^{k}})$; a~contradiction.
\end{proof}

\section{Stability under transductions}\label{sec:stability}
For a graph~$G = (V,E)$ and a parameter~$k \in \Nn$, we define the following two operations:
\begin{description}
  \item[power] The \emph{$k$th power} of~$G$, denoted~$\pow[k](G)$ is the graph with vertices~$V$, and an edge~$xy$ for all pairs of vertices at distance at most~$k$ in~$G$.
  \item[blowup] The \emph{$k$th blowup} of~$G$, denoted~$\blow[k](G)$, is obtained by duplicating each~$x \in V$ into a clique $\{x_1,\dots,x_k\}$, and connecting every~$x_i$ and~$y_j$ whenever~$xy$ is an edge of~$G$.
\end{description}

A graph parameter~$p$ is said to be stable under powers if there is some function~$f$ such that for every graph~$G$ and integer~$k$, $p(\pow[k](G)) \le f(p(G),k)$.
Being stable under blowups is defined similarly.
The goal of this section is to prove the following criteria for stability under first-order transductions of reduced parameters (see \cref{sec:FOprelim} for definitions).
\begin{restatable}[\cref{thm:transduction-intro} reformulated]{theorem}{transduction}
  \label{thm:transduction}
  Let~$p$ be a monotone parameter stable under powers and blowups.
  Then reduced-$p$ is stable under first-order transductions, that is, for any first-order transduction $\Phi$, there is a function~$f_\Phi$ such that for any graph~$G$ and $H \in \Phi(G)$,
  $\reduced{p}(H) \le f_\Phi(\reduced{p}(G))$.
\end{restatable}
For simplicity, we present this result only in the case of graphs, but it is easy to extend it to binary relational structures,
i.e.\ relational structures consisting only of unary and binary relations
(see e.g.\ the note at the end of \cite[Section~7]{twin-width1} on how to extend twin-width to binary structures; the same works for reduced parameters).

\subsection{First-order logic, transductions, and Gaifman locality}\label{sec:FOprelim}
Let us first recall the definition of first-order transductions, and introduce the main necessary logical tool: Gaifman locality.

First-order formulas on graphs are formulas constructed from quantifiers $\forall x, \exists y$ on vertices,
boolean operations, and predicates $E(x,y)$ testing adjacency between vertices, and $x=y$ testing equality.
When a formula~$\phi$ has free variables, they are written as parameters $\phi(x_1,\dots,x_n)$.
For a graph~$G$ and vertices $v_1,\dots,v_n \in V(G)$, one denotes by $G \models \phi(v_1,\dots,v_n)$ the fact that~$G$ satisfies~$\phi(x_1,\dots,x_n)$ when~$x_i$ is interpreted as~$v_i$, defined in the obvious way.
The \emph{quantifier rank} of a formula is the maximum nesting depth of quantifiers.

More generally, we will work with \emph{colored graphs}, meaning a graph~$G$ augmented by unary predicates $U_1,\dots,U_\eta \subseteq V(G)$.
A first-order formula on the colored graph $G^+ = (G,U_1,\dots,U_\eta)$ can additionally use predicates to test if a vertex~$x$ belongs to~$U_i$, denoted $U_i(x)$.

Let $G^+ = (G,U_1,\dots,U_\eta)$ be a colored graph, and~$\phi(x,y)$ a formula with two free variables, which may use the unary predicates $U_1,\dots,U_\eta$.
The \emph{first-order interpretation} $\phi(G^+)$ is defined as the graph with the same vertex set~$V(G)$, and where~$uv$ is an edge if and only if $G^+ \models \phi(u,v)$.\footnote{Some definitions allow a first-order interpretation to restrict the vertex set to some subset specified by a second formula~$\psi(x)$. We add the step of removing vertices in the definition of transductions instead.}
In general, $\phi(G^+)$ is a directed graph. We may require~$\phi(x,y)$ to be equivalent to~$\phi(y,x)$, or replace it by $\phi(x,y) \lor \phi(y,x)$ to ensure $\phi(G^+)$ is in fact a graph.

For $\gamma,\eta \in \Nn$ and~$\phi(x,y)$ a first-order formula with two free-variables on colored graphs with unary predicates~$U_1,\dots,U_\eta$,
the \emph{first-order transduction} $\Phi = (\gamma,\eta,\phi)$ is a one-to-many map from graphs to graphs defined as follows:
\begin{enumerate}
    \item Given a graph~$G$, first define its $\gamma$-copy~$G_\gamma$ by taking~$\gamma$ disjoint copies of~$G$, and for each vertex $v \in V(G)$, adding a clique of size~$\gamma$ on the~$\gamma$ copies of~$v$.
    Equivalently, $G_\gamma = K_\gamma \square G$ is the graph cartesian product of~$G$ with a clique of size~$\gamma$.\footnote{
        It is common to also add unary predicates to distinguish the~$\gamma$ copies, and to use a different binary relation to link the different copies of the same vertex~$v$.
        However, these can both be simulated by adding appropriate unary predicates in the next step.
    }
    \item Next, consider $G_\gamma^+ = (G_k,U_1,\dots,U_\eta)$ a colored graph obtained by augmenting~$G_k$ with \emph{any} subsets~$U_1,\dots,U_\eta$ as unary predicates.
    \item Then, apply the first-order interpretation~$\phi(G_k^+)$.
    \item Finally, take any induced subgraph~$H$ of~$\phi(G_k^+)$.
\end{enumerate}
Then $\Phi(G)$ is defined as the set of all possible graphs~$H$ obtained through the previous steps.

The following is a well-known consequence of Gaifman's locality theorem.
\begin{theorem}[Gaifman's locality~\cite{gaifman}]\label{thm:gaifman}
    For any first-order formula~$\phi(x,y)$, there are parameters $r,\ell \in \Nn$,
    where~$r$ only depends on the quantifier rank of~$\phi$, satisfying the following.
    For any (colored) graph~$G$, there is a map $\lambda : V(G) \to [\ell]$ such that
    for any vertices~$x,y$ at distance more than~$r$ in~$G$,
    whether or not~$\phi(x,y)$ holds depends only on $\lambda(x),\lambda(y)$.
\end{theorem}
The parameter~$r$ is called \emph{locality radius} or \emph{Gaifman radius} of~$\phi$.

For a graph~$G$ and two subsets~$A,B$ of vertices, the \emph{flip} $G \oplus (A,B)$
is the graph obtained by complementing the adjacency between~$A$ and~$B$, i.e.
\[ E(G \oplus (A,B)) = E(G) \triangle \{ ab : a \in A, b \in B\}. \]

A graph~$H$ is called a \emph{$t$-flip} of~$G$ if there are $(A_i,B_i)_{i \in [t]}$ such that
$H = \big(G \oplus (A_1,B_1) \big) \dots \oplus (A_t,B_t)$.
We will use the following variant of Gaifman's locality,
in which distances are measured not in the given graph~$G$, but in a bounded flip of~$G$.
This result is implicit in e.g.\ the proof of \cite[Lemma~5.3]{dreier2023flipflatness}.

\begin{lemma}[Gaifman's locality with flips]\label{lem:flip-gaifman}
    For any first-order formula~$\phi(x,y)$ and $t \in \Nn$, there are parameters $r,\ell \in \Nn$,
    where~$r$ only depends on the quantifier rank of~$\phi$ (and in particular not on~$t$), satisfying the following.
    Given~$G$ a (colored) graph and~$H$ a $t$-flip of~$G$, 
    there is a map $\lambda : V(G) \to [\ell]$ such that for any vertices~$x,y$ at distance more than~$r$ in~$H$,
    whether or not~$G \models \phi(x,y)$ depends only on $\lambda(x),\lambda(y)$.
\end{lemma}
\begin{proof}
    Consider~$2t$ additional unary predicates $A_1,B_1,\dots,A_t,B_t$,
    and let~$\phi'(x,y)$ be the formula obtained from~$\phi(x,y)$ by replacing any instance of the edge predicate~$E(u,v)$ by
    \begin{equation}\label{eq:flipped-edges-formula}
        E(u,v) \oplus \bigoplus_{i=1}^t (A(u) \land B(v)) \lor (B(u) \land A(v)),
    \end{equation}
    where~$\oplus$ denotes exclusive `or'.
    For any $t$-flip~$H$ of~$G$, when $(A_i,B_i)_{i \in [t]}$ are interpreted as the subsets of~$G$ that are flipped to obtain~$H$,
    then formula~\eqref{eq:flipped-edges-formula} describes the edges of~$H$, and thus
    \begin{equation}\label{eq:flip-interpertation}
        G \models \phi(x,y) \quad \iff \quad H \models \phi'(x,y).
    \end{equation}
    Let~$r,\ell \in \Nn$ be given by \cref{thm:gaifman} for~$\phi'$.
    Note that~$\phi,\phi'$ have the same quantifier rank, hence~$r$ only depends on the quantifier rank of~$\phi$, while~$\ell$ only depends on~$\phi$ and~$t$.
    Now \cref{thm:gaifman} applied to~$H$ gives a map $\lambda : V(H) \to [\ell]$
    such that for~$x,y$ at distance more than~$r$ in~$H$,
    whether or not $H \models \phi'(x,y)$ depends only on $\lambda(x),\lambda(y)$.
    Combined with~\eqref{eq:flip-interpertation}, this proves the statement.
\end{proof}

\subsection{Rank and contraction sequences}
To prove \cref{thm:transduction}, we will use a variant of contraction sequences defined using rank.
For a~graph~$G$, let $A(G)$ denote the \emph{adjacency matrix} of~$G$.
For two disjoint subsets $X$ and~$Y$ of vertices in a graph $G$, the \emph{rank} between $X$ and $Y$ in $G$, denoted by $\rk_G (X,Y)$, is the binary rank of $A(G)[X,Y]$.

Consider a graph~$G$, a partition~$\Pc$ of $V(G)$, and a second graph~$G_\Pc$ with vertex set~$\Pc$.
We say that~$G_\Pc$ is a \emph{rank-$k$ error graph} for~$(G,\Pc)$ if for any part~$P \in \Pc$, we have
\[ \textstyle \rk_G\left(P, \bigcup_{Q \not\in N_{G_\Pc}[P]} Q\right) \le k. \]

That is, the graph~$G_\Pc$ specifies a number of `bad parts' which are neighbors of the part~$P$, and after removing these bad parts, the rank between~$P$ and the remainder of $G$ is at~most~$k$.

The \emph{reduced-rank-$p$} of~$G$ is the smallest~$k$ such that the following holds:
There is a contraction sequence $\Pc_n,\dots,\Pc_1$ and a sequence of graphs $G_n,\dots,G_1$, such that
\begin{itemize}
    \item $G_i$ is a rank-$k$ error graph for~$(G,\Pc_i)$,
    \item $p(G_i) \le k$, and
    \item $G_n,\dots,G_1$ is monotone, i.e., for $i > j$, if~$PQ$ is an edge in~$G_i$
    and~$P',Q'$ are the parts of~$G_j$ containing~$P$ and~$Q$ respectively,
    then~$P'Q'$ is also an edge in~$G_j$.
\end{itemize}

The goal of this subsection is the following:
\begin{lemma}
  If~$p$ is a monotone parameter stable under blowups,
  then bounded reduced-rank-$p$ is equivalent to bounded reduced-$p$.
  \label{lem:reduced-rank-blowup}
\end{lemma}
We prove the following more general statement, which will be useful for stretch-width:
\begin{lemma}
    \label{lem:rank-contraction-refinement}
    Let $\Pc_n,\dots,\Pc_1$ be a contraction sequence for~$G$, and let~$G_i$ be a rank-$k$ error graph for~$(G,\Pc_i)$, with $G_n,\dots,G_1$ monotone.
    Then there is a second contraction sequence $\Pc'_n,\dots,\Pc'_1$ for~$G$ satisfying the following:
    for each~$i \in [n]$, there is a $j \in [n]$ such that
    \begin{enumerate}
        \item \label{item:refine-monotone} $\Pc'_i$ refines~$\Pc_j$ (i.e.\ each part of~$\Pc'_i$ is contained in some part of~$\Pc_j$);
        \item \label{item:refine-bounded} Each part of~$\Pc_j$ is the union of at most~$2^{2k+1}$ parts of~$\Pc'_i$; and
        \item \label{item:refine-homogeneous} Let $P',Q' \in \Pc'_i$ be contained in $P,Q \in \Pc_j$ respectively. If~$P',Q'$ are non-homogeneous in~$G$, then~$P,Q$ are equal or adjacent in~$G_j$.
    \end{enumerate}
\end{lemma}
\begin{proof}[Proof of \cref{lem:reduced-rank-blowup}]
  Consider a contraction sequence $\Pc_n,\dots,\Pc_1$ of a~graph~$G$, and the associated rank-$k$ error graph~$G_i$ with $p(G_i) \le k$.
  \Cref{lem:rank-contraction-refinement} gives a second contraction sequence $\Pc'_n,\dots,\Pc'_1$
  such that each~$\Pc'_i$ is obtained from some~$\Pc_j$ by splitting each part $P \in \Pc_j$ into at most~$2^{2k+1}$ parts.
  Further, two parts of~$\Pc'_i$ can only be non-homogeneous if they come from equal or adjacent (in~$G_j$) parts of~$\Pc_j$.
  This exactly means that the red graph $\redG(G/\Pc'_i)$ is a subgraph of the $2^{2k+1}$-blowup of~$G_j$.
  Since $p(G_j) \le k$, and~$p$ is stable under blowups, we obtain that $p(\redG(G/\Pc'_i))$ is bounded by a function of~$k$,
  hence~$G$ has bounded reduced-$p$.
\end{proof}

Let us now focus on the more technical statement about refining contraction sequences.
\begin{proof}[Proof of \cref{lem:rank-contraction-refinement}]
  Consider a contraction sequence $\Pc_n,\dots,\Pc_1$ of~$G$, with a monotone sequence of rank-$k$ error graphs~$G_n,\dots,G_1$.
  We refine the partition~$\Pc_i$ into~$\Pc''_i$ as follows.
  Consider some~$P \in \Pc_i$, and denote $R_P = \bigcup_{Q \not\in N_{G_i}[P]} Q$, so that by assumption $\rk_G(P,R_P) \le k$.
  We then partition~$P$ by neighborhood type in~$R_P$, i.e., $x,y \in P$ remain in the same part if and only if they have exactly the same neighbors in~$R_P$.
  By the rank condition, this splits~$P$ into at most~$2^k$ classes.
  Finally, $\Pc''_i$ is obtained by applying the above to all parts~$P \in \Pc_i$ independently.

  \begin{claim}\label{claim:homogeneous}
    Let $P''_1,P''_2 \in \Pc''_i$ be parts resulting from the blowup of $P_1,P_2 \in \Pc_i$ respectively.
    If~$P''_1$ and~$P''_2$ are not homogeneous in~$G$, then~$P_1$ and $P_2$ are equal or adjacent in~$G_i$.
  \end{claim}
  \begin{clproof}
    Suppose that~$P_1,P_2$ are distinct and non-adjacent in~$G_i$.
    Then~$\Pc''_i$ is defined so that the vertices of~$P''_1$ all have the same neighbors in~$P_2$, and a fortiori in~$P''_2$.
    Symmetrically, all vertices of~$P''_2$ have the same neighbors in~$P''_1$.
    It follows that~$P''_1$ and $P''_2$ are homogeneous.
  \end{clproof}
  Thus the sequence of partitions $\Pc''_n,\dots,\Pc''_1$ satisfies the three conditions of the statement.
  The only issue is that it is not a contraction sequence:
  when going from~$\Pc''_i$ to~$\Pc''_{i-1}$, we might merge several parts at once, or none at all.
  Let us now show how to fill these gaps.

  \begin{claim}
    \label{clm:neigh-refine-monotone}
    For any~$i > j$, $\Pc''_i$ refines~$\Pc''_j$.
  \end{claim}
  \begin{clproof}
    By definition, $\Pc_i$ refines $\Pc_j$.
    In particular, each part $P$ of $\Pc_i$ is included in a single part $\Tilde{P}$ of $\Pc_j$.
    By monotonicity of $G_n, \dots, G_1$, we have \[\bigcup_{Q \in N_{G_i}[P]} Q \subseteq \bigcup_{Q \in N_{G_j}[\Tilde{P}]} Q.\]
    The complements of these sets are~$R_P$ and~$R_{\Tilde{P}}$, and thus satisfy $R_{\Tilde{P}} \subseteq R_P$.
    Therefore, if $x,y \in P$ are in the same part of~$\Pc''_i$, meaning that they have the same neighbors in~$R_P$,
    then~$x,y$ are also both in~$\Tilde{P}$ and have the same neighbors in~$R_{\Tilde{P}}$, hence are in the same part of~$\Pc''_j$.
  \end{clproof}

  \begin{claim}
    \label{clm:bounded-refinment}
    Any~$P'' \in \Pc''_{i-1}$ is the union of at most~$2^{k+1}$ parts of~$\Pc''_i$.
  \end{claim}
  \begin{clproof}
    Pick~$P'' \in \Pc''_{i-1}$ contained in some~$P \in \Pc_{i-1}$.
    Either~$P \in \Pc_i$, or~$P$ is the union of two parts~$P_1,P_2$ of~$\Pc_i$.
    We will only consider the second case, which is worse.
    
    Consider now some~$Q'' \in \Pc''_i$ contained in~$P''$,
    and denote by~$Q$ the part of~$\Pc_i$ containing~$Q''$.
    Then~$Q$ must be either~$P_1$ or~$P_2$.
    We know that~$Q$ is split into at most~$2^k$ parts of~$\Pc''_i$.
    Adding the choice of either~$Q = P_1$ or~$Q = P_2$, this gives~$2 \cdot 2^k$ choices for~$Q''$.
  \end{clproof}

  Since the sequence $\Pc''_n,\dots,\Pc''_1$ is monotone under refinement by \cref{clm:neigh-refine-monotone},
  it can be completed into a contraction sequence by arbitrarily adding intermediate steps so that only two parts are merged at a time,
  and skipping steps when the same partition is seen several times in a row.
  Consider any partition~$\Pc''$ inserted between~$\Pc''_i$ and~$\Pc''_{i-1}$.
  It follows from \cref{clm:bounded-refinment} that each part of~$\Pc''_{i-1}$ is the union of at most~$2^{k+1}$ parts of~$\Pc''$.
  In turn, each part of~$\Pc_{i-1}$ is the union of at most~$2^k$ parts of~$\Pc''_{i-1}$.
  Thus~$\Pc''$ refines~$\Pc_{i-1}$, and each part of~$\Pc_{i-1}$ is the union of at most~$2^{2k+1}$ parts of~$\Pc''$,
  ensuring conditions~\labelcref{item:refine-monotone,item:refine-bounded} of the statement.
  Also, by \cref{claim:homogeneous}, if $P,Q \in \Pc_{i-1}$ are distinct and non-adjacent in~$G_{i-1}$,
  then any resulting parts $P'' \subseteq P$ and $Q'' \subseteq Q$ with $P'',Q'' \in \Pc''_{i-1}$ must be homogeneous.
  The same is a fortiory true when~$P'',Q''$ belong to~$\Pc''$ which refines~$\Pc''_{i-1}$, ensuring condition~\labelcref{item:refine-homogeneous}.

  It only remains to reindex as $\Pc'_n,\dots,\Pc'_1$ this contraction sequence obtained from $\Pc''_n,\dots,\Pc''_1$
  by adding arbitrary intermediate steps and skipping duplicates.
  After reindexing, there is no clear link between~$\Pc'_i$ and~$\Pc_i$,
  but it still holds that for any~$\Pc'_i$ there is a corresponding step~$\Pc_j$ (for a different index~$j$)
  satisfying conditions~\labelcref{item:refine-monotone,item:refine-bounded,item:refine-homogeneous}.
\end{proof}

\subsection{Stability under transductions}
Using Gaifman's locality, we can prove the following technical lemma, showing that when applying a first-order interpretation~$\phi$ to~$G$,
one can obtain bounded-rank error graphs for~$\phi(G)$ as bounded powers of red graphs for~$G$.
\begin{lemma}\label{lem:error-graph-transduction}
    For any first-order formula~$\phi(x,y)$ with unary predicates~$U_1,\dots,U_\eta$, and $\Delta \in \Nn$,
    there are parameters $r,\ell \in \Nn$, where~$r$ only depends on the quantifier rank of~$\phi$, satisfying the following.
    Let $G^+ = (G,U_1,\dots,U_\eta)$ be any colored graph, and~$\Pc$ be a~partition of~$V(G)$ with red graph $G_\Pc := \redG(G/\Pc)$.
    Assume that~$G_\Pc$ has maximum degree at most~$\Delta$. Then the $r$-th power of~$G_\Pc$ is a rank-$\ell$ error graph for~$(\phi(G^+),\Pc)$.
\end{lemma}
\begin{proof}
    Depending only on the quantifier rank of~$\phi$, \cref{lem:flip-gaifman} gives
    a locality radius~$r \in \Nn$.
    Choose the number of flips $t = \Delta^r+1$, which bounds the number of vertices in a ball of radius~$r$ in a graph with maximum degree~$\Delta$.
    Now, depending on~$\phi$ and~$t$, \cref{lem:flip-gaifman} gives a second parameter~$\ell$.

    Fix now the colored graph $G^+ = (G,U_1,\dots,U_\eta)$, and the partition~$\Pc$ of~$G$ whose red graph~$G_\Pc$ has maximum degree at most~$\Delta$, as well as a part $P \in \Pc$.
    Let $\Bc := \Ball_{G_\Pc}^r(P)$ be all the parts of~$\Pc$ within distance~$r$ of~$P$ in the red graph~$G_\Pc$, and note that $|\Bc| \le t$.
    Consider the graph~$G'$ obtained from~$G$ as follows:
    for each pair of parts $X,Y \in \Pc$, if~$X,Y$ are fully adjacent and either~$X$ or~$Y$ is in~$\Bc$, then remove all edges between~$X$ and~$Y$.
    For any fixed~$X$, one can with a 1-flip remove all edges between~$X$ and all the parts~$Y$ that are fully adjacent to~$Y$.
    The graph~$G'$ is obtained by doing such a 1-flip for each $X \in \Bc$, hence it is a $t$-flip of~$G$.
    
    Let $R_P = V(G) \setminus \bigcup_{Q \in \Bc} Q$ denote the vertices outside~$\Bc$. The graph~$G'$ is chosen to ensure that
    \begin{equation}\label{eq:dist-flip}
        \text{for any $x \in P$ and $y \in R_P$, $\dist_{G'}(x,y) > r$.}
    \end{equation}
    We now apply \cref{lem:flip-gaifman} to the colored graph~$G^+$ and its flip~$G'$,
    yielding a map $\lambda : V(G) \to [\ell]$ such that if~$x,y$ are at distance more than~$r$ in~$G'$,
    then whether or not $G^+ \models \phi(x,y)$ (or equivalently $xy \in E(\phi(G^+))$)
    only depends on~$\lambda(x),\lambda(y)$.
    By~\eqref{eq:dist-flip}, this is in particular the case if $x \in P$ and $y \in R_P$.
    It immediately follows that $\rk_{\phi(G^+)}(P,R_P) \le \ell$.

    To summarize, for any part $P \in \Pc$, if~$\Bc$ denotes all the parts within distance~$r$ of~$P$ in~$G_\Pc$,
    and~$R_P$ the set of vertices outside~$\Bc$, then $\rk_{\phi(G^+)}(P,R_P) \le \ell$.
    In the $r$th power of~$G_\Pc$, the neighborhood of~$P$ is exactly~$\Bc$,
    hence this shows that the $r$th power of~$G_\Pc$ is a rank-$\ell$ error graph for~$\phi(G^+)$, as desired.
\end{proof}

We need one last lemma for the copying operation. It is a simple variant of \cite[Lemma~8.2]{twin-width1}.
\begin{lemma}\label{lem:copying}
    Let~$p$ be a monotone parameter stable under blowups. Then for any $\gamma \in \Nn$, reduced-$p$ is stable under the $\gamma$-copy operation.
\end{lemma}
\begin{proof}
    For a graph~$G$, consider the $\gamma$-copy $G_\gamma = K_\gamma \square G$ of~$G$.
    For each vertex $v \in V(G)$, let us denote by $v_1,\dots,v_\gamma$ the copies of~$v$ in~$G_\gamma$.
    Now if~$\Pc$ is a partition of~$V(G)$, consider the partition~$\Pc'$ of~$G_\gamma$ in which $u_i,v_j$ are in the same part if and only if $i=j$ and $u,v$ are in the same part of~$\Pc$.
    That is, we replicate~$\Pc$ for each of the copies of~$G$, and keep vertices from different copies in different parts.
    Then the red graph $\redG(G_\gamma / \Pc)$ is a subgraph of the $\gamma$-copy of $\redG(G/\Pc)$, which itself is a subgraph of the $\gamma$-blowup of $\redG(G/\Pc)$.

    Now consider a contraction sequence $\Pc_n,\dots,\Pc_1$ witnessing that $\reduced{p}(G) \le k$,
    and the corresponding partitions $\Pc'_n,\dots,\Pc'_1$ of $V(G_\gamma)$ as described above.
    If~$\Pc_i$ is obtained from~$\Pc_{i+1}$ by merging parts~$X,Y$, then~$\Pc'_i$ is obtained from~$\Pc'_{i+1}$ by merging each of the~$\gamma$ copies of~$X$ with the corresponding copy of~$Y$.
    Thus $\Pc'_n,\dots,\Pc'_1$ is not quite a contraction sequence,
    but it can be turned into one by adding additional steps to (1) merge the copies of~$X$ and~$Y$ one by one between~$\Pc'_{i+1}$ and~$\Pc'_i$ for each~$i$, and (2) merge the~$\gamma$ final parts in~$\Pc_1$.
    It is simple to check that in all these intermediate steps, the corresponding quotient red graph is a subgraph of the $(2\gamma-1)$-blowup of $\redG(G/\Pc_i)$ for case (1), and a subgraph of~$K_\gamma$ (i.e.\ the $\gamma$-blowup of a single vertex) in case (2).
    Since~$p$ is stable under blowup, these red graphs have bounded~$p$ (function of~$k$), proving the result.
\end{proof}

We can now prove the main result of this section.
\transduction*
\begin{proof}
    Firstly, we may assume that
    \begin{equation}\label{eq:assume-bounded-degree}
        \text{bounded~$p$ implies bounded degree.}
    \end{equation}
    Indeed, observe that any graph has a contraction sequence in which every red graph is a star; this can be obtained by repeatedly identifying the part corresponding to the center of the star with any other part.
    If bounded~$p$ does not imply bounded degree but~$p$ is monotone, then stars must have bounded~$p$, hence all graphs have bounded reduced-$p$, and the result is trivial.

    Consider $\Phi = (\gamma,\eta,\phi)$ a first-order transduction, and~$G$ a graph with $\reduced{p}(G) = k$.
    \Cref{lem:copying} gives that its $\gamma$-copy~$G_\gamma$ satisfies $\reduced{p}(G_\gamma) \le k_1$ for some~$k_1$ depending only on~$k,\gamma$.
    Fix $\Pc_n,\dots,\Pc_1$ a contraction sequence for~$G_\gamma$ witnessing $\reduced{p}(G_\gamma) \le k_1$.    
    By \eqref{eq:assume-bounded-degree}, there is some bound $\Delta \in \Nn$ such that any graph~$H$ with $p(H) \le k_1$ has maximum degree at most~$\Delta$. In particular the red graphs $\redG(G/\Pc_i)$ have maximum degree at most~$\Delta$.

    Now consider any colored graph $G_\gamma^+ = (G_\gamma,U_1,\dots,U_\eta)$ obtained by adding unary predicates to~$G_\gamma$.
    By \cref{lem:error-graph-transduction}, there are parameters $r,\ell \in \Nn$ depending only on~$\phi,\Delta$ such that
    the $r$-th power $G_i = \pow[r](\redG(G/\Pc_i))$ is a rank-$\ell$ error graph for $(\phi(G_\gamma^+),\Pc_i)$.
    Since~$p$ is stable under powers, there is some bound~$k_2$ depending only on~$k_1,r$ such that $p(G_i) \le k_2$.
    Thus $\Pc_n,\dots,\Pc_1$ with the error graphs $G_n,\dots,G_1$ witnesses that~$\phi(G_\gamma^+)$ has reduced-rank-$p$ at most $\max(k_2,\ell)$.
    By \cref{lem:reduced-rank-blowup}, this implies that $\reduced{p}(\phi(G_\gamma^+)) \le k_3$ for some bound~$k_3$ depending only on~$\max(k_2,\ell)$.
    Finally, taking any induced subgraph of~$\phi(G_\gamma^+)$ does not increase its reduced-$p$.

    Thus, any graph in~$\Phi(G)$ has reduced-$p$ bounded by~$k_3$, which is a function of~$k,\gamma,\phi$ only, proving the result.
\end{proof}
One can observe that bandwidth, cutwidth, $\text{pathwidth}+\Delta$, $\text{treewidth}+\Delta$, and $\Delta$ are all monotone and stable under powers and blowups.
Therefore, \cref{thm:transduction} implies the following.

\begin{corollary}\label{cor:fo-transductions}
The families of classes of bounded reduced bandwidth, reduced cutwidth, reduced $\text{pathwidth}+\Delta$, reduced $\text{treewidth}+\Delta$, and twin-width respectively,
are each closed under applying first-order transductions.
\end{corollary}

This fact was shown for twin-width in~\cite{twin-width1}.
Notably, the family of classes of bounded reduced component max-leaf is \emph{not} closed under first-order transductions since unit interval graphs (of bounded $\compmaxleaf^\downarrow$) first-order transduce the planar grids (of unbounded $\compmaxleaf^\downarrow$).

\subsection{Stretch-width}
\Cref{thm:transduction} does not directly apply to stretch-width, as it is not defined as a reduced-$p$ parameter.
Nonetheless, \cref{lem:stretch-blowup,lem:error-graph-transduction} can be used to show that it is stable under transductions.

Consider an ordered graph~$(G,\psi)$.
We first extend the notion of stretch to a partition~$\Pc$ of~$V(G)$ equipped with any graph~$G_\Pc$ with vertex set~$\Pc$.
Recall that two sets $X,Y \in V(G)$ are said to be in conflict if their spans (relative to~$\psi$) intersect.
For parts $X,Y \in \Pc$, we say that~$Y$ \emph{$G_\Pc$-interferes} with~$X$ if~$Y$ conflicts with $\bigcup N_{G_\Pc}[X]$.
The \emph{$G_\Pc$-stretch} of~$X$, denoted by  $\str_{G_\Pc}(X)$, is the number of parts in~$\Pc$ which $G_\Pc$-interfere with~$X$.
The stretch of~$(\Pc,G_\Pc)$, denoted by  $\str(\Pc,G_\Pc)$, is the maximum of~$\str_{G_\Pc}(X)$ over all choices of $X \in \Pc$.
When~$G_\Pc$ is the red graph $\redG(G/\Pc)$, we recover the usual notion of stretch.

This notion of stretch relative to a graph~$G_\Pc$ on~$\Pc$ is clearly monotone: removing edges in~$G_\Pc$ can only decrease its stretch.
The next lemmas show that it is stable under powers and blowups.
\begin{lemma}\label{lem:stretch-power}
    Let~$(G,\psi)$ be an ordered graph, $\Pc$ a partition of~$V(G)$, and~$G_\Pc$ a graph with vertex set~$\Pc$.
    Then for any $r \in \Nn$, the $r$-th power of~$G_\Pc$ satisfies
    \[ \str(\Pc,\pow[r](G_\Pc)) \le \left(\str(\Pc,G_\Pc)+1\right)^r. \]
\end{lemma}
\begin{proof}
    Consider the directed graph~$D$ on vertex set~$\Pc$ where an edge $X \to Y$ indicates that~$Y$ $G_\Pc$-interferes with~$X$.
    By definition, $\str(\Pc,G_\Pc)$ is the maximum out-degree of~$D$.
    Denote by $\vec{\Ball}^r_D(X)$ the directed ball of radius~$r$,
    i.e.\ the set of parts~$Y$ such that there is a directed path from~$X$ to~$Y$ of length at most~$r$.
    We claim that this directed $r$-ball in~$D$ captures the parts interfering with~$X$ in the $r$th-power:
    \begin{equation}\label{eq:span-rball}
        \text{if~$Y$ conflicts with } \conv\left(\bigcup \Ball^r_{G_\Pc}(X)\right) \text{, then $Y \in \vec{\Ball}^r_D(X)$}.
    \end{equation}
    Since~$D$ has out-degree at most~$\str(\Pc,G_\Pc)$, this immediately implies the statement.

    To prove~\eqref{eq:span-rball}, call $S^r = \conv\left(\bigcup \Ball^r_{G_\Pc}(X)\right)$ the span of the $r$-ball under consideration,
    and consider a part~$Y$ in conflict with~$S^r$.
    We proceed by induction on~$r$, the base case $r=1$ being immediate by choice of~$D$.
    If~$Y$ is also in conflict with the span~$S^{r-1}$ (defined similarly for the $(r-1)$-ball), then the result holds by induction.
    Thus we can assume that~$S^{r-1}$ and~$\conv(Y)$ are disjoint.

    Call~$a,b$ the minimum and maximum of~$S^r$ respectively.
    The corresponding parts $A \ni a, B \ni b$ of~$\Pc$ must belong to~$\Ball^r_{G_\Pc}(X)$,
    hence there are neighbors~$A',B' \in \Ball^{r-1}_{G_\Pc}$ of~$A,B$ respectively in~$G_\Pc$.
    Note that $\bigcup N_{G_\Pc}[A']$ contains~$a$ (since~$A$ is a neighbor of~$A'$), and intersects~$S^{r-1}$ (since $A' \subseteq S^{r-1}$).
    Similarly, $\bigcup N_{G_\Pc}[B']$ contains~$b$ and intersects~$S^{r-1}$.
    It follows that
    \begin{equation}
        S^r = [a,b] \subseteq \conv\left(\bigcup N_{G_\Pc}[A']\right) \cup S^{r-1} \cup \conv\left(\bigcup N_{G_\Pc}[B']\right).
    \end{equation}
    Now, since~$\conv(Y)$ intersects~$S^r$ but not~$S^{r-1}$, it must intersect either $\conv\left(\bigcup N_{G_\Pc}[A']\right)$ or $\conv\left(\bigcup N_{G_\Pc}[B']\right)$.
    By definition, this means~$Y$ is an out-neighbor of~$A'$ or~$B'$ in~$D$.
    Finally, any edge $PQ \in G_\Pc$ directly yields an edge $P \to Q$ (and $Q \to P$) in~$D$,
    hence there also are directed paths of length at most~$r-1$ from~$X$ to~$A'$ and to~$B'$.
    Thus there is a directed path of length at most~$r$ from~$X$ to~$Y$ in~$D$, proving~\eqref{eq:span-rball} and the lemma.
\end{proof}

\begin{lemma}\label{lem:stretch-blowup}
    Let~$(G,\psi)$ be an ordered graph, $\Pc$ a partition of~$V(G)$, and~$H$ a graph with vertex set~$\Pc$.
    Consider a second partition~$\Pc'$ obtained by splitting each part of~$\Pc$ into at most~$k$ new parts,
    and call~$H'$ the corresponding blowup of~$H$, i.e.\ the graph with vertex set~$\Pc'$
    where $P',Q' \in \Pc'$ are adjacent if and only if the parts $P,Q \in \Pc$ containing them are equal or adjacent.
    Then
    \[ \str(\Pc',H') \le k \cdot (\str(\Pc,H)+1). \]
\end{lemma}
\begin{proof}
    For any $X' \in \Pc'$ contained in $X \in \Pc$, it is immediate to check that
    \begin{equation}
        \bigcup N_{H'}[X'] = \bigcup N_{H}[X].
    \end{equation}
    Thus, if a part $Y' \in \Pc'$ $H'$-interferes with~$X'$, then the part~$Y \in \Pc$ containing~$Y'$ also $H$-interferes with~$X$, or alternatively~$Y=X$.
    Since the number of parts~$Y$ interfering with~$X$ is at most $\str(\Pc,H)$, and each of them is split into at most~$k$ parts of~$\Pc'$, this proves the lemma.
\end{proof}

Now the \emph{rank-stretch-width} of the ordered graph~$(G,\psi)$ is the smallest~$k$ such that the following holds:
There is a contraction sequence $\Pc_n,\dots,\Pc_1$ and a sequence of graphs $G_n,\dots,G_1$, such that
\begin{itemize}
    \item $G_i$ is a rank-$k$ error graph for~$(G,\Pc_i)$,
    \item $\str(\Pc_i,G_i) \le k$, and
    \item $G_n,\dots,G_1$ is monotone: for $i > j$, if~$PQ$ is an edge in~$G_i$
    and~$P',Q'$ are the parts of~$G_j$ containing~$P$ and~$Q$ respectively,
    then~$P'Q'$ is also an edge in~$G_j$.
\end{itemize}

\begin{lemma}\label{lem:rank-stretch-width}
    An ordered graph~$(G,\psi)$ with rank-stretch-width~$k$ has stretch-width at most~$k \cdot 2^{2k+1}$.
\end{lemma}
\begin{proof}
    Let $\Pc_n,\dots,\Pc_1$ be a contraction sequence for~$(G,\psi)$,
    and $G_n,\dots,G_1$ a corresponding monotone sequence of rank-$k$ error graphs, such that $\str(\Pc_i,G_i) \le k$.
    By \cref{lem:rank-contraction-refinement}, one can obtain from this a new contraction sequence $\Pc'_n,\dots,\Pc'_1$,
    such that each~$\Pc'_i$ is obtained from some~$\Pc_j$ by splitting each part into~$2^{2k+1}$ smaller parts,
    and the red graph $\redG(G/\Pc'_i)$ is a subgraph of the corresponding $2^{2k+1}$-blowup of~$G_i$.
    By \cref{lem:stretch-blowup}, this implies that
    \begin{equation}\label{eq:stretch-blowup}
        \str(\Pc'_i,\redG(G/\Pc'_i)) \le k \cdot 2^{2k+1}.
    \end{equation}
    Since \eqref{eq:stretch-blowup} refers to the `usual' notion of stretch, for the red graph of~$\Pc'_i$,
    this shows that~$\Pc'_n,\dots,\Pc'_1$ is a contraction sequence witnessing $\stw(G,\psi) \le k \cdot 2^{2k+1}$.
\end{proof}

Next, let us observe that bounded stretch-width is stable under copying.
\begin{lemma}\label{lem:stw-copying}
    For any graph~$G$ and $\gamma \in \Nn$, the $\gamma$-copy $G_\gamma := K_\gamma \square G$ satisfies 
    \[ \stw(G_\gamma) \le 2\gamma (\stw(G)+1). \]
\end{lemma}
\begin{proof}
    Say $\stw(G) = k$, and consider~$\psi$ a vertex ordering of~$G$ witnessing this stretch-width.
    We extend~$\psi$ to a vertex ordering~$\psi'$ of~$G_\gamma$, keeping the copies $v_1,\dots,v_\gamma$ of each vertex $v \in V(G)$ consecutive in~$\psi'$.
    Given a contraction sequence $\Pc_n,\dots,\Pc_1$ for~$G$ witnessing $\stw(G,\psi) = k$,
    define the partition~$\Pc'_i$ of~$V(G_\gamma)$ by replicating~$\Pc_i$ in each copy of~$G$, while keeping vertices from distinct copies in distinct parts (cf.\ the proof of \cref{lem:copying}).
    Thus, each part $X \in \Pc_i$ becomes~$\gamma$ parts $X_1,\dots,X_\gamma \in V(G)$, with~$X_i$ being in the $i$-th copy of~$G$.
    Now one may check that if~$X_i$ interferes with~$Y_j$ in~$(G_\gamma,\psi')$, then~$X$ interferes with~$Y$ in~$(G,\psi')$, or $X=Y$.
    This implies that~$\Pc'_i$ has stretch at most~$(k+1)\gamma$.

    To go from~$\Pc'_{i+1}$ to~$\Pc'_i$, one merges~$\gamma$ disjoint pairs of parts simultaneously.
    Thus $\Pc'_n,\dots,\Pc'_1$ is not quite a contraction sequence, but it can, as in \cref{lem:copying}, be turned into one:
    for each~$i$, add intermediate steps to perform the~$\gamma$ merges between~$\Pc'_{i+1}$ and~$\Pc'_i$ in any order, as well as~$\gamma-1$ steps at the end to merge the~$\gamma$ parts of~$\Pc'_1$ into a single one.
    These additional steps between~$\Pc'_{i+1}$ and~$\Pc'_i$ have at worse double the stretch compared to~$\Pc'_i$, and the ones at the end have stretch trivially bounded by~$\gamma$ since there are at most~$\gamma$.
    This yields a contraction sequence witnessing $\stw(G_\gamma,\psi') \le 2\gamma (k+1)$.
\end{proof}

We can now prove that stretch-width is stable under transductions:
\cref{lem:stw-copying} handles the initial copy operation;
\cref{lem:error-graph-transduction,lem:stretch-power} together imply that for a fixed interpretation~$\phi$, if~$G$ has bounded stretch-width, then~$\phi(G)$ has bounded rank-stretch-width;
and finally \cref{lem:rank-stretch-width} shows that bounded rank-stretch-width and bounded stretch-width are equivalent.
We omit the details, as they are largely the same as in the proof of \cref{thm:transduction}
\begin{theorem}\label{thm:stw-transduction}
    For any first-order transduction~$\Phi$, there is a function $f:\Nn \to \Nn$ such that any graph~$G$ and any output $H \in \Phi(G)$ satisfy $\stw(H) \le f(\stw(G))$.
\end{theorem}

\section{Algorithmic applications of reduced component max-leaf}\label{sec:algoapplication}

In this section, we show that \textsc{Maximum Independent Set}, \textsc{Maximum Induced Matching},  and \textsc{Maximum Induced $d$-Regular Subgraph} problems can be solved in polynomial time on classes of bounded reduced component max-leaf, when a~contraction sequence witnessing low $\compmaxleaf^\downarrow$ is given. Instead of dealing with each problem separately, we consider a general problem, called \textsc{$\sigma$-Neighborhood} for $\sigma \subseteq \mathbb{N}$, which asks for a maximum-size vertex subset $S$ of a graph $G$ such that, for each $v \in S$, the number of neighbors of $v$ in $S$ lies in $\sigma$. In Subsection~\ref{subsec:sigmaneighbor}, we prove that if $\sigma$ is finite, then \textsc{$\sigma$-Neighborhood} can be solved in polynomial time. As reduced component max-leaf does not change when taking the complement of a graph, \textsc{Maximum Clique} is also polynomial time solvable on these classes.

In Subsection~\ref{subsec:IDP}, we show that  \textsc{Induced Disjoint Paths} can be solved in polynomial time on classes of bounded reduced component max-leaf, when a~contraction sequence witnessing low $\compmaxleaf^\downarrow$ is given.

We describe a general strategy for our algorithms. Let $G$ be a graph and $\mathcal{P}_n, \ldots, \mathcal{P}_1$ be a given contraction sequence of $G$, where each trigraph $G/\mathcal{P}_i$ has component max-leaf~$t$. For each connected red-induced subtrigraph $H$ of $G/\mathcal{P}_i$, we recursively store partial solutions on $\bigcup_{P \in V(H)} P$ with respect to certain conditions on the red neighbors of $H$ in $G/\mathcal{P}_i$.
By~\Cref{lem:numberofinducedsubgraphs}, the number of connected red-induced subtrigraphs of $G/\mathcal{P}_i$ is bounded by $n^{t}$. It will imply that the number of tables to consider is polynomial in~$n$, when $t$ is fixed.

The following lemmas are commonly used to transfer the information from $G/\mathcal{P}_i$ to $G/\mathcal{P}_{i+1}$. See \Cref{fig:PstarinH,fig:PstaroutH} for illustrations of the setting of the lemmas.

\begin{lemma}\label{lem:contraction}
Let $G$ be a graph and let $\mathcal{P}$ be a partition of $V(G)$.
Let $P_1$ and $P_2$ be two parts of $\mathcal{P}$, and 
let $\mathcal{P}'$ be the partition obtained from $\mathcal{P}$ by merging $P_1$ and $P_2$ into $P^*$. 

Let $H$ be a non-empty connected red-induced subtrigraph of $G/\mathcal{P}'$ where $P^* \in V(H) \cup \big(R_{G/\mathcal{P}'}(V(H)) \setminus B_{G/\mathcal{P}'}(V(H))\big)$.  
For each $j \in [2]$, let $\mathcal{Q}_j = B_{G/\mathcal{P}}(P_j) \cap (V(H)\setminus \{P^*\})$. Then the following hold. 
\begin{enumerate}[(1)]
    \item For each $j\in [2]$, $\mathcal{Q}_j\cap B_{G/\mathcal{P}}(P_{3-j})=\emptyset$.
    \item  $\mathcal{Q}_1$ and $\mathcal{Q}_2$ are disjoint.
    \item 
    Every connected subtrigraph of $H - \{P^*\}$ is a connected red-induced subtrigraph of $\mathcal{R}(G/\mathcal{P})$.
\end{enumerate}
    
\end{lemma}
\begin{proof}
    (1) 
    We claim that there are no black edges between $P^*$ and $V(H) \setminus \{P^*\}$ in $G/\mathcal{P}'$. 
    If $P^*\in V(H)$, then this is true because $H$ has no black edges.
    If $P^*\in (R_{G/\mathcal{P}'}(V(H)) \setminus B_{G/\mathcal{P}'}(V(H)))$, then $P^*$ has no black neighbors on $H$ in $G/\mathcal{P}'$.
    Thus, the claim holds.

    If there is a part $Q$ in $\mathcal{Q}_j\cap B_{G/\mathcal{P}}(P_{3-j})$, then $Q$ is a black common neighbor of $P_1$ and $P_2$.
    Thus, 
    $Q$ and $P^*$ are linked by a black edge in $G/\mathcal{P}'$. As $Q\in V(H)\setminus \{P^*\}$, this is a contradiction. 
    
    \medskip
    (2) This follows from (1).

    \medskip
    (3) This holds because 
    $G/\mathcal{P} - \{P_1,P_2\} = G/\mathcal{P}' - \{P^*\}$.
 \end{proof}

\begin{lemma}\label{lem:componentcount}
    Let $t$ be a positive integer. Let $G$ be a graph and let $\mathcal{P}$ be a partition of $V(G)$. Let $P_1$ and $P_2$ be two parts of $\mathcal{P}$, and let $\mathcal{P}'$ be the partition obtained from $\mathcal{P}$ by merging $P_1$ and $P_2$ into $P^*$. Let $H$ be a non-empty connected red-induced subtrigraph of $G/\mathcal{P}'$ where $P^* \in V(H) \cup \big(R_{G/\mathcal{P}'}(V(H)) \setminus B_{G/\mathcal{P}'}(V(H))\big)$ and $(G/\mathcal{P}')[V(H)\cup \{P^*\}]$ has max-leaf at most $t$. For each $j \in [2]$, let $\mathcal{Q}_j = B_{G/\mathcal{P}}(P_j) \cap (V(H)\setminus \{P^*\})$. Let $F$ be the subtrigraph of $G/\mathcal{P}$ induced by $(V(H)\setminus\{P^*\})\cup \{P_1, P_2\}$. Let $\mathcal{O} \subseteq \{P_1, P_2\} \cup \mathcal{Q}_1 \cup \mathcal{Q}_2$ such that for each $j \in [2]$ the following hold:
    \begin{itemize}
        \item If $\mathcal{Q}_j \cap \mathcal{O} \neq \emptyset$, then either $\mathcal{Q}_j \subseteq \mathcal{O}$ or $\{P_j\}\cup \mathcal{Q}_j \subseteq \mathcal{O}$.
        \item If $\mathcal{Q}_j\setminus \mathcal{O} \neq \emptyset$, then $P_j \in \mathcal{O}$.
    \end{itemize} 
    Then $F - \mathcal{O}$ has at most $t+2$ connected components.
\end{lemma}

\begin{proof}
    Since $H$ is connected, $F$ consists of at most two connected components. 
    If $\mathcal{O} = \emptyset$, then $F - \mathcal{O}=F$ has at most two components and we are done. So, we may assume that $\mathcal{O} \neq \emptyset$.

    We first consider the case that $\{P_1, P_2\} \subseteq \mathcal{O}$. As $\mathcal{P} \setminus \{P_1, P_2\} = \mathcal{P}' \setminus \{P^*\}$, we have 
    \[F - \mathcal{O} = (G/\mathcal{P}')[V(H)\cup \{P^*\}] - \big((\mathcal{O}\setminus \{P_1, P_2\}) \cup \{P^*\}\big).\] 
    By the definition of $\mathcal{Q}_j$, the set $(\mathcal{O}\setminus \{P_1, P_2\}) \cup \{P^*\}$ induces a connected red-induced subtrigraph of $G\setminus \mathcal{P}'$. As $(G/\mathcal{P}')[V(H)\cup \{P^*\}]$ has max-leaf at most $t$, by~\Cref{lem:thenumberofredneighbor}, $F - \mathcal{O}$ has at most $t$ connected components.

    So, we may assume that $P_j \notin \mathcal{O}$ for some $j \in [2]$. Let $X = \{P_j \in \{P_1, P_2\} : P_j \notin \mathcal{O}\}$. For every $P_j \in X$, since $P_j \notin \mathcal{O}$, we have either $\mathcal{Q}_j = \emptyset$ or $\mathcal{Q}_j \subseteq \mathcal{O}$. Let $\mathcal{O}' = (\mathcal{O} \setminus \{P_1, P_2\})\cup \{P^*\}$. Then 
    \[F - \mathcal{O}=(G/\mathcal{P})[(V(H)\setminus \mathcal{O}')\cup X\}].\]
    By~\Cref{lem:thenumberofredneighbor}, $H - \mathcal{O}'=(G/\mathcal{P}')[V(H)\cup \{P^*\}]-\mathcal{O}'$ has at most $t$ connected components. Since $|X| \le 2$, it follows that $F - \mathcal{O}$ has at most $t+2$ connected components.

\end{proof}

We use the following notations.
Let $G$ be a graph and $\mathcal{P}$ be a vertex partition of $G$. For a subgraph $H$ of $G$, let $\trace_{G,\mathcal{P}}(H)$ denote the set of all parts in $\mathcal{P}$ that contain a vertex of $H$. 
If $G$ is clear from the context, we simply write it as $\trace_{\mathcal{P}}(H)$.
  For a family $\mathcal{Q}$ of vertex sets in $G$ and $I\subseteq V(G)$, we write $I_{\mathcal{Q}}=I\cap \left(\bigcup_{X\in \mathcal{Q}}X\right)$.

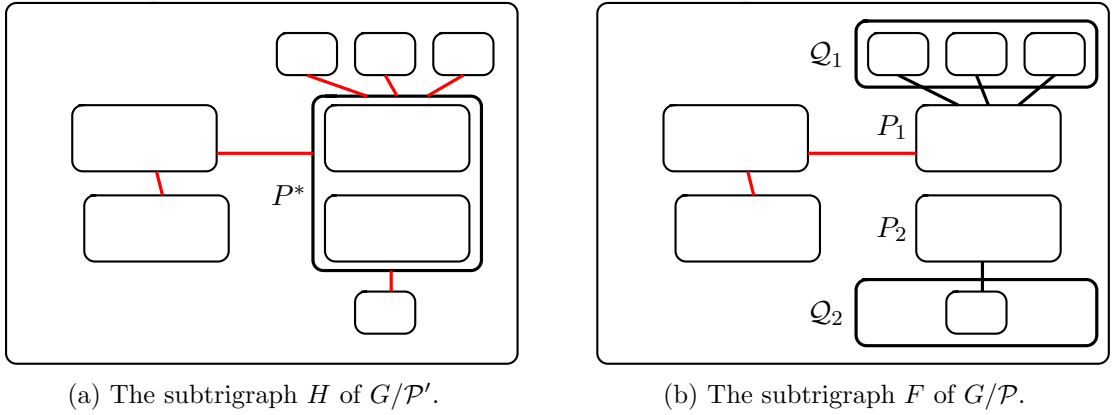
\begin{figure}
    \centering
    \begin{subfigure}{0.4\linewidth}
        \centering
           \begin{tikzpicture}[scale=0.8]
        \tikzstyle{v}=[circle, draw, solid, fill=black, inner sep=0pt, minimum width=3pt]
        \tikzset{c1/.style={purple, line width=6pt,opacity=0.5,line cap=round,shorten >=-2pt, shorten <=-2pt}}

        \draw[rounded corners, thick] (0.5,5.5)--(-2,5.5)--(-2,-0.5)--(6.5,-0.5)--(6.5,5.5)--(0.4,5.5); 

        \draw[rounded corners, thick] (3.5,2.3)--(3.3,2.3)--(3.3,1.2)--(5.7,1.2)--(5.7,2.3)--(3.4,2.3); 
        \draw[rounded corners, thick] (3.5,3.8)--(3.3,3.8)--(3.3,2.7)--(5.7,2.7)--(5.7,3.8)--(3.4,3.8); 

        \draw[rounded corners, very thick] 
        (3.5,4.1-0.15)--(3.1,4.1-0.15)--(3.1,0.9+0.15)--(5.9,0.9+0.15)--(5.9,4.1-0.15)--(3.3,4.1-0.15);

    \draw[rounded corners, thick] (2.7,5)--(2.5,5)--(2.5,4.3)--(3.5,4.3)--(3.5,5)--(2.6,5);
    \draw[rounded corners, thick] (2.7+1.3,5)--(2.5+1.3,5)--(2.5+1.3,4.3)--(3.5+1.3,4.3)--(3.5+1.3,5)--(2.6+1.3,5);
    \draw[rounded corners, thick] (2.7+2.6,5)--(2.5+2.6,5)--(2.5+2.6,4.3)--(3.5+2.6,4.3)--(3.5+2.6,5)--(2.6+2.6,5); 
    \draw[rounded corners, thick] (2.7+1.3,.7)--(2.5+1.3,.7)--(2.5+1.3,0)--(3.5+1.3,0)--(3.5+1.3,.7)--(2.6+1.3,.7); 

     \draw[rounded corners, thick] (3.5-4,2.3)--(3.3-4,2.3)--(3.3-4,1.2)--(5.7-4,1.2)--(5.7-4,2.3)--(3.4-4,2.3); 
        \draw[rounded corners, thick] (3.5-4.2,3.8)--(3.3-4.2,3.8)--(3.3-4.2,2.7)--(5.7-4.2,2.7)--(5.7-4.2,3.8)--(3.4-4.2,3.8); 

   \draw[red,very thick] (3,4.3)--(4.5-0.5,4.1-0.15);
   \draw[red,very thick] (3+1.3,4.3)--(4.5,4.1-0.15);
   \draw[red,very thick] (3+2.6,4.3)--(4.5+0.5,4.1-0.15);

   \draw[red,very thick] (3+1.4,0.9+0.15)--(4.4,0.7);

   \draw[red, very thick] (3.1,3)--(1.5,3);
   \draw[red, very thick] (4.6-4,2.3)--(4.5-4,2.7);

   \node [label=$P^*$] (v) at (2.7, 1.8){};

    \end{tikzpicture}
    \caption{The subtrigraph $H$ of $G/\mathcal{P}'$.}
    \end{subfigure}
    \hspace{1cm}
    \begin{subfigure}{0.4\linewidth}
        \centering
        \begin{tikzpicture}[scale=0.8]
        \tikzstyle{v}=[circle, draw, solid, fill=black, inner sep=0pt, minimum width=3pt]
        \tikzset{c1/.style={purple, line width=6pt,opacity=0.5,line cap=round,shorten >=-2pt, shorten <=-2pt}}

        \draw[rounded corners, thick] (0.5,5.5)--(-2,5.5)--(-2,-0.5)--(6.5,-0.5)--(6.5,5.5)--(0.4,5.5); 

        \draw[rounded corners, thick] (3.5,2.3)--(3.3,2.3)--(3.3,1.2)--(5.7,1.2)--(5.7,2.3)--(3.4,2.3); 
        \draw[rounded corners, thick] (3.5,3.8)--(3.3,3.8)--(3.3,2.7)--(5.7,2.7)--(5.7,3.8)--(3.4,3.8); 
   
        \draw[rounded corners, very thick] (2.5,5.2)--(2.3,5.2)--(2.3,4.1)--(6.3,4.1)--(6.3,5.2)--(2.4,5.2);
        \draw[rounded corners, very thick] (2.5,.9)--(2.3,.9)--(2.3,-.2)--(6.3,-.2)--(6.3,.9)--(2.4,.9); 
   
    \draw[rounded corners, thick] (2.7,5)--(2.5,5)--(2.5,4.3)--(3.5,4.3)--(3.5,5)--(2.6,5);
    \draw[rounded corners, thick] (2.7+1.3,5)--(2.5+1.3,5)--(2.5+1.3,4.3)--(3.5+1.3,4.3)--(3.5+1.3,5)--(2.6+1.3,5);
    \draw[rounded corners, thick] (2.7+2.6,5)--(2.5+2.6,5)--(2.5+2.6,4.3)--(3.5+2.6,4.3)--(3.5+2.6,5)--(2.6+2.6,5); 
    \draw[rounded corners, thick] (2.7+1.3,.7)--(2.5+1.3,.7)--(2.5+1.3,0)--(3.5+1.3,0)--(3.5+1.3,.7)--(2.6+1.3,.7); 

     \draw[rounded corners, thick] (3.5-4,2.3)--(3.3-4,2.3)--(3.3-4,1.2)--(5.7-4,1.2)--(5.7-4,2.3)--(3.4-4,2.3); 
        \draw[rounded corners, thick] (3.5-4.2,3.8)--(3.3-4.2,3.8)--(3.3-4.2,2.7)--(5.7-4.2,2.7)--(5.7-4.2,3.8)--(3.4-4.2,3.8); 

   \draw[very thick] (3,4.3)--(4.5-0.5,3.8);
   \draw[very thick] (3+1.3,4.3)--(4.5,3.8);
   \draw[very thick] (3+2.6,4.3)--(4.5+0.5,3.8);

   \draw[very thick] (3+1.4,1.2)--(4.4,0.7);

   \draw[red, very thick] (3.3,3)--(1.5,3);
   \draw[red, very thick] (4.6-4,2.3)--(4.5-4,2.7);

    \node [label=$P_2$] (v) at (2.9, 1.2){};
   \node [label=$P_1$] (v) at (2.9, 2.9){};
   \node [label=$\mathcal{Q}_1$] (v) at (1.8, 4.1){};
   \node [label=$\mathcal{Q}_2$] (v) at (1.8, -.2){};
   
    \end{tikzpicture}      
    \caption{The subtrigraph $F$ of $G/\mathcal{P}$.}
    \end{subfigure}

    \caption{The partition $\mathcal{P}'$ is obtained from a partition $\mathcal{P}$ by merging $P_1$ and $P_2$ into $P^*$, and $H$ is a connected red-induced subtrigraph of $G/\mathcal{P}'$ containing $P^*$.  The trigraph $F$ is the subtrigraph of $G/\mathcal{P}$ induced by $(V(H)\setminus \{P^*\})\cup \{P_1, P_2\}$. Each part in $\mathcal{Q}_1\cup \mathcal{Q}_2$ is adjacent to one of  $P_1$ and $P_2$ by a black edge, but become adjacent to $P^*$ by a red edge. }
    \label{fig:PstarinH}
\end{figure}
\begin{figure}
    \centering
    \begin{subfigure}{0.4\linewidth}
        \centering
           \begin{tikzpicture}[scale=0.8]
        \tikzstyle{v}=[circle, draw, solid, fill=black, inner sep=0pt, minimum width=3pt]
        \tikzset{c1/.style={purple, line width=6pt,opacity=0.5,line cap=round,shorten >=-2pt, shorten <=-2pt}}

        \draw[rounded corners, thick] (0.5,5.5)--(-2,5.5)--(-2,-0.5)--(4.3,-0.5)--(4.3,5.5)--(0.4,5.5); 

        \draw[rounded corners, thick] (3.5+1.5,2.3)--(3.3+1.5,2.3)--(3.3+1.5,1.2)--(5.7+1.5,1.2)--(5.7+1.5,2.3)--(3.4+1.5,2.3); 
        \draw[rounded corners, thick] (3.5+1.5,3.8)--(3.3+1.5,3.8)--(3.3+1.5,2.7)--(5.7+1.5,2.7)--(5.7+1.5,3.8)--(3.4+1.5,3.8); 

        \draw[rounded corners, very thick] 
        (3.5+1.5,4.1-0.15)--(3.1+1.5,4.1-0.15)--(3.1+1.5,0.9+0.15)--(5.9+1.5,0.9+0.15)--(5.9+1.5,4.1-0.15)--(3.3+1.5,4.1-0.15);

    \draw[rounded corners, thick] (2.7+1.3-2,5)--(2.5+1.3-2,5)--(2.5+1.3-2,4.3)--(3.5+1.3-2,4.3)--(3.5+1.3-2,5)--(2.6+1.3-2,5);
    \draw[rounded corners, thick] (2.7+2.6-2,5)--(2.5+2.6-2,5)--(2.5+2.6-2,4.3)--(3.5+2.6-2,4.3)--(3.5+2.6-2,5)--(2.6+2.6-2,5); 
    \draw[rounded corners, thick] (2.7+1.3-2,.7)--(2.5+1.3-2,.7)--(2.5+1.3-2,0)--(3.5+1.3-2,0)--(3.5+1.3-2,.7)--(2.6+1.3-2,.7); 

     \draw[rounded corners, thick] (3.5-4,2.3)--(3.3-4,2.3)--(3.3-4,1.2)--(5.7-4,1.2)--(5.7-4,2.3)--(3.4-4,2.3); 
        \draw[rounded corners, thick] (3.5-4.2,3.8)--(3.3-4.2,3.8)--(3.3-4.2,2.7)--(5.7-4.2,2.7)--(5.7-4.2,3.8)--(3.4-4.2,3.8);

   \draw[red,very thick] (3+1.3-2,4.3)--(4.6,4.1-0.15-0.6);
   \draw[red,very thick] (3+2.6-2,4.3)--(4.6,4.1-0.15-0.3);

   \draw[red,very thick] (3+1.3-2,0.7)--(4.6,1.5);

   \draw[red,very thick] (4.6-4,2.3)--(4.5-4,2.7);

   \node [label=$P^*$] (v) at (6.1, 3.9){};

    \end{tikzpicture}
    \caption{The subtrigraph $H$ of $G/\mathcal{P}'$.}
    \end{subfigure}
    \quad
    \hspace{1cm}
    \begin{subfigure}{0.4\linewidth}
        \centering
                  \begin{tikzpicture}[scale=0.8]
        \tikzstyle{v}=[circle, draw, solid, fill=black, inner sep=0pt, minimum width=3pt]
        \tikzset{c1/.style={purple, line width=6pt,opacity=0.5,line cap=round,shorten >=-2pt, shorten <=-2pt}}

        \draw[rounded corners, thick] (0.5,5.5)--(-2,5.5)--(-2,-0.5)--(4.4,-0.5)--(4.4,5.5)--(0.4,5.5); 

        \draw[rounded corners, thick] (3.5+1.5,2.3)--(3.3+1.5,2.3)--(3.3+1.5,1.2)--(5.7+1.5,1.2)--(5.7+1.5,2.3)--(3.4+1.5,2.3); 
        \draw[rounded corners, thick] (3.5+1.5,3.8)--(3.3+1.5,3.8)--(3.3+1.5,2.7)--(5.7+1.5,2.7)--(5.7+1.5,3.8)--(3.4+1.5,3.8);

        \draw[rounded corners, very thick] (2.5,5.2)--(1.6,5.2)--(1.6,4.1)--(4.25,4.1)--(4.25,5.2)--(2.4,5.2);
        \draw[rounded corners, very thick] (2.5,.9)--(1.6,.9)--(1.6,-.2)--(4.25,-.2)--(4.25,.9)--(2.4,.9); 
   
    \draw[rounded corners, thick] (2.7+1.3-2,5)--(2.5+1.3-2,5)--(2.5+1.3-2,4.3)--(3.5+1.3-2,4.3)--(3.5+1.3-2,5)--(2.6+1.3-2,5);
    \draw[rounded corners, thick] (2.7+2.6-2,5)--(2.5+2.6-2,5)--(2.5+2.6-2,4.3)--(3.5+2.6-2,4.3)--(3.5+2.6-2,5)--(2.6+2.6-2,5); 
    \draw[rounded corners, thick] (2.7+1.3-2,.7)--(2.5+1.3-2,.7)--(2.5+1.3-2,0)--(3.5+1.3-2,0)--(3.5+1.3-2,.7)--(2.6+1.3-2,.7); 

     \draw[rounded corners, thick] (3.5-4,2.3)--(3.3-4,2.3)--(3.3-4,1.2)--(5.7-4,1.2)--(5.7-4,2.3)--(3.4-4,2.3); 
        \draw[rounded corners, thick] (3.5-4.2,3.8)--(3.3-4.2,3.8)--(3.3-4.2,2.7)--(5.7-4.2,2.7)--(5.7-4.2,3.8)--(3.4-4.2,3.8);

   \draw[very thick] (3+1.3-2,4.3)--(4.8,4.1-0.15-0.6);
   \draw[very thick] (3+2.6-2,4.3)--(4.8,4.1-0.15-0.3);

   \draw[very thick] (3+1.3-2,0.7)--(4.8,1.5);

  \draw[red,very thick] (4.6-4,2.3)--(4.5-4,2.7);

   \node [label=$P_1$] (v) at (6.1, 3.7){};
   \node [label=$P_2$] (v) at (6.1, .2){};
  \node [label=$\mathcal{Q}_1$] (v) at (1.2, 4.1){};
   \node [label=$\mathcal{Q}_2$] (v) at (1.2, -.2){};
   
    \end{tikzpicture}
    \caption{$H$ in $G/\mathcal{P}$.}
    \end{subfigure}
    \caption{The partition $\mathcal{P}'$ is obtained from a partition $\mathcal{P}$ by merging $P_1$ and $P_2$ into $P^*$, and $H$ is a connected red-induced subtrigraph of $G/\mathcal{P}'$ not containing $P^*$.  It is possible that in $G/\mathcal{P}$, $P_1$ and $P_2$ have no red edges to $H$, but $P^*$ is adjacent to $H$ by red edges in $G/\mathcal{P}'$. }
    \label{fig:PstaroutH}
\end{figure}
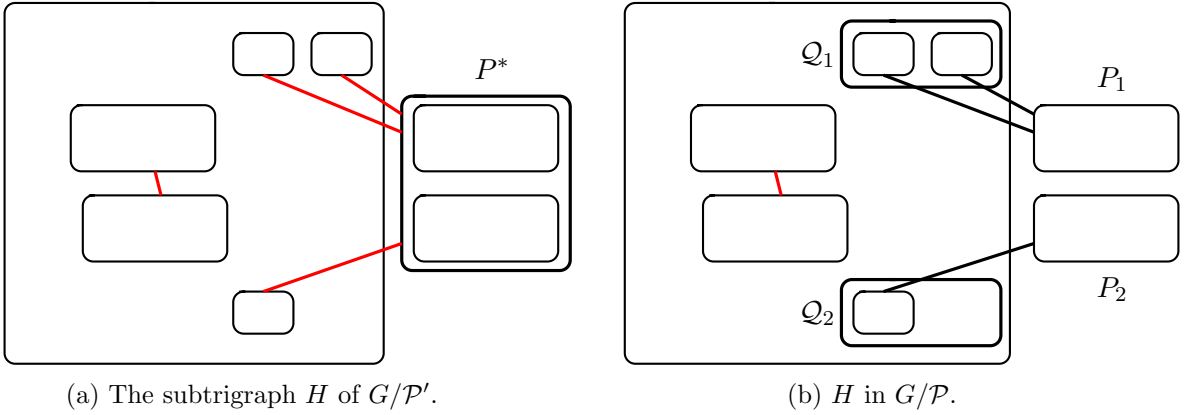

\subsection{\textsc{$\sigma$-Neighborhood} problem}\label{subsec:sigmaneighbor}

We introduce the \textsc{$\sigma$-Neighborhood} problem.

\noindent
\fbox{\parbox{0.97\textwidth}{
	\textsc{$\sigma$-Neighborhood} 
    ($\sigma$ is a finite subset of $\mathbb{N}\cup \{0\}$ with $\alpha = \max(\sigma)$)\\
	\textbf{Input:} A graph $G$.
	\\
	\textbf{Output:} A subset $S \subseteq V(G)$ of maximum size such that $|N_G(v) \cap S| \in \sigma$ for every $v \in S$.}}
\medskip
  
  Let $P_1$ and $P_2$ be two parts of a~partition $\mathcal{P}$ of a vertex set of a graph $G$, and for each $i\in [2]$, let
  $\mathcal{Q}_i\subseteq B_{G/\mathcal{P}}(P_i)\setminus B_{G/\mathcal{P}}(P_{3-i}).$
  We say that a pair $(\mathcal{O}, U)$ of a subset $\mathcal{O}\subseteq \{P_1, P_2\}\cup \mathcal{Q}_1\cup \mathcal{Q}_2$ and a vertex subset $U\subseteq \bigcup_{Q\in \mathcal{O}}Q$ is a~\emph{$\sigma$-separator} with respect to $(P_1, P_2, \mathcal{Q}_1, \mathcal{Q}_2)$ if the following conditions hold:
\begin{itemize}
    \item $U$ contains at most $\alpha$ vertices from each part in $\mathcal{O}$,
    \item if $\mathcal{O} \cap \mathcal{Q}_i \neq \emptyset$ for some $i$, then $\{P_i\}\cup \mathcal{Q}_i\subseteq \mathcal{O}$, 
    $|U_{\mathcal{Q}_i}|\le \alpha$, and
    $|U\cap Q|\ge 1$ for every $Q\in \mathcal{Q}_i$,
    \item $G[U]$ has maximum degree $\alpha$, and 
    
    \item for every part $Q\in \{P_1, P_2\}\cup \mathcal{Q}_1\cup \mathcal{Q}_2$ that is not in $\mathcal{O}$ and every part $Q'\in \{P_1, P_2\}\cup \mathcal{Q}_1\cup \mathcal{Q}_2$ where $Q$ and $Q'$ are linked by a black edge in $G/\mathcal{P}$, it holds that $Q'\in \mathcal{O}$ and $U \cap Q'=\emptyset$.
\end{itemize}
The last condition will be referred to as the `black pair condition'. 

\begin{lemma}\label{lem:numberofseparators}
    Let $G$ be a graph and let $\mathcal{P}$ be a partition of $V(G)$.
    Let $P_1$ and $P_2$ be two parts of $\mathcal{P}$, and for each $i\in [2]$, let $\mathcal{Q}_i\subseteq B_{G/\mathcal{P}}(P_i)\setminus B_{G/\mathcal{P}}(P_{3-i}).$ Then the number of $\sigma$-separators with respect to $(P_1, P_2, \mathcal{Q}_1, \mathcal{Q}_2)$ is at most $4\cdot|V(G)|^{4\alpha}$.
\end{lemma}
\begin{proof}
    Let $(\mathcal{O}, U)$ be a $\sigma$-separator with respect to $(P_1, P_2, \mathcal{Q}_1, \mathcal{Q}_2)$. By the first and second conditions, $|U \cap P_i|, |U_{\mathcal{Q}_i}| \le \alpha$ for each $i \in [2]$. It follows that $|U| \le 4\alpha$, and hence there are at most $|V(G)|^{4\alpha}$ choices for $U$. If $U \cap P_i = \emptyset$ for some $i \in [2]$, then both cases $P_i \in \mathcal{O}$ and $P_i \notin \mathcal{O}$ are possible. Therefore, the number of $\sigma$-separators is at most $4\cdot|V(G)|^{4\alpha}$.
\end{proof}

Briefly speaking, $\sigma$-separators capture the possible positions that contain a bounded number of vertices from a partial solution.
We explain how this concept will be used and give an outline of the proof. Let $\mathcal{P}$ be a partition of $V(G)$ and $\mathcal{P}'$ be the partition of $V(G)$ obtained from $\mathcal{P}$ by merging two parts $P_1$ and $P_2$ into $P^*$. In the algorithm we consider all connected red-induced subtrigraphs $H$ of $G/\mathcal{P}'$ satisfying certain conditions on red, but not black, neighbors and store maximum partial solutions on $\bigcup_{P\in V(H)}P$ with respect to them. When we separate using a $\sigma$-separator, the black pair condition ensures that a similar condition holds again, since if a part has a black neighbor, then it does not contain a partial solution.

We distinguish two cases depending on whether $P^*\in V(H)$. 

\begin{itemize}
    \item ($P^*\in V(H)$): In this case, we consider the subtrigraph $F$ of $G/\mathcal{P}$ on $(V(H)\setminus \{P^*\})\cup \{P_1, P_2\}$. See \Cref{fig:PstarinH} for an illustration. Note that $F$ may have black edges, and these edges appear between $P_j$ and its neighbors that are not adjacent to $P_{3-j}$ by a black edge. Also, $P_1$ and $P_2$ may be adjacent by a black edge. We consider 
    \[\mathcal{Q}_j = B_{G/\mathcal{P}}(P_j) \cap (V(H)\setminus \{P^*\})\]
    Now, we want to guess some partial solutions on $\{P_1, P_2\}\cup \mathcal{Q}_1\cup \mathcal{Q}_2$ and then divide $F$ into smaller red-induced subtrigraphs. The $\sigma$-separators with respect to $(P_1, P_2, \mathcal{Q}_1, \mathcal{Q}_2)$ are possible candidates of subsets of $\{P_1, P_2\}\cup \mathcal{Q}_1\cup \mathcal{Q}_2$ that contain at most $4\alpha$ vertices from a solution and meet all black edges. We enumerate all possible partial solutions and divide $F$ into small red-induced subtrigraphs, from which we can obtain the information from $\mathcal{P}$.   
    \item ($P^*\notin V(H)$): Here it is natural to consider the same $H$ for $G/\mathcal{P}$. See \Cref{fig:PstaroutH} for an illustration. As in the figure, $P^*$ may be a new red neighbor of $H$. In this situation, we consider the sets $\mathcal{Q}_1$ and $\mathcal{Q}_2$ as defined above. We will consider all $\sigma$-separators and guess partial solutions whose restriction on $P_1\cup P_2$ matches the given condition on $P^*$, and then do the similar arguments as in the previous case.
\end{itemize}
The following lemma guarantees the existence of a $\sigma$-separator, when there is a partial solution for~$G/\mathcal{P}'$.
\begin{lemma}\label{lem:sigmaseparator}
    Let $G$ be a graph and let $\mathcal{P}$ be a partition of $V(G)$.
    Let $P_1$ and $P_2$ be two parts of $\mathcal{P}$, and for each $i\in [2]$, let $\mathcal{Q}_i\subseteq B_{G/\mathcal{P}}(P_i)\setminus B_{G/\mathcal{P}}(P_{3-i}).$ 
    Let $\alpha$ be a non-negative integer, and let $I\subseteq V(G)$ such that $G[I]$ has maximum degree at most $\alpha$, and $I\cap X\neq\emptyset$ for each $X\in \{P_1\cup P_2\}\cup \mathcal{Q}_1\cup \mathcal{Q}_2$. Then the following hold.
    \begin{enumerate}[(1)]
        \item\label{case:sigmaseparator1} There exists a subset $\mathcal{O}\subseteq \{P_1, P_2\}\cup \mathcal{Q}_1\cup \mathcal{Q}_2$ such that $(\mathcal{O}, I_{\mathcal{O}})$ is a $\sigma$-separator with respect to $(P_1, P_2, \mathcal{Q}_1, \mathcal{Q}_2)$.
        \item\label{case:sigmaseparator2} If $|I\cap (P_1\cup P_2)|\le \alpha$, then there exists a subset $\mathcal{O}\subseteq \{P_1, P_2\}\cup \mathcal{Q}_1\cup \mathcal{Q}_2$ such that $\{P_1, P_2\}\subseteq \mathcal{O}$ and $(\mathcal{O}, I_{\mathcal{O}})$ is a $\sigma$-separator with respect to $(P_1, P_2, \mathcal{Q}_1, \mathcal{Q}_2)$.
    \end{enumerate} 
\end{lemma}
\begin{proof}
    We prove (1) and (2) together.
    
    Assume that $I \cap P_j = \emptyset$ for some $j \in [2]$. Since $I\cap (P_1\cup P_2)\neq \emptyset$, $I \cap P_{3-j} \neq \emptyset$. 
        As $P_{3-j}$ is adjacent to each part of $\mathcal{Q}_{3-j}$ by a black edge in $G/\mathcal{P}$ and $G[I]$ has maximum degree at most $\alpha$, we have $|I_{\mathcal{Q}_{3-j}}| \le \alpha$. If $|\mathcal{Q}_{3-j}| = 0$, then 
        \begin{itemize}
            \item $(\{P_j\}, \emptyset)$ is a $\sigma$-separator, and
            \item moreover, if $|I\cap (P_1\cup P_2)|\le \alpha$, then $(\mathcal{O}=\{P_1, P_2\}, I_{\mathcal{O}})$ is a $\sigma$-separator.
        \end{itemize}  
        Otherwise $|I\cap P_{3-j}|\le \alpha$ as $I_{\mathcal{Q}_{3-j}}$ is non-empty, and 
        $|I_{\mathcal{Q}_{3-j}}|\le \alpha$ as $I\cap P_{3-j}$ is non-empty.
        Thus,
        $(\mathcal{O}=\{P_1, P_2\}\cup\mathcal{Q}_{3-j}, I_{\mathcal{O}})$ is a $\sigma$-separator.

        Now, we assume that $I \cap P_j \neq \emptyset$ for every $j \in [2]$. Note that if $\mathcal{Q}_i$ is non-empty, then by the assumption, $I_{\mathcal{Q}_i}\neq \emptyset$ and $I$ has at most $\alpha$ vertices in $P_i$. If $P_1$ and $P_2$ are connected by a black edge in $G/\mathcal{P}$, then $(\mathcal{O}=\{P_1, P_2\}\cup \mathcal{Q}_1 \cup \mathcal{Q}_2, I_{\mathcal{O}})$ is a $\sigma$-separator. Thus, we may assume that $P_1$ and $P_2$ are not connected by a black edge in $G/\mathcal{P}$. Then we can divide into three cases.
        \begin{itemize}
            \item If $|\mathcal{Q}_1|, |\mathcal{Q}_2| \ge 1$, then $(\mathcal{O}=\{P_1, P_2\}\cup \mathcal{Q}_1 \cup \mathcal{Q}_2, I_{\mathcal{O}})$ is a $\sigma$-separator.
            \item If $|\mathcal{Q}_j| \ge 1$ and $|\mathcal{Q}_{3-j}| = 0$ for some $j\in [2]$, then $(\mathcal{O}=\{P_j\}\cup \mathcal{Q}_j, I_{\mathcal{O}})$ is a $\sigma$-separator.
                Moreover, if $|I \cap (P_1\cup P_2)| \le \alpha$, then $(\mathcal{O}=\{P_1, P_2\}\cup \mathcal{Q}_j, I_{\mathcal{O}})$ is a $\sigma$-separator.
            \item If $|\mathcal{Q}_j| = |\mathcal{Q}_{3-j}| = 0$, then $(\emptyset, \emptyset)$ is a $\sigma$-separator. Moreover, if $|I \cap (P_1\cup P_2)| \le \alpha$, then $(\mathcal{O}=\{P_1, P_2\}, I_{\mathcal{O}})$ is a $\sigma$-separator.

        \end{itemize}
        
        This proves the lemma.
\end{proof}

We present the algorithm in detail.
\begin{theorem}\label{thm:sigmaneighborhood}
 There is an algorithm that, given an $n$-vertex graph $G$ and a~contraction sequence witnessing that $\compmaxleaf^\downarrow(G) \leqslant t$, solves \textsc{$\sigma$-Neighborhood} on~$G$ in time $n^{O(t\alpha^2 )}$. 
\end{theorem}
\begin{proof}
    
    Let $\mathcal{P}_n, \ldots, \mathcal{P}_1$ be a given contraction sequence of $G$ witnessing that $\compmaxleaf^\downarrow(G) \leqslant t$.
    Let $G_i = G / \mathcal{P}_i$ for each $i$. 

    For each $i\in [n]$, let $\mathcal{C}(i)$ be the collection of all tuples $(H, S, \eta)$ such that 
    \begin{itemize}
        \item $H$ is a non-empty connected red-induced subtrigraph of $G_i$,

        \item $S$ is a subset of $V(G)$ such that 
        \begin{itemize}
            \item $|S \cap P| \le \alpha$ for every $P \in R_{G_i}(V(H))\setminus B_{G_i}(V(H))$, and
            \item $|S \cap P| = 0$ for every $P \in V(G_i) \setminus \big(R_{G_i}(V(H))\setminus B_{G_i}(V(H))\big)$,
        \end{itemize}
        \item $\eta: S\to \{0, 1, \ldots, \alpha\}$ is a function.
    \end{itemize}
   For $(H, S, \eta) \in \mathcal{C}(i)$, a vertex subset $I$ of $G$ is \emph{valid} with respect to $(i, H, S, \eta)$ if 
    \begin{itemize}
        \item $\trace_{\mathcal{P}_i}(I) = V(H)$,
        \item for every $v \in I$, $|N_G(v) \cap (I \cup S)| \in \sigma$, and 
        \item for every $v \in S$, $|N_G(v) \cap I| = \eta(v)$.
    \end{itemize}
    Note that $H$ is always non-empty, and by the first condition, any valid set is non-empty. 
    Moreover, from the second condition, $G[I]$ has maximum degree at most $\alpha$.
    Let $\zeta(i,H,S,\eta)$ be the set of all valid vertex sets with respect to $(i,H,S,\eta)$, and let $\zeta_{max}(i,H,S,\eta)$ be the set of all maximum valid vertex sets with respect to $(i,H,S,\eta)$.
    
    We recursively store a maximum valid vertex set $I$ of $G$ with respect to $(i,H,S,\eta)$.
    For $(H,S,\eta)\in \mathcal{C}(i)$, let $\Phi(i,H,S,\eta)$ be a largest valid vertex set.  
    If there is no valid vertex set, then we assign $\Phi(i,H,S,\eta)=\bot$.
    Observe that $\Phi(1, G_1, \emptyset, \eta)$ is the required solution where $\eta$ is the empty function. 
    
     In the base case when $i=n$, $G_n$ has no red edges. 
    Thus, for every $(H,S,\eta)\in \mathcal{C}(n)$, $H$ consists of a single vertex $\{v\}\in \mathcal{P}_n$ and $S = \emptyset$ because $R_{G_n}(V(H))=\emptyset$. If $0 \in \sigma$, then we set  $\Phi(n,H,S,\eta)=\{v\}$ for every $(H,S,\eta)\in \mathcal{C}(n)$. Otherwise, we record $\Phi(n,H,S,\eta) = \bot$ for every $(H,S,\eta)\in \mathcal{C}(n)$.
    
    Henceforth, we assume that $i<n$.
    We further assume that for every $(H,S,\eta)\in \mathcal{C}(i+1)$, $\Phi(i+1,H,S,\eta)$ has been recorded.
    Let $P_1$ and $P_2$ be the parts of $\mathcal{P}_{i+1}$ that are merged into $P^*$ in~$\mathcal{P}_i$.
    Let $(H,S,\eta) \in \mathcal{C}(i)$.

    \medskip
    We first deal with two simpler cases. 
     First assume that \[P^* \in V(G_i) \setminus \big(V(H) \cup (R_{G_i}(V(H)) \setminus B_{G_i}(V(H)))\big).\] 
    As $(H,S,\eta)\in \mathcal{C}(i)$, 
    we have $|S \cap P^*| = 0$ and it follows that $|S \cap P_j| = 0$ for every $j \in [2]$. Moreover, since $G_{i}- \{P^*\} = G_{i + 1}- \{P_1, P_2\}$, 
    we have \[R_{G_{i+1}}(V(H)) \setminus B_{G_{i+1}}(V(H))=R_{G_{i}}(V(H)) \setminus B_{G_{i}}(V(H)).\]

    Also $H$ is a connected red-induced subtrigraph of $G_{i+1}$, because $P_1, P_2 \notin V(H)$. Hence, we conclude that $(H,S,\eta) \in \mathcal{C}(i+1)$ and we can set $\Phi(i,H,S,\eta) = \Phi(i+1,H,S,\eta)$. 
    
    Second, we assume that $P^* \in R_{G_i}(V(H)) \setminus B_{G_i}(V(H))$ and $S \cap P^* = \emptyset$. In this case, $S\cap P_1=S\cap P_2=\emptyset$ and thus we have the properties that 
    \begin{itemize}
        \item  $|S \cap P| \le \alpha$ for every $P \in R_{G_{i+1}}(V(H)) \setminus B_{G_{i+1}}(V(H))$, and 
        \item $|S \cap P| = 0$ for every $ P \in V(G_{i+1}) \setminus \big(R_{G_{i+1}}(V(H)) \setminus B_{G_{i+1}}(V(H))\big)$.
    \end{itemize}
    So, we have $(H, S, \eta) \in \mathcal{C}(i+1)$. 
    As $S\cap P_1=S\cap P_2=\emptyset$, any valid solution with respect to $(i+1,H,S,\eta)$ is also valid with respect to $(i,H,S,\eta)$.
    Thus, we can set $\Phi(i,H,S,\eta) = \Phi(i+1,H,S,\eta)$.

\medskip
    From now on, we may assume that either 
    \begin{itemize}
        \item $P^* \in V(H)$ or 
        \item $P^* \in R_{G_i}(V(H)) \setminus B_{G_i}(V(H))$ with $P^* \cap S \neq \emptyset$.
    \end{itemize}
    For each $j \in [2]$, let $\mathcal{Q}_j = B_{G_{i+1}}(P_j) \cap (V(H)\setminus \{P^*\})$.     By~\Cref{lem:contraction}, $\mathcal{Q}_1$ and $\mathcal{Q}_2$ are disjoint.

    Using \Cref{lem:sigmaseparator}, one can observe that for a valid subset $I$ with respect to $(i, H, S, \eta)$, its restriction to the union of some subset of $\{P_1, P_2\}\cup \mathcal{Q}_1\cup \mathcal{Q}_2$ forms a $\sigma$-separator, and if $P^* \in R_{G_i}(V(H)) \setminus B_{G_i}(V(H))$, then we can take so that $\{P_1, P_2\}\subseteq \mathcal{O}$.
    We explain, for a given $\sigma$-separator, how to construct a valid subset corresponding to it.
    Let $(\mathcal{O}, U)$ be a $\sigma$-separator with respect to $(P_1, P_2, \mathcal{Q}_1, \mathcal{Q}_2)$ such that 
    \begin{itemize}
        \item if $P^* \in R_{G_i}(V(H)) \setminus B_{G_i}(V(H))$, then $\{P_1, P_2\}\subseteq \mathcal{O}$ and $U\cap P^*=S\cap P^*$.
    \end{itemize} 
    When $P^* \in R_{G_i}(V(H)) \setminus B_{G_i}(V(H))$, we require that our solution on $P^*$ is exactly $S \cap P^*$; thus, it suffices to consider $U$ with $U \cap P^* = S \cap P^*$.
    Let $F$ be the subtrigraph of $G_{i+1}$ induced by $(V(H)\setminus\{P^*\})\cup \{P_1, P_2\}$. We define $\mathcal{C}(\mathcal{O}, U)$, $\mathcal{Q}(\mathcal{O}, U)$, $\mathcal{S}(\mathcal{O}, U)$, $\mathcal{L}(\mathcal{O}, U)$ as follows:
    \begin{itemize}
        \item $\mathcal{C}(\mathcal{O}, U)=\{H_1, \ldots, H_k\}$ is the set of components of $F-\mathcal{O}$.
        \item $\mathcal{Q}(\mathcal{O}, U)=(Q_1, \ldots, Q_k)$ where for each $j\in [k]$, $Q_j$ is the union of all parts in $R_{G_{i + 1}}(V(H_j))$.
        \item $\mathcal{S}(\mathcal{O}, U)=(S_1, \ldots, S_k)$ where for each $j\in [k]$, $S_j=(S\cup U)\cap Q_j$.
        \item $\mathcal{L}(\mathcal{O}, U)$ is the collection of functions $(\Lambda_1, \ldots, \Lambda_k)$ where for each $j\in [k]$, $\Lambda_j: S\cup U \to \{0, 1, \ldots, \alpha\}$ is a function such that 
         \begin{itemize}
        \item if $v\in (S \cup U)\setminus S_{j}$, then 
        $\Lambda_{j}(v) = 0$, 
        \item for every $v \in S$, $\sum_{j=1}^k \Lambda_{j}(v) + |N_G(v) \cap (U \setminus S)| = \eta(v)$, and

        \item for every $v \in U\setminus S$, 
       $\sum_{j=1}^k \Lambda_{j}(v) + |N_G(v) \cap (S \cup U)| \in \sigma$.
       \end{itemize} 
    \end{itemize}

    By the black pair condition, each $H_j$ is a connected red-induced subtrigraph of $G_{i + 1}$.    
    So, for each~$j\in [k]$, $(H_j, S_{j}, \Lambda_{j})\in \mathcal{C}(i+1)$ because $R_{G_{i+1}}(V(H_j))\setminus \mathcal{O}\subseteq R_{G_{i}}(V(H))$ and $\Lambda_{j}(v) \in \{0, 1, \ldots, \alpha\}$ for every $v \in S_{j}$, and $0$ otherwise. Therefore, 
    $\Phi(i + 1, H_j, S_{j}, \Lambda_{j})$ has been recorded.

    Note that by the black pair condition, if $\mathcal{Q}_j\setminus \mathcal{O} \neq \emptyset$, then $P_j \in \mathcal{O}$. This implies that  $k\le t+2$ by \Cref{lem:componentcount}. 

    For $\Lambda=(\Lambda_{1}\ldots, \Lambda_{k})\in \mathcal{L}(\mathcal{O}, U)$,
    we set 
    \[I_{\mathcal{O}, U,\Lambda}=\begin{cases}\left(\bigcup\limits_{j = 1}^k \Phi(i + 1, H_j, S_{j}, \Lambda_{j})\right) \cup (U\setminus S) & \text{ if }\, \Phi(i + 1, H_j, S_{j}, \Lambda_{j}) \neq \bot \,\, \text{for all }\, j \in [k],\\
    \emptyset & \text{ otherwise}.
    \end{cases}\]

    We prove the following.

    \begin{claim}\label{claim:zetamax}
        If $\zeta(i,H,S,\eta)\neq \emptyset$, then there exist a $\sigma$-separator $(\mathcal{O}, U)$ and $\Lambda\in \mathcal{L}(\mathcal{O},U)$ for which $I_{\mathcal{O},U, \Lambda}\in \zeta_{max}(i,H,S,\eta)$.
    \end{claim}
    \begin{clproof}
        
        As $\zeta(i,H,S,\eta)\neq \emptyset$, we have $\zeta_{max}(i,H,S,\eta)\neq \emptyset$.
        Let $I \in \zeta_{max}(i,H,S,\eta)$. Then, by~\Cref{lem:sigmaseparator}, there exists a $\sigma$-separator $(\mathcal{O}, U)$ with respect to $(P_1, P_2, \mathcal{Q}_1, \mathcal{Q}_2)$ such that 
        $U = I_{\mathcal{O}}$ and if $|I\cap (P_1\cup P_2)|\le \alpha$, then $\{P_1, P_2\}\subseteq \mathcal{O}$. 
        Let $\mathcal{C}(\mathcal{O}, U)=\{H_1, \ldots, H_k\}$ and $\mathcal{S}(\mathcal{O}, U)=(S_1, \ldots, S_k)$.
        For each $j\in [k]$, let $I_j = I \cap \bigcup_{X \in V(H_j)} X$.  
        
        By the black pair condition, each $H_j$ is a connected red-induced subtrigraph of $G_{i+1}$. Moreover, $R_{G_{i+1}}(H_j) \subseteq R_{G_{i}}(H) \cup \mathcal{O}$. Since $(H, S, \eta) \in \mathcal{C}(i)$ and $(\mathcal{O}, U)$ is a $\sigma$-separator, it follows that $|(S \cup U) \cap P| \le \alpha$ for every $P \in R_{G_{i+1}}(V(H_j)) \setminus B_{G_{i+1}}(V(H_j))$.
        For each $j \in [k]$, let $\Lambda_j : S \cup U \to \{0, 1, \ldots, \alpha\}$ be a function such that 
        \begin{itemize}
            \item if $v \in S_ j$, then $\Lambda_j(v) = |N_G(v) \cap I_j|$, and
            \item if $v \notin (S \cup U) \setminus S_j$, then $\Lambda_j(v) = 0$.
        \end{itemize} 
        Then, $(H_j, S_j, \Lambda_j) \in \mathcal{C}(i+1)$. Next, we show that $I_j \in \zeta(i+1, H_j, S_j, \Lambda_j)$ for each $j \in [k]$.

        Since $(\mathcal{O}, U)$ is a $\sigma$-separator, we have $\trace_{\mathcal{P}_{i+1}}\big(\bigcup_{j \in [k]} I_j\big) = V(H) \setminus \mathcal{O}$. Thus, for each $j \in [k]$, it follows that $\trace_{\mathcal{P}_{i+1}}(I_j) = V(H_j)$. Moreover, since $R_{G_{i+1}}(H_j) \subseteq R_{G_{i}}(H) \cup \mathcal{O}$, we obtain $|N_G(v) \cap (I_j \cup S_j)| = |N_G(v) \cap (I \cup S)| \in \sigma$ for every $v \in I_j$. By the definition, for each $v \in S_j$, we have $\Lambda_j(v) = |N_G(v) \cap I_j| \le \alpha$. Hence, for each $j \in [k]$, $I_j \in \zeta(i+1, H_j, S_j, \Lambda_j)$. 

        Since $I$ is valid, for every $v \in U\setminus S$, we have
        \begin{align*}
          \sum_{j=1}^k \Lambda_j(v) + |N_G(v) \cap (U \cup S)| &= \sum_{j=1}^k |N_G(v) \cap I_j| + |N_G(v) \cap (U \cup S)| \\
          &= |N_G(v) \cap (I \cup S)| \in \sigma,  
        \end{align*}
        and for each $v \in S$, we have
        \begin{align*}
            \sum_{j=1}^k \Lambda_j(v) + |N_G(v) \cap (U\setminus S)| & = \sum_{j=1}^k |N_G(v) \cap I_j| + |N_G(v) \cap (U\setminus S)| \\
            &= |N_G(v) \cap I| = \eta(v).
        \end{align*}
        Therefore, $(\Lambda_1, \ldots, \Lambda_k) \in \mathcal{L}(\mathcal{O}, U)$. Consequently,
        \[|I| = |U\setminus S| + \sum_{j = 1}^k |I_j| \le |U \setminus S| + \sum_{j = 1}^k |\Phi(i+1, H_j, S_j, \Lambda_j)| = |I_{\mathcal{O}, U, \Lambda}|.\]
        This completes the proof.
        
    \end{clproof}

    Recall that any valid set is non-empty.
    
    We compute $I_{\mathcal{O}, U, \Lambda}$ for each $\Lambda \in \mathcal{L}(\mathcal{O}, U)$.
    
    If $I_{\mathcal{O}, U, \Lambda} = \emptyset$ for every possible choice of a $\sigma$-separator $(\mathcal{O}, U)$ and $\Lambda\in \mathcal{L}(\mathcal{O},U)$, then by \Cref{claim:zetamax}, we have $\zeta(i, H, S, \eta)=\emptyset$. In this case, we can set $\Phi(i, H, S, \eta) = \bot$.

    Otherwise, 
    there exist a $\sigma$-separator $(\mathcal{O}, U)$ and $\Lambda\in \mathcal{L}(\mathcal{O},U)$ for which $I_{\mathcal{O},U, \Lambda}\neq \emptyset$. In this case, we select a set of maximum size among them and record it as $\Phi(i, H, S, \eta)$. By \Cref{claim:zetamax}, this set must be a set in $\zeta_{max}(i,H,S,\eta)$.

    This proves the correctness of the algorithm.

\medskip    
    We analyze the running time. We present upper bounds on  $|\mathcal{C}(i)|$ and $|\mathcal{L}(\mathcal{O}, U)|$.

    \begin{claim}\label{claim:number of C(i)}
        $|\mathcal{C}(i)| \le n^{(\alpha + 1)t} \cdot (\alpha + 1)^{\alpha t}$.
    \end{claim}
    \begin{clproof}
         By~\Cref{lem:numberofinducedsubgraphs}, the number of non-empty connected red-induced subtrigraphs of $G_i$ is at most $n^{t}$. Let $H$ be a non-empty connected red-induced subtrigraph of $G_i$. By applying~\Cref{lem:thenumberofredneighbor} to the red component of $G_i$ containing $H$, we have 
        \[|R_{G_i}(V(H))\setminus B_{G_i}(V(H))| \le |R_{G_i}(V(H))| \le t.\]

         It follows that $|S| \le \alpha t$. Hence, there are at most $n^{\alpha t}$ choices for $S$, and for each such choice, there are at most $(\alpha + 1)^{\alpha t}$ possible functions $\eta$. Thus, we have $|\mathcal{C}(i)| \le n^{(\alpha + 1)t} \cdot (\alpha + 1)^{\alpha t}$.
    \end{clproof}

\begin{claim}\label{claim:numberoffunctions}
        $|\mathcal{L}(\mathcal{O}, U)| \le (\alpha + k)^{\alpha^2(t+4)}$.
    \end{claim}
    \begin{clproof}
        Consider a vertex $v \in S$. Since $\sum_{j=1}^k \Lambda_{j}(v) = \eta(v) - |N_G(v) \cap (U\setminus S)| \le \alpha$, the number of possible tuples $(\Lambda_1(v), \ldots, \Lambda_k(v))$ is at most $(\alpha + k)^\alpha$. Now consider a vertex $v \in U\setminus S$. Since $\sum_{j=1}^k \Lambda_{j}(v) + |N_G(v) \cap (U \cup S)| \in \sigma$, it follows that $\sum_{j=1}^k \Lambda_{j}(v) \in \sigma$. So, the number of possible tuples $(\Lambda_1(v), \ldots, \Lambda_k(v))$ is at most $(\alpha + k)^{\alpha}$. As $|S|\le t\alpha$ and $|U|\le 4\alpha$, there are at most 
        \[   (\alpha + k)^{\alpha |S|}\cdot (\alpha + k)^{\alpha|U|}
        \le 
        (\alpha + k)^{\alpha^2(t+4)} \]
        possible functions $(\Lambda_1, \ldots, \Lambda_k)$.
        
    \end{clproof}

    By \Cref{lem:numberofseparators}, the number of possible $\sigma$-separators with respect to a given $(P_1, P_2, \mathcal{Q}_1, \mathcal{Q}_2)$ is at most $4n^{4 \alpha}$. Thus, by \Cref{claim:numberoffunctions}, each $\Phi(i, H, S, \eta)$ can be computed in time $O(n^{4\alpha}\cdot (\alpha + t + 2)^{\alpha^2(t+4)})$.
    So, for fixed $i \in [n]$, $\Phi(i, H, S, \eta)$ for all $(H, S, \eta)\in \mathcal{C}(i)$ can be computed in time $O(|\mathcal{C}(i)|\cdot n^{4\alpha}\cdot (\alpha + t + 2)^{\alpha^2(t+4)} )$. Overall, by~\Cref{claim:number of C(i)}, this algorithm runs in time
    \[O\big(n^{(\alpha + 1)t + 4\alpha + 1} \cdot (\alpha + 1)^{\alpha t} \cdot (\alpha + t + 2)^{\alpha^2(t+4)}\big)=n^{O(\alpha^2 t)}, \]
    where $\alpha\le n$ can be assumed because  it suffices to consider $\sigma$ only up to $n$.
    \end{proof}

\subsection{Induced Disjoint Paths}\label{subsec:IDP}

   We consider the following problem.  
   Paths $P_1, \ldots, P_k$ in a graph are \emph{mutually induced}
    if each $P_i$ is induced, and for distinct $i,j\in [k]$, $P_i$ and $P_j$ have neither common vertices nor adjacent vertices.
\medskip

\noindent
\fbox{\parbox{0.97\textwidth}{
	\textsc{Induced Disjoint Paths (IDP)}\\
	\textbf{Input :} A graph $G$ with pairs of vertices $(x_1, y_1), \ldots, (x_\ell,y_\ell)$ of $G$.
	\\
	\textbf{Output :} A set of mutually induced paths $P_1, \ldots, P_\ell$ such that for each $i\in [\ell]$, $P_i$ is a $(x_i,y_i)$-path.}}

\medskip

Similar to what we did for the \textsc{$\sigma$-Neighborhood} problem, we recursively find mutually induced disjoint paths within the union of parts corresponding to connected red-induced subtrigraphs of$~G/\mathcal{P}_i$. We similarly define a separator, called an IDP separator, to divide the red-induced subtrigraph into components obtained by deleting the separator to which we can apply induction.

For a set $\mathcal{G}$ of graphs, we denote by $\bigcup \mathcal{G}$ the union of graphs in $\mathcal{G}$. For a vertex set $Q \subseteq V(G)$ and a set $\mathcal{G}$ of graphs, we define $\mathcal{G}_{Q} := \{H[V(H)\cap Q] : H \in \mathcal{G}\}$. For a family $\mathcal{Q}$ of vertex sets in $G$ and a set $\mathcal{G}$ of graphs, we define $\mathcal{G}_{\mathcal{Q}} := \{H[V(H) \cap \big(\bigcup_{X \in \mathcal{Q}} X\big)] : H \in \mathcal{G}\}$. Recall that, for a family $\mathcal{Q}$ of vertex sets in $G$ and $I\subseteq V(G)$, we write $I_{\mathcal{Q}}=I\cap \left(\bigcup_{X\in \mathcal{Q}}X\right)$. If $\mathcal{Q}=\{Q\}$ consists of one set, then we write $I_Q=I_{\{Q\}}$.

For a set $\mathcal{Z}=\{Z_1, Z_2, \ldots, Z_m\}$ of paths in $G$, let $\pairs(\mathcal{Z})$ denote the set of pairs $\{a_i,b_i\}$ where each $Z_i$ is an $(a_i, b_i)$-path. In the case $a_i=b_i$, the corresponding set is $\{a_i\}$.

Let $P_1, P_2$ be two parts of $\mathcal{P}$ and for each $i\in [2]$, let $\mathcal{Q}_i\subseteq B_{G/\mathcal{P}}(P_i)\setminus B_{G/\mathcal{P}}(P_{3-i})$. We say that a pair $(\mathcal{O}, \mathcal{U})$ of a subset $\mathcal{O}\subseteq \{P_1, P_2\}\cup \mathcal{Q}_1\cup \mathcal{Q}_2$ and a family $\mathcal{U}$ of mutually induced paths of length at least $1$ in $G[\bigcup_{Q\in \mathcal{O}}Q]$ with $U = V(\bigcup \mathcal{U})$ is an \emph{IDP separator} with respect to $(P_1, P_2, \mathcal{Q}_1, \mathcal{Q}_2)$ if the followings hold: 

\begin{itemize}
    \item $\bigcup \mathcal{U}$ contains at most $2$ vertices from each part in $\mathcal{O}$, 
    \item if $\mathcal{O} \cap \mathcal{Q}_i \neq \emptyset$ for some $i$, then $\mathcal{Q}_i \subseteq \mathcal{O}$, 
    \item if $|U_{P_i}| \ge 1$ or $|U_{\mathcal{Q}_i}| \ge 1$ for some $i \in [2]$, then $\{P_i\} \cup \mathcal{Q}_i \subseteq \mathcal{O}$, and
    \item for every part $Q\in \{P_1, P_2\}\cup \mathcal{Q}_1\cup \mathcal{Q}_2$ that is not in $\mathcal{O}$ and every part $Q'\in \{P_1, P_2\}\cup \mathcal{Q}_1\cup \mathcal{Q}_2$ where $Q$ and $Q'$ are linked by a black edge in $G/\mathcal{P}$, $Q'\in \mathcal{O}$ and $|U_{Q'}|=0$.
\end{itemize}
Similarly, the last condition will be referred to as the `black pair condition'. Note that if $|U_{P_i}| \ge 1$ or $|U_{\mathcal{Q}_i}| \ge 1$ for some $i \in [2]$, then we always have $|U_{P_i} \cup U_{\mathcal{Q}_i}| \le 3$, because otherwise, $\bigcup\mathcal{U}$ has a subgraph isomorphic to $C_4$ or $K_{1,3}$.

For a connected red-induced subtrigraph $H$ of $G/\mathcal{P}_i$, we guess a possible set of terminals on the parts that are red, but not black, neighbors of $H$.
As illustrated in \Cref{fig:inducedpath1}, we guess a partial solution at an IDP separator and consider subproblems on the components obtained by removing the IDP separator. There may be a special case where an IDP separator lies both outside and inside $H$; see \Cref{fig:inducedpath2}. In this situation, some paths in $\mathcal{U}$ may start at vertices of $S$, and special care is required when constructing the subproblems.

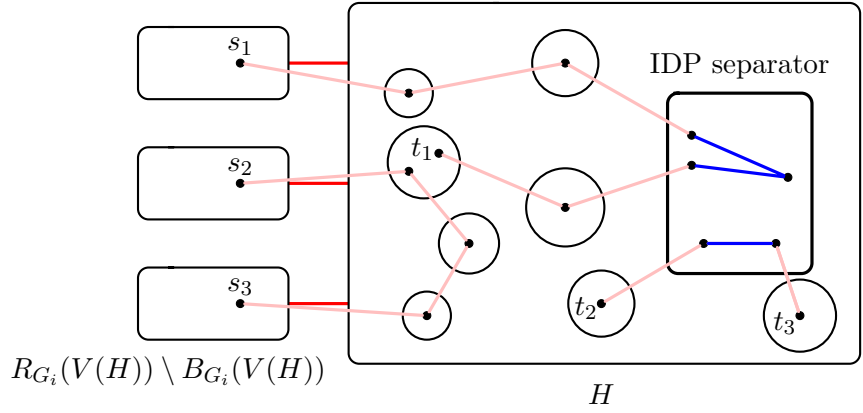
\begin{figure}
    \centering
    \begin{tikzpicture}[scale=0.8]
        \tikzstyle{v}=[circle, draw, solid, fill=black, inner sep=0pt, minimum width=3pt]
        \tikzset{c1/.style={purple, line width=6pt,opacity=0.5,line cap=round,shorten >=-2pt, shorten <=-2pt}}

        \draw[red,very thick] (-3.5, 4.5)--(-1, 4.5);
        \draw[red,very thick] (-3.5, 2.5)--(-1, 2.5);
        \draw[red,very thick] (-3.5, 0.5)--(-1, 0.5);

        \draw[rounded corners, thick, fill=white!] (0.5,5.5)--(-2,5.5)--(-2,-0.5)--(6.5,-0.5)--(6.5,5.5)--(0.4,5.5);

        \draw[rounded corners, thick, fill=white!] (-5, 5.1)--(-3, 5.1)--(-3, 3.9)--(-5.5, 3.9)--(-5.5, 5.1)--(-4.9, 5.1);  
        \draw[rounded corners, thick, fill=white!] (-5, 3.1)--(-3, 3.1)--(-3, 1.9)--(-5.5, 1.9)--(-5.5, 3.1)--(-4.9, 3.1); 
        \draw[rounded corners, thick, fill=white!] (-5, 1.1)--(-3, 1.1)--(-3, -0.1)--(-5.5, -0.1)--(-5.5, 1.1)--(-4.9, 1.1);

        \draw[rounded corners, very thick, fill=white!] (3.5, 4)--(3.3, 4)--(3.3, 1)--(5.7, 1)--(5.7, 4)--(3.4, 4);

        \draw[thick, fill = white!] (-1, 4) circle (0.4);
        \draw[thick, fill = white!] (1.6, 4.5) circle (0.55);
        \draw[thick, fill = white!] (1.6, 2.1) circle (0.65);
        \draw[thick, fill = white!] (-0.75, 2.85) circle (0.6);
        \draw[thick, fill = white!] (0, 1.5) circle (0.5);
        \draw[thick, fill = white!] (2.2, 0.5) circle (0.55);
        \draw[thick, fill = white!] (-0.7, 0.3) circle (0.4);
        \draw[thick, fill = white!] (5.5, 0.3) circle (0.6);

        \node[v] (a1) at (-3.8, 4.5) {};
        \node[v] (a2) at (-1, 4) {};
        \node[v] (a4) at (1.6, 4.5) {};
        \node[v] (a5) at (3.7, 3.3) {};
        \node[v] (a6) at (5.3, 2.6) {};
        \node[v] (a7) at (3.7, 2.8) {};
        \node[v] (a8) at (1.6, 2.1) {};
        \node[v] (a9) at (-0.5, 3) {};
                
        \draw[pink, very thick, line cap=round] (a1) -- (a2) -- (a4) -- (a5);
        \draw[blue, very thick, line cap=round] (a5) -- (a6) -- (a7);
        \draw[pink, very thick, line cap=round] (a7) -- (a8) -- (a9);

        \node[v] (b1) at (-3.8, 2.5) {};
        \node[v] (b2) at (-1, 2.7) {};
        \node[v] (b3) at (0, 1.5) {};
        \node[v] (b4) at (-0.7, 0.3) {};
        \node[v] (b5) at (-3.8, 0.5) {};

        \draw[pink, very thick, line cap=round] (b1) -- (b2) -- (b3) -- (b4) -- (b5);

        \node[v] (c1) at (2.2, 0.5) {};
        \node[v] (c2) at (3.9, 1.5) {};
        \node[v] (c3) at (5.1, 1.5) {};
        \node[v] (c4) at (5.5, 0.3) {};

        \draw[pink, very thick, line cap=round] (c1)--(c2)--(c3)--(c4);
        \draw[blue, very thick, line cap=round] (c2)--(c3);
        \draw[pink, very thick, line cap=round] (c3)--(c4);

        \node [label=$H$] (v) at (2.2, -1.5){};
        \node [label=$R_{G_i}(V(H)) \setminus B_{G_i}(V(H))$] (v) at (-5, -1.2){};
        \node [label={IDP separator}] (v) at (4.5, 3.9){};

        \node [label=$s_1$] (v) at (-3.8, 4.3){};
        \node [label=$s_2$] (v) at (-3.8, 2.3){};
        \node [label=$s_3$] (v) at (-3.8, 0.3){};
        \node [label=$t_1$] (v) at (-0.77, 2.55){};
        \node [label=$t_2$] (v) at (1.95, -0.1){};
        \node [label=$t_3$] (v) at (5.25, -0.35){};
   
    \end{tikzpicture}
    \caption{We require that the restriction of each solution to an IDP separator consists of paths of length at least $1$.} 
    \label{fig:inducedpath1}
\end{figure}

\begin{figure}
    \centering
    \begin{tikzpicture}[scale=0.8]
        \tikzstyle{v}=[circle, draw, solid, fill=black, inner sep=0pt, minimum width=3pt]
        \tikzset{c1/.style={purple, line width=6pt,opacity=0.5,line cap=round,shorten >=-2pt, shorten <=-2pt}}

        \draw[red,very thick] (-3.5, 4)--(-1, 4);
        \draw[red,very thick] (-3.5, 0.5)--(-1, 0.5);

        \draw[rounded corners, thick, fill=white!] (0.5,5.5)--(-2,5.5)--(-2,-0.5)--(6.5,-0.5)--(6.5,5.5)--(0.4,5.5);

        \draw[rounded corners, very thick, fill=white!] (-4.5, 5.1)--(-3, 5.1)--(-3, 3) -- (-0.5, 3) -- (-0.5, 1.9) --(-5, 1.9)--(-5, 5.1)--(-4.4, 5.1);

        \draw[rounded corners, thick, fill=white!] (-5, 1.1)--(-3, 1.1)--(-3, -0.1)--(-5, -0.1)--(-5, 1.1)--(-4.9, 1.1);

        \draw[thick, fill = white!] (-1.1, 4.9) circle (0.4);
        \draw[thick, fill = white!] (1.6, 4.5) circle (0.55);
        \draw[thick, fill = white!] (1.8, 2.1) circle (0.65);
        \draw[thick, fill = white!] (0, 1.5) circle (0.5);
        \draw[thick, fill = white!] (2.2, 0.6) circle (0.65);
        \draw[thick, fill = white!] (-0.7, 0.3) circle (0.4);
        \draw[thick, fill = white!] (5.5, 0.45) circle (0.65);
        \draw[thick, fill = white!] (3.5, 1.2) circle (0.6);
        \draw[thick, fill = white!] (4, 3) circle (1);
        \draw[thick, fill = white!] (-0.2, 3.9) circle (0.7);
        \draw[thick, fill = white!] (5.7, 4.6) circle (0.5);

        \node[v] (a2) at (-0.5, 3.7) {};
        \node[v] (a3) at (1.6, 2.1) {};
        \node[v] (a4) at (3.9, 3.5) {};

        \node[v] (b1) at (-3.8, 4.7) {};
        \node[v] (b2) at (-1, 5) {};
        \node[v] (b3) at (1.6, 4.5) {};
        \node[v] (b4) at (5.8, 4.8) {};
        
        \draw[pink, very thick, line cap=round] (a2) -- (a3) -- (a4);
        
        \draw[pink, very thick, line cap=round] (b1) -- (b2) -- (b3) -- (b4);

        \node[v] (c1) at (-3.5, 2.45) {};
        \node[v] (c2) at (-1.3, 2.45) {};
        \node[v] (c3) at (-0.7, 0.3) {};
        \node[v] (c4) at (0, 1.5) {};
        \node[v] (c5) at (2.2, 0.5) {};
        \node[v] (c6) at (3.4, 1.2) {};
        \node[v] (c7) at (3.9, 2.3) {};
        \node[v] (c8) at (4.7, 2.7) {};
        \node[v] (c9) at (5.5, 0.3) {};

        \draw[blue, very thick, line cap=round] (c1)--(c2);
        \draw[pink, very thick, line cap=round] (c2)--(c3)--(c4)--(c5)--(c6)--(c7)--(c8)--(c9);

        \draw[blue, very thick, line cap=round] (c1) -- (b1);

        \node[v] (d) at (-3.8, 0.5) {};

        \node [label=$H$] (v) at (2.2, -1.5){};
        \node [label=$R_{G_i}(V(H)) \setminus B_{G_i}(V(H))$] (v) at (-5, -1.2){};
        \node [label={IDP separator}] (v) at (-4.2, 5){};

        \node [label=$s_1$] (v) at (-4.2, 4.2){};
        \node [label=$s_2$] (v) at (-1, 2.2){};
        \node [label=$t_1$] (v) at (5.8, 3.95){};
        \node [label=$t_2$] (v) at (-0.2, 3.5){};
        \node [label=$t_3$] (v) at (4.3, 2.95){};
        \node [label=$t_4$] (v) at (5.7, 0.1){};
        \node [label=$a_1$] (v) at (-3.85, 2.0){};
        \node [label=$a_2$] (v) at (-4.2, 0){};

        \draw[thick, dashed] (d) -- (c3);

    \end{tikzpicture}
    \caption{The case where an IDP separator lies both outside and inside $H$. Then a path in $\mathcal{U}$ may start from a vertex in $S$, and in the subproblems, we have to exclude it from consideration. } 
    \label{fig:inducedpath2}
\end{figure}
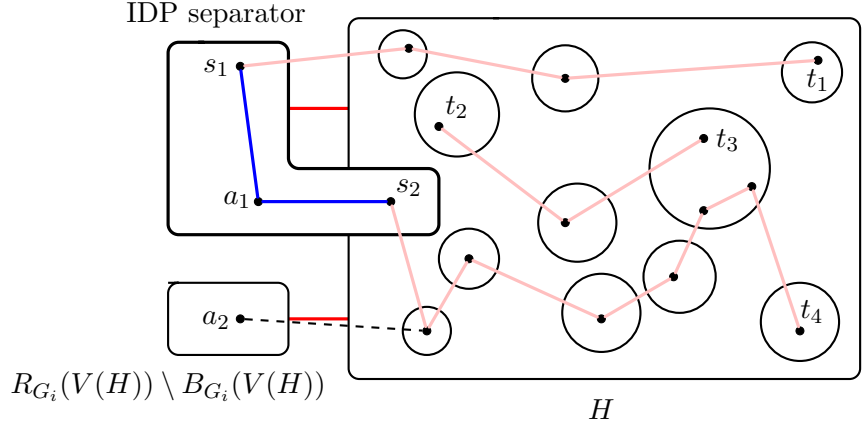

    We obtain an upper bound on the number of IDP separators.
\begin{lemma}\label{lem:numberofIDPseparators}
    Let $G$ be a graph and let $\mathcal{P}$ be a partition of $V(G)$.
    Let $P_1$ and $P_2$ be two parts of $\mathcal{P}$, and for each $i\in [2]$, let $\mathcal{Q}_i\subseteq B_{G/\mathcal{P}}(P_i)\setminus B_{G/\mathcal{P}}(P_{3-i}).$ Then the number of IDP separators with respect to $(P_1, P_2, \mathcal{Q}_1, \mathcal{Q}_2)$ is at most $4\cdot|V(G)|^{6}$.
\end{lemma}
\begin{proof}
    Let $(\mathcal{O}, \mathcal{U})$ be an IDP separator with respect to $(P_1, P_2, \mathcal{Q}_1, \mathcal{Q}_2)$. Let $U = V(\bigcup \mathcal{U})$. By the definition, $|U_{P_i} \cup U_{\mathcal{Q}_i}| \le 3$ for each $i \in [2]$. It follows that $|U| \le 6$, and hence the number of possible choice of $U$ is at most $|V(G)|^{6}$. If $U_{\{P_i\}} = \emptyset$ for some $i \in [2]$, then both cases $P_i \in \mathcal{O}$ and $P_i \notin \mathcal{O}$ are possible. Therefore, the number of IDP separators is at most $4\cdot|V(G)|^{6}$.
\end{proof}

The following lemma will be used to guarantee the existence of an IDP separator.
 \begin{lemma}\label{lem:sepIDP}
        Let $G$ be a graph and let $\mathcal{P}$ be a partition of $V(G)$. Let $P_1$ and $P_2$ be two parts of $\mathcal{P}$, and for each $i \in [2]$, let $\mathcal{Q}_i \subseteq B_{G/\mathcal{P}}(P_i) \setminus B_{G/\mathcal{P}}(P_{3-i})$. 
        Let $\mathcal{Z}$ be a set of mutually induced paths in $G$. Then there exists a subset $\mathcal{O}\subseteq \{P_1, P_2\}\cup \mathcal{Q}_1\cup \mathcal{Q}_2$ such that $(\mathcal{O}, \mathcal{Z}_{\mathcal{O}})$ is an IDP separator with respect to $(P_1, P_2, \mathcal{Q}_1, \mathcal{Q}_2)$.
    \end{lemma}
        
    \begin{proof}
        Let $Z = V(\bigcup \mathcal{Z})$. Assume that $|Z_{P_j}| \ge 3$ for some $j \in [2]$. As $\mathcal{Z}$ is a set of mutually induced paths and $P_j$ is adjacent to each part of $\mathcal{Q}_j$ in $G/\mathcal{P}$, we have $|Z_{\mathcal{Q}_j}| = 0$. If $Z_{P_{3-j}} = \emptyset$, then $(\{P_{3-j}\} \cup \mathcal{Q}_j, \emptyset)$ is an IDP separator. Thus, we may assume that $Z_{P_{3-j}} \neq \emptyset$.
        \begin{itemize}
            \item If $|Z_{\mathcal{Q}_{3-j}}| \ge 1$, then every vertex in $Z_{P_{3-j}}\cup Z_{\mathcal{Q}_{3-j}}$ has a neighbor in $\mathcal{Z}_{\{P_{3-j}\} \cup \mathcal{Q}_{3-j}}$. 
            Therefore, $(\mathcal{O} = \{P_{3-j}\} \cup \mathcal{Q}_1 \cup \mathcal{Q}_2, \mathcal{Z}_{\mathcal{O}})$ is an IDP separator.
            \item Assume $|Z_{\mathcal{Q}_{3-j}}| = 0$. If there is a black edge between $P_1$ and $P_2$, then the vertex in $Z_{P_{3-j}}$ has three neighbors in $\mathcal{Z}$, which is not possible. So, there is no black edge between $P_1$ and $P_2$, and this deduces that $(\mathcal{Q}_1 \cup \mathcal{Q}_2, \emptyset)$ is an IDP separator.
        \end{itemize}
        Hence, we may assume that $|Z_{P_j}| \le 2$ for every $j \in [2]$. Assume that $Z_{P_j} = \emptyset$ for some $j \in [2]$. 
        \begin{itemize}
            \item If $Z_{P_{3-j}} = \emptyset$, then $(\{P_1, P_2\}, \emptyset)$ is an IDP separator.
            \item If $Z_{P_{3-j}} \neq \emptyset$, then $(\mathcal{O} = \{P_1, P_2\} \cup \mathcal{Q}_{3-j}, \mathcal{Z}_{\mathcal{O}})$ is an IDP separator.
        \end{itemize}

        Now, we assume that $Z_{P_j} \neq \emptyset$ for every $j \in [2]$. If $P_1$ and $P_2$ are connected by a black edge in $G/\mathcal{P}$, then $(\mathcal{O} = \{P_1, P_2\} \cup \mathcal{Q}_1 \cup \mathcal{Q}_2, \mathcal{Z}_{\mathcal{O}})$ is an IDP separator, because every vertex in $Z$ has a neighbor in $\mathcal{Z}_{\{P_1, P_2\}}$. Thus, we may assume that $P_1$ and $P_2$ are not connected by a black edge in $G/\mathcal{P}$. Then we can divide into three cases.
        \begin{itemize}
            \item If $|Z_{\mathcal{Q}_1}|, |Z_{\mathcal{Q}_2}| \ge 1$, then $(\mathcal{O} = \{P_1, P_2\} \cup \mathcal{Q}_1 \cup \mathcal{Q}_2, \mathcal{Z}_{\mathcal{O}})$ is an IDP separator.
            \item If $|Z_{\mathcal{Q}_j}| \ge 1$ and $|Z_{\mathcal{Q}_{3-j}}| = 0$ for some $j \in [2]$, then $(\mathcal{O} = \{P_j\} \cup \mathcal{Q}_1 \cup \mathcal{Q}_2, \mathcal{Z}_{\mathcal{O}})$ is an IDP separator. 
            \item If $|Z_{\mathcal{Q}_1}|= |Z_{\mathcal{Q}_2}| = 0$, then $(\mathcal{Q}_1 \cup \mathcal{Q}_2, \emptyset)$ is an IDP separator. 
        \end{itemize}
This proves the lemma.
    \end{proof}

    We present the algorithm in detail.
\begin{theorem}
  There is an algorithm that, given an $n$-vertex graph $G$ and a~contraction sequence witnessing that $\compmaxleaf^\downarrow(G) \leqslant t$, solves \textsc{Induced Disjoint Paths} on~$G$ in time $n^{O(t)}$.  
\end{theorem}
\begin{proof}
     Let $\mathcal{P}_n, \ldots, \mathcal{P}_1$ be a given contraction sequence of $G$ witnessing that $\compmaxleaf^\downarrow(G) \leqslant t$.
    Let $G_i = G / \mathcal{P}_i$ for each $i$.  Let $T$ be the set of all terminal vertices.

    For each $i \in [n]$, let $\mathcal{C}(i)$ be the collection of all tuples $(H, S, A, \mathcal{J})$ such that 
    \begin{itemize}
        \item $H$ is a non-empty connected red-induced subtrigraph of $G_i$, 
        \item $S$ and $A$ are disjoint subsets of $V(G)$ such that
        \begin{itemize}
            \item $|(S \cup A) \cap P| \le 2$ for every $P \in R_{G_i}(V(H)) \setminus B_{G_i}(V(H))$ and
            \item $|(S \cup A) \cap P| = 0$ for every $P \in V(G_i) \setminus \big(R_{G_i}(V(H)) \setminus B_{G_i}(V(H))\big)$, and
        \end{itemize} 
        \item $\mathcal{J}$ is a partition of $S\cup (T\cap \bigcup_{X\in V(H)} X)$ into pairs of distinct vertices.
    \end{itemize}
    For $(H,S,A, \mathcal{J})\in \mathcal{C}(i)$, a set of paths $\mathcal{Z}$ in $G$ is \emph{valid} with respect to $(i,H,S,A,\mathcal{J})$ if 
    \begin{itemize}
        \item paths in $\mathcal{Z}$ are mutually induced in $G-E(G[S])$,
        \item $G[V(\bigcup\mathcal{Z})]$ is a disjoint union of paths,
        \item for every $(x,y)$-path $Q$ in $\mathcal{Z}$,  $\trace_{\mathcal{P}_i}(Q-\{x,y\})\subseteq V(H)$ and there is no edge between $A$ and $Q-\{x, y\}$ in $G$, and
        \item $\pairs(\mathcal{Z})=\mathcal{J}$.
        
    \end{itemize} 
    Let $\zeta(i,H,S,A,\mathcal{J})=1$ if there exists a valid set with respect to $(i,H,S,A,\mathcal{J})$, and $0$ otherwise. 
    
    We recursively store a valid set of $G$ with respect to $(i,H,S,A,\mathcal{J})$ if one exists.
    For $(H, S, A, \mathcal{J})\in \mathcal{C}(i)$ with $\zeta(i,H,S,A,\mathcal{J})=1$, we assign a valid set to $\Phi(i,H,S,A,\mathcal{J})$. If $\zeta(i,H,S,A,\mathcal{J}) = 0$, then we assign $\Phi(i,H,S,A,\mathcal{J}) = \bot$. 

    Observe that $\Phi(1, G_1, \emptyset, \emptyset, \mathcal{J})$ with $\mathcal{J} = \{\{x_1, y_1\}, \ldots, \{x_\ell, y_\ell\}\}$ is a desired solution.

    In the base case when $i=n$, $G_n$ has no red edges. 
    Thus, for every $(H, S,A, \mathcal{J})\in \mathcal{C}(n)$, $H$ consists of a single vertex $\{v\}\in \mathcal{P}_n$ and $S = A = \emptyset$ because $R_{G_n}(V(H))=\emptyset$.  Also, as $H$ has only one vertex and thus, there is no pair in $\mathcal{J}$. Therefore, the empty graph is a valid set with respect to $(n, H, S, A, \mathcal{J})$.
    Thus, we record an empty graph to $\Phi(n,H,S,A,\mathcal{J})$ for every $(H, S,A,\mathcal{J})\in \mathcal{C}(n)$.

   Henceforth, we assume that $i < n$. We further assume that for every $(H, S, A, \mathcal{J})\in \mathcal{C}(i+1)$, $\zeta(i+1,H,S,A,\mathcal{J})$ and $\Phi(i + 1,H,S,A,\mathcal{J})$ have been recorded. Let $P_1$ and $P_2$ be the parts of $\mathcal{P}_{i + 1}$ that are merged to $P^*$ in $\mathcal{P}_i$. Let $(H, S,A, \mathcal{J}) \in \mathcal{C}(i)$.

    \medskip
    We first deal with two simpler cases.  First assume that \[P^* \in V(G_i) \setminus \big(V(H) \cup (R_{G_i}(V(H)) \setminus B_{G_i}(V(H)))\big).\] 
    As $(H,S,A,\mathcal{J})\in \mathcal{C}(i)$, 
    we have $|(S \cup A) \cap P^*| = 0$ and it follows that $|(S \cup A) \cap P_j| = 0$ for every $j \in [2]$. Moreover, since $G_{i}- \{P^*\} = G_{i + 1}- \{P_1, P_2\}$, 
    we have \[R_{G_{i+1}}(V(H)) \setminus B_{G_{i+1}}(V(H))=R_{G_{i}}(V(H)) \setminus B_{G_{i}}(V(H)).\]
    Also $H$ is a connected red-induced subtrigraph of $G_{i+1}$, and $\mathcal{J}$ is a partition of $S\cup (T\cap \bigcup_{X\in V(H)} X)$ into pairs. Hence, we conclude that $(H,S,A,\mathcal{J}) \in \mathcal{C}(i+1)$ and we can set $\zeta(i,H,S,A,\mathcal{J}) = \zeta(i+1,H,S,A,\mathcal{J})$ and $\Phi(i,H,S,A,\mathcal{J}) = \Phi(i+1,H,S,A,\mathcal{J})$. 
    
    Second, we assume that $P^* \in R_{G_i}(V(H)) \setminus B_{G_i}(V(H))$ and $(S \cup A) \cap P^* = \emptyset$. In this case, $(S \cup A)\cap P_1=(S\cup A)\cap P_2=\emptyset$ and thus we have the properties that 
    \begin{itemize}
        \item  $|(S \cup A) \cap P| \le 2$ for every $P \in R_{G_{i+1}}(V(H)) \setminus B_{G_{i+1}}(V(H))$, and 
        \item $|(S \cup A) \cap P| = 0$ for every $ P \in V(G_{i+1}) \setminus \big(R_{G_{i+1}}(V(H)) \setminus B_{G_{i+1}}(V(H))\big )$.
    \end{itemize}
    So, we have $(H,S,A,\mathcal{J}) \in \mathcal{C}(i+1)$ and we can set $\zeta(i,H,S,A,\mathcal{J}) = \zeta(i+1,H,S,A,\mathcal{J})$ and $\Phi(i,H,S,A,\mathcal{J}) = \Phi(i+1,H,S,A,\mathcal{J})$.

    \medskip
    From now on, we may assume that either 
    \begin{itemize}
        \item $P^* \in V(H)$ or
        \item $P^* \in R_{G_i}(V(H)) \setminus B_{G_i}(V(H))$ with $(S \cup A) \cap P^* \neq \emptyset$.
    \end{itemize}
      For each $j \in [2]$, let $\mathcal{Q}_j = B_{G_{i+1}}(P_j) \cap (V(H)\setminus \{P^*\})$.
      By~\Cref{lem:contraction},  $\mathcal{Q}_1$ and $\mathcal{Q}_2$ are disjoint.

    We now explain, for a given IDP separator, how to construct a valid subset corresponding to it.
    Let $(\mathcal{O}, \mathcal{U})$ with $U = V(\bigcup \mathcal{U})$ be an IDP separator with respect to $(P_1, P_2, \mathcal{Q}_1, \mathcal{Q}_2)$ such that 
      \begin{itemize}
        \item if $P^* \in R_{G_i}(V(H)) \setminus B_{G_i}(V(H))$, then $U\cap P^*=(S\cup A)\cap P^*$.
    \end{itemize} 
    When $P^* \in R_{G_i}(V(H)) \setminus B_{G_i}(V(H))$, we require that our solution on $P^*$ is exactly $(S\cup A) \cap P^*$; thus, it suffices to consider $U$ with $U \cap P^* = (S\cup A) \cap P^*$.

      We define that
      \begin{itemize}
        \item $\mathcal{C}(\mathcal{O}, \mathcal{U})=\{H_1, \ldots, H_k\}$ is the set of components of $F-\mathcal{O}$, and
        \item $\mathcal{Q}(\mathcal{O}, \mathcal{U})=(Q_1, \ldots, Q_k)$ where for each $j\in [k]$, $Q_j$ is the union of all parts in $R_{G_{i + 1}}(V(H_j))$.
        \end{itemize}
    Let $F$ be the subtrigraph of $G_{i+1}$ induced by $(V(H)\setminus\{P^*\})\cup \{P_1, P_2\}$. By the black pair condition, each $H_j$ is a connected red-induced subtrigraph of $G_{i + 1}$. 
   
    By the definition of an IDP separator, every vertex in $U$ has degree $1$ or $2$ in $\bigcup \mathcal{U}$. 
    Let $U_1$ be the set of all vertices $v$ in $U$ having degree $1$ in $\mathcal{U}$. Let $U_1'$ be the set of vertices in $U_1\cap S$ whose unique neighbor in $\mathcal{U}$ is in $\bigcup_{X\in V(H)} X$.
    Note that the vertex in $U_1'$ is already an endpoint of some path in~$\mathcal{U}$, which is a starting path of a desired partial solution. So we do not treat it as an endpoint of a partial solution for components $H_j$. 
    Also, every vertex in $(U_1\cap T)\cap \bigcup_{X\in V(H)}X$ is a terminal that is already an endpoint of some path, which is a starting path of our desired partial solution. Thus, we also do not treat it as an endpoint of a partial solution for components $H_j$. 
    So, the set of remaining vertices to consider is
    \[M:=(S\cup U_1)\setminus U_1'\setminus \left((U_1\cap T)\cap \bigcup_{X\in V(H)}X \right).\]

  Let $\mathcal{S}(\mathcal{O}, \mathcal{U})$ be the collection of pairs $\{(S_1, A_1), \ldots, (S_k, A_k)\}$ such that for each $j \in [k]$,
        \begin{itemize}
            \item $S_j \subseteq M\cap Q_j$,
            \item every vertex in $M$ belongs to exactly one of $S_1, \ldots, S_k$, and
            \item $A_j := \left(A \cup (U\setminus S_j)\right)\cap Q_j$.
        \end{itemize}
        For each $\mathcal{S} \in \mathcal{S}(\mathcal{O}, \mathcal{U})$, we define $\mathcal{J}(\mathcal{O}, \mathcal{U}, \mathcal{S})$ as the collection of all $(\mathcal{J}_1, \ldots, \mathcal{J}_k)$ where each $\mathcal{J}_j$ is a partition of $S_j$ into pairs.

    Let $\mathcal{J}^*=\pairs(\bigcup\mathcal{U} - A - E(G[S]))$. Note that some vertex $w$ of $S$ may remain isolated in $\bigcup\mathcal{U} - A - E(G[S])$, which gives a pair $\{w\}$.
    For $(\mathcal{J}_1, \ldots, \mathcal{J}_k)\in \mathcal{J}(\mathcal{O}, \mathcal{U}, \mathcal{S})$, we define 
    \[(\mathcal{J}_1, \ldots, \mathcal{J}_k)\otimes \mathcal{J}^*\]
    as an auxiliary graph $Y$ on the vertices that appear in the pairs in  
   $\big(\bigcup_{q\in [k]}\mathcal{J}_q\big)\cup \mathcal{J}^*$ such that for two vertices $v$ and $w$ in $Y$, $v$ is adjacent to $w$ in $Y$ if and only if $\{v,w\}$ is a pair in $\big(\bigcup_{q\in [k]}\mathcal{J}_q\big)\cup \mathcal{J}^*$.

    \begin{claim}\label{claim:mergingpaths}
        $\zeta(i,H,S,A,\mathcal{J})=1$ if and only if there exists an IDP separator $(\mathcal{O}, \mathcal{U})$ such that 
        \begin{itemize}
            \item $\prod\limits_{j \in [k]} \zeta(i + 1, H_j, S_j, A_j, \mathcal{J}_j) =1$,
            \item $(\mathcal{J}_1, \ldots, \mathcal{J}_k)\otimes \mathcal{J}^*$ is a disjoint union of paths of length at least $1$, and 
            \item $\pairs\big((\mathcal{J}_1, \ldots, \mathcal{J}_k)\otimes \mathcal{J}^*\big)=\mathcal{J}$.
        \end{itemize}
    \end{claim}
    \begin{clproof}
        Suppose that $\zeta(i,H,S,A,\mathcal{J})=1$, that is, there exists a valid set $\mathcal{Z}$  with respect to $(i,H,S,A,\mathcal{J})$. Let $\mathcal{Z}^*$ be the set of components of $G[V(\bigcup \mathcal{Z})]$. Then $\mathcal{Z}^*$ is a set of mutually induced paths in $G$ since $\mathcal{Z}$ is valid.
        By \Cref{lem:sepIDP}, there exists a subset $\mathcal{O}\subseteq \{P_1, P_2\}\cup \mathcal{Q}_1\cup \mathcal{Q}_2$ such that $(\mathcal{O}, \mathcal{Z}^*_{\mathcal{O}})$ is an IDP separator with respect to $(P_1, P_2, \mathcal{Q}_1, \mathcal{Q}_2)$. Let $\mathcal{U}=\mathcal{Z}^*_{\mathcal{O}}$.

        Let $H_j\in \mathcal{C}(\mathcal{O}, \mathcal{U})$. Let $\mathcal{M}_j$ be the set of all maximal subpaths of a path in $\mathcal{Z}$ whose internal vertices are contained in $\bigcup_{X\in V(H_j)}X$. By the maximality, for each endpoint of a path in $\mathcal{M}_j$, either it is contained in $Q_j$ or it is a terminal contained in $T\cap \big(\bigcup_{X\in V(H_j)}X\big)$. Let $S_j$ be the union of the endpoints of paths in $\mathcal{M}_j$, and $A_j$ be the union of $A\cap Q_j$ and the set of all vertices in $V(\bigcup \mathcal{Z}) \cap Q_j$ that are not in $S_j$. Lastly, let $\mathcal{J}_j=\pairs(\mathcal{M}_j)$.

        We claim that the endpoints of two distinct paths in $\mathcal{M}_j$ are distinct. To verify this, assume that two paths $R_1$ and $R_2$ in $\mathcal{M}_j$ share an endpoint, say $r$. Then $R_1\cup R_2$ is a subpath of a path in $\mathcal{M}_j$. If $r$ is in $S\cup A$, then $\mathcal{Z}$ is not a valid solution with respect to $(i,H,S,A,\mathcal{J})$. Otherwise, $r$ is a vertex in the IDP separator, and in particular, $r$ is an isolated vertex in the IDP separator. But this contradicts the fact that $(\mathcal{O}, \mathcal{U})$ is an IDP separator, because no vertex in $\mathcal{U}$ is isolated in $\bigcup\mathcal{U}$.
        This implies that $\mathcal{J}_j$ is a partition of $S_j\cup (T\cap \bigcup_{X\in V(H_j)} X)$ into pairs.

        Thus, $(H_j, S_j, A_j, \mathcal{J}_j)\in \mathcal{C}(i+1)$ and $\mathcal{M}_j$ is a valid set with respect to $(i+1, H_j, S_j, A_j, \mathcal{J}_j)$. As we construct all $\mathcal{J}_j$'s from $\mathcal{Z}$, we have that $(\mathcal{J}_1, \ldots, \mathcal{J}_k)\otimes \mathcal{J}^*$ is a disjoint union of paths of length at least $1$, and $\pairs\big((\mathcal{J}_1, \ldots, \mathcal{J}_k)\otimes \mathcal{J}^*\big)=\mathcal{J}$.
        This proves the forward direction. 

        For the converse direction, suppose that there exists an IDP separator $(\mathcal{O}, \mathcal{U})$ satisfying the three conditions. For each $j\in [k]$, let $\mathcal{Z}_j$ be a valid set with respect to $(i + 1, H_j, S_j, A_j, \mathcal{J}_j)$. Because of the assumption on $(\mathcal{J}_1, \ldots, \mathcal{J}_k)\otimes \mathcal{J}^*$, the subgraph obtained from the union of all paths in \[\left(\bigcup_{j\in [k]}\mathcal{Z}_j\right)\cup \mathcal{U}\] 
        by removing edges on $G[S]$

        consists of a set $\mathcal{Z}$ of paths of length at least $1$ connecting the pairs in $\mathcal{J}$. Moreover, we know that there are no edges of $G$ between distinct sets $\bigcup_{X\in V(H_j)}X$. Also, $A_j$'s are chosen so that there are no edges between distinct paths in $\mathcal{Z}$. Therefore, $\mathcal{Z}$ is a valid set with respect to $(i,H,S,A,\mathcal{J})$.
    \end{clproof}
   
    Note that $(H_j, S_j, A_j, \mathcal{J}_j) \in \mathcal{C}(i + 1)$ because $R_{G_{i+1}}(H_j)\setminus \mathcal{O}\subseteq R_{G_{i}}(H)$ and $\mathcal{J}_j$ is a partition of $S_j$ into pairs. Therefore, $\zeta(i+1, H_j, S_j, A_j, \mathcal{J}_j)$ has been recorded.

    For $\mathcal{S} = \{(S_1, A_1), \ldots, (S_k, A_k)\} \in \mathcal{S}(\mathcal{O}, \mathcal{U})$ and $\Lambda = (\mathcal{J}_1, \ldots, \mathcal{J}_k) \in \mathcal{J}(\mathcal{O}, \mathcal{U}, \mathcal{S})$, we set
    $\mathcal{Z}_{\mathcal{O}, \mathcal{U},\mathcal{S}, \Lambda}$ as follows:
    \begin{itemize}
        \item If $\prod\limits_{j \in [k]} \zeta(i + 1, H_j, S_j, A_j, \mathcal{J}_j) = 1$, then $\mathcal{Z}_{\mathcal{O}, \mathcal{U},\mathcal{S}, \Lambda}$ is the set of connected components of the graph obtained from the union of all paths in
    \[\left(\bigcup\limits_{j = 1}^k \Phi(i + 1, H_j, S_{j}, A_j, \mathcal{J}_j)\right) \cup \mathcal{U}\]
    by removing edges on $G[S]$.
    \item It is the emptyset otherwise.
    \end{itemize}

    Let $\Psi(i, H, S, A, \mathcal{J})$ be the collection of $\mathcal{Z}_{\mathcal{O}, \mathcal{U},\mathcal{S}, \Lambda}$ over all possible choice of $(\mathcal{O}, \mathcal{U})$, $\mathcal{S} \in \mathcal{S}(\mathcal{O}, \mathcal{U})$, and $\Lambda \in \mathcal{J}(\mathcal{O}, \mathcal{U}, \mathcal{S})$ such that $\mathcal{Z}_{\mathcal{O}, \mathcal{U},\mathcal{S}, \Lambda} \neq \emptyset$. When $(\mathcal{O}, \mathcal{U})$ is fixed, it is sufficient to check all cases of $\mathcal{S}(\mathcal{O}, \mathcal{U})$ and $\mathcal{J}(\mathcal{O}, \mathcal{U}, \mathcal{S})$. By~\Cref{lem:componentcount}, the number of components of $F - \mathcal{O}$ is at most $t+2$. 
    
    If $\pairs(\mathcal{Z})=\mathcal{J}$
    for some $\mathcal{Z} \in \Psi(i, H, S, A, \mathcal{J})$, then by \Cref{claim:mergingpaths}, $\mathcal{Z}$ is a valid set with respect to $(i, H, S, A, \mathcal{J})$.
    Thus, in this case, we can set $\Phi(i, H, S, A, \mathcal{J}) = \mathcal{Z}$.

     If $\pairs(\mathcal{Z})\neq \mathcal{J}$
    for all $\mathcal{Z} \in \Psi(i, H, S, A, \mathcal{J})$, then by \Cref{claim:mergingpaths}, then there is no valid set with respect to $(i, H, S, A, \mathcal{J})$.  In this case, we can set $\Phi(i, H, S, A, \mathcal{J}) = \bot$.

    This proves the correctness of the algorithm.

\medskip 
    We analyze the running time. We first present an upper bound of $|\mathcal{C}(i)|$.

    \begin{claim}\label{claim:IDP:number of C(i)}
        $|\mathcal{C}(i)| \le n^{3t}\cdot (2t+1)\cdot (2t + 2\ell)^{2t}$.
    \end{claim}
    \begin{clproof}
         By~\Cref{lem:numberofinducedsubgraphs}, the number of non-empty connected red-induced subtrigraphs of $G_i$ is at most $n^{t}$. Let $H$ be a non-empty connected red-induced subtrigraph of $G_i$. By the same argument as in \Cref{claim:number of C(i)}, $|R_{G_i}(V(H)) \setminus B_{G_i}(V(H))| \le t$. It follows that $|S \cup A| \le 2t$. Since each vertex in $S \cup A$ can be assigned either to $S$ or to $A$, there are at most $(2t+1) \cdot n^{2t}$ possible choices. 

         We now consider the number of choices for $\mathcal{J}$. Note that $|S\cup (T\cap \bigcup_{X\in V(H)} X)| \le 2t + 2\ell$. 
         Every vertex $v \in S$ is paired with some vertex $v' \in S \cup T$, and the remaining vertices in $T\setminus\bigcup_{v \in S} \{v, v'\}$ are determined automatically because the pairs of $T$ are given. Thus, the total number of possible partitions is at most $(2t + 2\ell)^{2t}$. Therefore, for each $i \in [n]$, the total number of tuples $(H, S, A, \mathcal{J})$ is at most $(2t+1) \cdot (2t + 2\ell)^{2t}\cdot n^{3t}$.
     \end{clproof}

    By \Cref{lem:numberofIDPseparators}, the number of possible IDP separators is at most $4n^6$. We observe that the number of possible candidates in $\mathcal{S}(\mathcal{O}, \mathcal{U})$ and $\mathcal{J}(\mathcal{O}, \mathcal{U}, \mathcal{S})$ are bounded.

    \begin{claim}\label{claim:IDP:number of collection of sets}
        $|\mathcal{S}(\mathcal{O}, \mathcal{U})| \le k^{2t + 6}$.
    \end{claim}

    \begin{clproof}
        We observe that $|M| \le |S \cup U_1| \le |S \cup U| \le 2t + 6$. Since each vertex in $M$ is contained in one of $S_1, \ldots, S_k$, the total number of possible candidates is at most $k^{2t + 6}$.
    \end{clproof}

    \begin{claim}\label{claim:IDP:number of pairs}
        Let $\mathcal{S} \in \mathcal{S}(\mathcal{O}, \mathcal{U})$. Then $|\mathcal{J}(\mathcal{O}, \mathcal{U}, \mathcal{S})| \le (2t + 2\ell +5)^6$.
    \end{claim}

    \begin{clproof}
        By definition, it suffices to consider possible pairs having at least one vertex in $U$, as other pairs are determined from $\mathcal{J}$. Since $|U| \le 6$ and $|S| \le 2t$, $|S\cup (T\cap \bigcup_{X\in V(H)} X) \cup U| \le 2t + 2\ell +6$. Therefore, the number of possible candidates is at most $(2t+2\ell + 5)^6$.
    \end{clproof}
    By \Cref{claim:IDP:number of collection of sets} and \Cref{claim:IDP:number of pairs}, $|\mathcal{S}(\mathcal{O}, \mathcal{U})|\cdot|\mathcal{J}(\mathcal{O}, \mathcal{U}, \mathcal{S})| \le (t+2)^{2t+6}\cdot(2t + 2\ell+5)^6$. By \Cref{lem:numberofIDPseparators}, the number of possible IDP separators is at most $4n^6$, and we have \[|\Psi(i, H, S, A, \mathcal{J})| \le 4\cdot(t+2)^{2t+6}\cdot(2t + 2\ell+5)^6 \cdot n^6.\]

        For each $\mathcal{Z}\in \Psi(i, H, S, A, \mathcal{J})$, we can check whether $\pairs(\mathcal{Z})=\mathcal{J}$ or not in time $O(n)$. Thus, $\Phi(i, H, S, A, \mathcal{J})$ can be computed in time $O((t+2)^{2t+6}\cdot(2t + 2\ell+5)^6 \cdot n^7)$.

    Finally, for each $i \in [n]$, all of sets $\Phi(i, H, S, A, \mathcal{J})$ are computed in time $O(|\mathcal{C}(i)|\cdot (t+2)^{2t+6}\cdot(2t + 2\ell+5)^6 \cdot n^7)$. Overall, by \Cref{claim:IDP:number of C(i)}, this algorithm runs in time
    \[O\big(n^{3t + 8}\cdot (2t+1) \cdot (t+2)^{2t+6}\cdot (2t + 2\ell)^{2t}\cdot(2t + 2\ell+5)^6\big)=n^{O(t)}.\]     
   This proves the theorem.

\end{proof}

\section{Comparing width parameters}\label{sec:comparewidth}

\begin{table}
    \begin{center}
\begin{tabular}{ ||c|c|| } 
 \hline
 bounded $\reduced{\compmaxleaf}$ and unbounded mim-width & clique-column grid \\
   & Pohoata--Davies grid \\ \hline
 bounded $\reduced{\compmaxleaf}$ and unbounded stretch-width & unit interval graphs \\ \hline
 bounded mim-width and unbounded $\reduced{\compmaxleaf}$ & interval graphs \\ \hline
 bounded mim-width and unbounded stretch-width & unit interval graphs \\  \hline
  bounded stretch-width and unbounded $\reduced{\compmaxleaf}$ & subdivisions of walls \\  \hline
  bounded stretch-width and unbounded mim-width & subdivisions of walls \\ 
 \hline
\end{tabular}
\end{center}
    \caption{Separating examples}
    \label{table:separating}
\end{table}

In this section, we prove \Cref{thm:incomparable}, namely that 
reduced component max-leaf, stretch-width, and mim-width are mutually incomparable. We list examples separating two parameters in \Cref{table:separating}. Clique-column grids are illustrated in \Cref{fig:cliquecolumn}.

Here are known facts.
\begin{itemize}
    \item Unit-interval graphs have twin-width~2 \cite[Lemma 3.6]{twin-width3}, hence $\reduced{\compmaxleaf}$~2.
    \item Interval graphs have mim-width $1$~\cite{BelmonteV2013}.
    \item Sufficiently long subdivisions of grids have bounded stretch-width~\cite{BonnetDuron23}.
    \item Clique-column grids have unbounded mim-width~\cite[Lemma 3.8]{KangKST2017}.
    \item Mim-width, sim-width~\cite{KangKST2017}, and treewidth are equivalent on $K_{t,t}$-subgraph-free graphs~\cite{BrettellMPY2025}. Thus, Pohoata--Davies grids have unbounded mim-width and unbounded sim-width.
    \item Interval graphs have unbounded $\reduced{\compmaxleaf}$, as they have unbounded twin-width~\cite{twin-width2}.
    \item Subdivisions of walls have unbounded $\reduced{\compmaxleaf}$ (\Cref{cor:cml-degree-tw}) and unbounded mim-width~\cite[Theorem 8]{BrettellHMPP2022}. 
\end{itemize}

We need to additionally show that clique-column grids have bounded $\reduced{\compmaxleaf}$, and clique-column grids and unit interval graphs have unbounded stretch-width. 

\begin{figure}[!t]
    \centering
    \begin{tikzpicture}[scale=0.8]
        \tikzstyle{v}=[circle, draw, solid, fill=black, inner sep=0pt, minimum width=3pt]

        \node[v,label={[xshift=-3pt]left:{$v_{1,1}$}}] (u1) at (0,4.5) {};
        \node[v,label={[xshift=-9pt]left:{$v_{1,2}$}}] (u2) at (0,3) {};
        \node[v,label={[xshift=-9pt]left:{$v_{1,3}$}}] (u3) at (0,1.5) {};
        \node[v,label={[xshift=-3pt]left:{$v_{1,4}$}}] (u4) at (0,0) {};
        \node[v,label=below right:{$v_{2,1}$}] (v1) at (2.5,4.5) {};
        \node[v,label=below right:{$v_{2,2}$}] (v2) at (2.5,3) {};
        \node[v,label=below right:{$v_{2,3}$}] (v3) at (2.5,1.5) {};
        \node[v,label=below right:{$v_{2,4}$}] (v4) at (2.5,0) {};

        \node[v,label=right:{$v_{3,1}$}] (w1) at (5,4.5) {};
        \node[v,label=right:{$v_{3,2}$}] (w2) at (5,3) {};
        \node[v,label=right:{$v_{3,3}$}] (w3) at (5,1.5) {};
        \node[v,label=right:{$v_{3,4}$}] (w4) at (5,0) {};

        \draw (v1)--(u1);
        \draw (v1)--(u2);
        \draw (v1)--(u3);
        \draw (v1)--(u4);
        \draw (v2)--(u2);
        \draw (v2)--(u3);
        \draw (v2)--(u4);
        \draw (v3)--(u3);
        \draw (v3)--(u4);
        \draw (v4)--(u4);

        \draw (w1)--(v1);
        \draw (w1)--(v2);
        \draw (w1)--(v3);
        \draw (w1)--(v4);
        \draw (w2)--(v2);
        \draw (w2)--(v3);
        \draw (w2)--(v4);
        \draw (w3)--(v3);
        \draw (w3)--(v4);
        \draw (w4)--(v4);

        \draw (u1)--(u4);
        \draw (v1)--(v4);
        \draw (w1)--(w4);
 
        \draw[bend right=30] (u1) to (u4);
        \draw[bend right=20] (u1) to (u3);
        \draw[bend right=20] (u2) to (u4);

        \draw[bend right=30] (v1) to (v4);
        \draw[bend right=20] (v1) to (v3);
        \draw[bend right=20] (v2) to (v4);

        \draw[bend right=30] (w1) to (w4);
        \draw[bend right=20] (w1) to (w3);
        \draw[bend right=20] (w2) to (w4);

    \end{tikzpicture}
    \caption{A~$(3,4)$-mixed grid that is a unit interval graph.}
    \label{fig:34chain}
\end{figure}
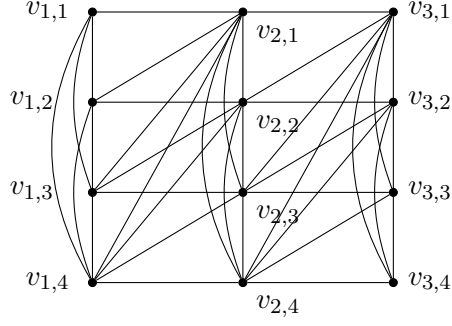

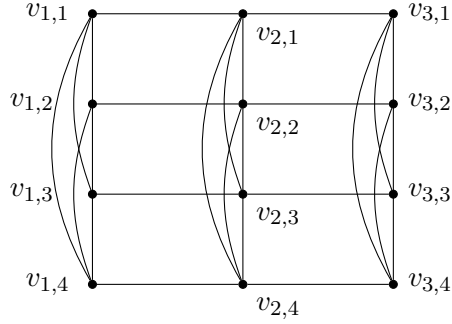
\begin{figure}[!t]
    \centering
    \begin{tikzpicture}[scale=0.8]
        \tikzstyle{v}=[circle, draw, solid, fill=black, inner sep=0pt, minimum width=3pt]

        \node[v,label={[xshift=-3pt]left:{$v_{1,1}$}}] (u1) at (0,4.5) {};
        \node[v,label={[xshift=-9pt]left:{$v_{1,2}$}}] (u2) at (0,3) {};
        \node[v,label={[xshift=-9pt]left:{$v_{1,3}$}}] (u3) at (0,1.5) {};
        \node[v,label={[xshift=-3pt]left:{$v_{1,4}$}}] (u4) at (0,0) {};
        \node[v,label=below right:{$v_{2,1}$}] (v1) at (2.5,4.5) {};
        \node[v,label=below right:{$v_{2,2}$}] (v2) at (2.5,3) {};
        \node[v,label=below right:{$v_{2,3}$}] (v3) at (2.5,1.5) {};
        \node[v,label=below right:{$v_{2,4}$}] (v4) at (2.5,0) {};

        \node[v,label=right:{$v_{3,1}$}] (w1) at (5,4.5) {};
        \node[v,label=right:{$v_{3,2}$}] (w2) at (5,3) {};
        \node[v,label=right:{$v_{3,3}$}] (w3) at (5,1.5) {};
        \node[v,label=right:{$v_{3,4}$}] (w4) at (5,0) {};

        \draw (v1)--(u1);
        \draw (v2)--(u2);
        \draw (v3)--(u3);
        \draw (v4)--(u4);

        \draw (w1)--(v1);
        \draw (w2)--(v2);
        \draw (w3)--(v3);
        \draw (w4)--(v4);

        \draw (u1)--(u4);
        \draw (v1)--(v4);
        \draw (w1)--(w4);
 
        \draw[bend right=30] (u1) to (u4);
        \draw[bend right=20] (u1) to (u3);
        \draw[bend right=20] (u2) to (u4);

        \draw[bend right=30] (v1) to (v4);
        \draw[bend right=20] (v1) to (v3);
        \draw[bend right=20] (v2) to (v4);

        \draw[bend right=30] (w1) to (w4);
        \draw[bend right=20] (w1) to (w3);
        \draw[bend right=20] (w2) to (w4);

    \end{tikzpicture}
    \caption{A~$(3,4)$-mixed grid that is a clique-column grid.}
    \label{fig:cliquecolumn}
\end{figure}

We introduce $(p,q)$-mixed grids and $(p,q)$-clique-column grids.
For integers $p, q \ge 1$, a \emph{$(p, q)$-mixed grid} is the graph on vertex set 
$V := \{v_{i,j}: i\in [p], j\in [q]\}$ where
\begin{itemize}
    \item for every $i\in [p]$, $\{v_{i,j}:j\in [q]\}$ is a clique, 

    \item for every $i\in [p-1]$ and $j_1, j_2\in [q]$ with $j_1\neq j_2$, there is at most one vertex in $\{v_{i+1, j_1}, v_{i+1, j_2}\}$ that is homogeneous to $\{v_{i, j_1}, v_{i, j_2}\}$, and
    there is at most one vertex in $\{v_{i, j_1}, v_{i, j_2}\}$ that is homogeneous to $\{v_{i+1, j_1}, v_{i+1, j_2}\}$,
    
    \item for every $i_1, i_2\in [p]$ and $j_1, j_2\in [q]$ with $|i_1-i_2|\ge 2$, $v_{i_1, j_1}$ is not adjacent to $v_{i_2, j_2}$.
\end{itemize}
Note that unit interval graphs are $(p,q)$-mixed grids; see Figure~\ref{fig:34chain} for an illustration.
For each $i\in [p]$, we call $\{v_{i,j}:j\in [q]\}$ a~column, and for each $j\in [q]$, we call $\{v_{i,j}:i\in [p]\}$ a~row of the $(p, q)$-mixed grid.  
A~mapping $\sigma$ from~$V$ to $[p q]$ is called the \emph{canonical ordering} if $\sigma(v_{i,j})=(i-1)q+j$ for every $i \in [p], j \in [q]$. 

That is, it orders the first column from the first row to the last row, then similarly orders the second column, and so on.

A $(p,q)$-mixed grid is called a \emph{$(p,q)$-clique-column grid} if it satisfies that 
\begin{itemize}
    \item for every $i\in [p-1]$ and $j_1, j_2\in [q]$, $v_{i,j_1}$ is adjacent to $v_{i,j_2}$ if and only if $j_1=j_2$.
\end{itemize}
We first show that $(p,q)$-clique-column grids have bounded $\reduced{\compmaxleaf}$.
\begin{lemma}\label{lem:cliquecolumn}
    For every positive integers $p,q$, the $(p,q)$-clique-column grid has reduced component max-leaf at most $3$.
\end{lemma}
\begin{proof}
    Let $G$ be the $(p,q)$-clique-column grid. We merge as follows:
    \begin{itemize}
        \item For each $i = 1, \ldots, p$, merge $v_{i,1}$ and $v_{i,2}$. For convenience, we denote the resulting merged vertex by $v_{i,2}$ for each $i \in [p]$. Next, for each $i = 1, \ldots, p$, merge $v_{i,2}$ and $v_{i,3}$. We repeat this process until only the last row remains.
        \item We recursively merge a vertex of degree $1$ and its neighbor.
    \end{itemize}
    It is straightforward to verify that every red graph is a subdivision of $K_{1,3}$. Therefore, $G$ has reduced component max-leaf at most $3$.
    \end{proof}

In the remainder, we show that clique-column grids and unit interval graphs have unbounded stretch-width. 

We start with two technical lemmas on $(p, q)$-mixed grids and trigraphs. 

\begin{lemma}\label{lem:findingmixedset}
	Let $p, q, d, r$ be positive integers and 
	let $G$ be a $(p, q)$-mixed grid. 
	Let $\mathcal{P}$ be a partition of $V(G)$ such that every part of~$\mathcal{P}$ intersects at most two columns of $G$, and the maximum red degree of $G/\mathcal{P}$ is at most $d$.
	
	\begin{enumerate}[(1)]
	\item If there exist $i\in [p-1]$ and $J\subseteq [q]$ such that $|J|\ge (d+1)r$ and $\{v_{i,j}:j\in J\}$ is contained in the same part $P$ of $\mathcal{P}$, then there exists a subset $J'$ of $J$ of size at least $r$ such that $\{v_{i+1,j}:j\in J'\}$ is contained in the same part of $\mathcal{P}$.
	\item If there exist $i\in [p]\setminus \{1\}$ and $J\subseteq [q]$ such that $|J|\ge (d+1)r$ and $\{v_{i,j}:j\in J\}$ is contained in the same part $P$ of $\mathcal{P}$, then there exists a subset $J' \subseteq J$ of size at least $r$ such that $\{v_{i-1,j}:j\in J'\}$ is contained in the same part of $\mathcal{P}$.
	\end{enumerate}
\end{lemma}
\begin{proof}
	By symmetry, it suffices to show (1). 

	As $G$ is a~$(p, q)$-mixed grid, by assumption, 
	there is at most one vertex in $\{v_{i+1, j}:j\in J\}$ that is homogeneous to $\{v_{i, j}:j\in J\}$.
	Let $b \in J$ be such that no vertex of $\{v_{i+1,j}:j\in J \setminus \{b\}\}$ is homogeneous to $\{v_{i,j}:j\in J\}$. 
   
   If $\{v_{i+1,j}:j\in J\setminus \{b\}\}$ contains at least $r$ vertices of $P$, then we are done. 
   So, we may assume that $\{v_{i+1,j}:j\in J\setminus \{b\}\}$ contains less than $r$ vertices of $P$. 
   Let $J_1 \subseteq J\setminus \{b\}$ be such that 
   $\{v_{i+1,j}:j\in~J_1\} = \{v_{i+1,j}:j\in J\setminus \{b\}\} \setminus P$.
   Note that $|J_1| \ge (d+1)r-r=dr$.
   
   By assumption, the part $P$ containing $\{v_{i,j}:j\in J\}$  has red degree at most $d$ in $G/\mathcal{P}$.
   This implies that there are at most $d$ parts of $\mathcal{P}$ intersecting $\{v_{i+1,j}:j\in J_1\}$. Thus, there is a subset $J_2\subseteq J_1$ of size at least  
    $\frac{dr}{d}\ge r$
    such that 
    $\{v_{i+1,j}:j\in J_2\}$ is contained in the same part of $\mathcal{P}$. This proves the lemma.       
\end{proof}

\begin{lemma}\label{lem:mixedgridtwocolumn}
	There is a function $f:\mathbb{N}\to \mathbb{N}$ satisfying the following. 
	Let $t$ be a positive integer and 
	let $G$ be a~$(6t+6, f(t))$-mixed grid.
	Let $\psi$ be a linear ordering of $G$, and let $\mathcal{P}^*$ be a partition of $V(G)$ such that each part of $\mathcal{P}^*$ is either a column of $G$, or the union of two consecutive columns of~$G$.
	Let $\Pc_n,\dots,\Pc_1$ be a contraction sequence of $G$ such that the sequence contains $\mathcal{P}^*$. Then there exists $i\in [n]$ such that the stretch of $\mathcal{P}_i$ in $(G, \psi)$ is more than $t$.
\end{lemma}

\begin{proof}
   We define $f_1(t)= 2(t+1)^{6t+5}$, and $f(t)=4(t+1)(6t+5)f_1(t)$.
    For each $i\in [6t+6]$, let $C_i=\{v_{i,j}:j\in [f(t)]\}$ be the $i$-th column.

    For each $i\in [6t+6]$, let $a_i=\min (\psi(C_i))$ and $b_i=\max (\psi(C_i))$.
    
    Observe that if there is an integer $i\in [6t+6]$ such that the interval $[a_i,b_i]$ intersects at least $2t+2$ other intervals $[a_j, b_j]$ for $j\in [6t+6]\setminus \{i\}$, then the stretch of $\mathcal{P}^*$ is at least $t+1$. Thus, we may assume that each interval $[a_i,b_i]$ intersects at most $2t+1$ other intervals $[a_j, b_j]$. 
    
    This implies that the intersection graph of these $6t+6$ intervals has maximum degree at most $2t+1$, and thus has maximum clique size at most $2t+2$. As interval graphs are perfect, this intersection graph has an independent set of size at least $(6t+6)/(2t+2)=3$. Let $i_1, i_2, i_3\in [6t+6]$ such that the intervals $[a_{i_1}, b_{i_1}], [a_{i_2}, b_{i_2}], [a_{i_3}, b_{i_3}]$ are pairwise disjoint. Without loss of generality, we assume that $b_{i_1}<a_{i_2}$ and $b_{i_2}<a_{i_3}$. Observe that 
    \[a_{i_3}-b_{i_1}> b_{i_2}-a_{i_2}+1\ge f(t). \] 
    
    Now, let $\alpha$ be the maximum integer such that $\Pc_{\alpha}$ has a part $P$ containing at least $f_1(t)$ vertices in the same column, say $C_z$. 
    Let $J\subseteq [f(t)]$ be such that $P\cap C_z=\{v_{z,j}:j\in J\}$. 
    We show that the stretch of $\Pc_{\alpha}$ is more than $t$. Suppose that it has stretch at most $t$.
    
    We may assume that $G/\Pc_{\alpha}$ has maximum red degree at most $t$, otherwise, the stretch of $\Pc_{\alpha}$ is more than $t$ by Lemma~\ref{lem:largereddegree}. Also, $\Pc_{\alpha}$ appears earlier than $\mathcal{P}^*$ and it implies that every part of  $\mathcal{P}_\alpha$ intersects at most two columns of $G$.
     Thus, by applying Lemma~\ref{lem:findingmixedset} $6t+5$ times, we get a subset $J'\subseteq J$ of size at least 
     $\frac{f_1(t)}{(t+1)^{6t+5}}\ge 2$
     such that 
     \begin{itemize}
     	\item for each $i\in [6t+6]$, $U_i := \{v_{i,j}:j\in J'\}$ is contained in the same part of $\mathcal{P}_{\alpha}$. 
     \end{itemize}
     For each $i\in [6t+6]$, let $c_i=\min (\psi(U_i))$ and $d_i=\max(\psi(U_i))$. Note that by the condition of the $(p,q)$-mixed grid, for each $i\in [6t+5]$, the part containing $U_i$ and the part containing $U_{i+1}$ are adjacent by a red edge in $G/\Pc_{\alpha}$ if they are distinct parts.
    
     We claim that for each $i\in [6t+5]$, $|d_{i+1}-c_i|\le 4(t+1)f_1(t)$. 
     Suppose for contradiction that for some $i\in [6t+5]$, $|d_{i+1}-c_i|> 4 (t+1)f_1(t)$. Note that every part of $\Pc_{\alpha}$ has size at most $4f_1(t)$. Thus, there are at least $t+1$ parts of $\Pc_{\alpha}$ interfering with the part containing $U_i$. Then the stretch of $\Pc_{\alpha}$ is more than $t$, a~contradiction. Thus, $|d_{i+1}-c_i|\le 4(t+1)f_1(t)$.

     It implies that for any $i, i'\in [6t+6]$ with $i'>i$, \[|d_{i'}-c_i|\le 4(t+1)(i'-i)f_1(t)\le 4(t+1)(6t+5)f_1(t)=f(t).\]
     
     If $i_3>i_1$, then $|d_{i_3}-c_{i_1}| \le f(t).$
     However, $|d_{i_3}-c_{i_1}|\ge a_{i_3}-b_{i_1}> f(t)$, a contradiction. If $i_3<i_1$, then $|d_{i_1}-c_{i_3}| \le f(t)$, and this contradicts the fact that $|d_{i_1}-c_{i_3}|\ge a_{i_3}-b_{i_1}> f(t)$. Thus, the stretch of $\Pc_{\alpha}$ is more than $t$.   
\end{proof}

We can now show that the stretch-width of $(m,m)$-mixed grids is an unbounded function of~$m$.

\begin{theorem}
 There is a function $g:\mathbb{N}\to \mathbb{N}$ such that any $(g(t), g(t))$-mixed grid has stretch-width at least $t+1$.
\end{theorem}
\begin{proof}
	Let $f$ be the function in Lemma~\ref{lem:mixedgridtwocolumn}. We define
	\begin{itemize}
		\item $g_2(t)=f(t)$, 
		\item  $g_1(t)=g_2(t)(t+1)^{6t+5}$, and 
		\item $g(t)=\max ( (t+3)g_1(t), 12t+12)$.
	\end{itemize}
	Let $m=g(t)$.

    Let $G$ be a $(m, m)$-mixed grid with canonical ordering $\sigma$. 
    Suppose that $G$ has stretch-width at most $t$. Then there is a linear ordering $\psi$ on $V(G)$ such that $\stw(G, \psi)\le t$. 
    So, there is a~contraction sequence $\Pc_n,\dots,\Pc_1$ such that the stretch of each $\Pc_i$ with respect to $\psi$ is at most $t$. 

Let $\alpha$ be the maximum integer such that $\Pc_{\alpha}$ has a part $P$ satisfying one of the following:
\begin{enumerate}[(1)]
    \item\label{cond:part1} $P$ contains two vertices $v$ and $w$ that are contained in non-consecutive columns.
    \item\label{cond:part2} $P$ contains at least $g_1(t)$ vertices in the same column. 
\end{enumerate}
As $\Pc_{\alpha+1}$ has no part satisfying Condition~(\ref{cond:part1}) or (\ref{cond:part2}), $P$ is the unique part of $\Pc_{\alpha}$ satisfying Condition~(\ref{cond:part1}) or (\ref{cond:part2}). Also, $P$ contains at most $2g_1(t)$ in the same column. 

    First assume that $P$ satisfies Condition~(\ref{cond:part1}). 
    Let $v$ and $w$ be vertices that are contained in non-consecutive columns. 
    By symmetry, we may assume that $\sigma(v)< \sigma(w)$.
    
    Let $C$ be the column containing $v$. As $P$ contains at most $2g_1(t)$ vertices in $C$, $C$ intersects at least 
	\[ \frac{m-2g_1(t)}{g_1(t)}\ge t+1 \]
	parts of $\mathcal{P}_{\alpha}$ other than $P$.
As $w$ is anti-complete to $C$, these parts are red neighbors of $P$ in $G/\mathcal{P}_{\alpha}$.
	By Lemma~\ref{lem:largereddegree}, $\mathcal{P}$ has stretch at least $t+1$.  
This contradicts the assumption that the stretch of $\mathcal{P}_\alpha$ is at most $t$.

  Therefore, we may assume that $P$ does not satisfy Condition~(\ref{cond:part1}) and thus satisfies Condition~(\ref{cond:part2}). 
  It implies that every part of $\Pc_{\alpha}$ intersects at most two columns of $G$. 
   For each $i\in [m]$, let $C_i=\{v_{i,j}:j\in [m]\}$ be the $i$-th column.
   Let $q\in [m]$ be such that $P$ contains at least $g(t)$ vertices of $C_q$, and 
    let $J\subseteq [m]$ be the set of integers $j\in [m]$ such that $v_{q,j}\in P$. 
    By symmetry, we may assume that $q\le m/2$. Note that $m/2\ge 6t+6$.

       We can also assume that $G/\Pc_{\alpha}$ has maximum red degree at most $t$, otherwise, the stretch of $\Pc_{\alpha}$ is more than $t$ by Lemma~\ref{lem:largereddegree}.   
     Thus, by applying Lemma~\ref{lem:findingmixedset} $6t+5$ times, we get a subset $J'\subseteq J$ of size at least 
     $\frac{g_1(t)}{(t+1)^{6t+5}}\ge g_2(t)$
     such that 
     \begin{itemize}
     	\item for each $i\in [6t+6]$, $\{v_{q+i-1,j}:j\in J'\}$ is contained in the same part of $\mathcal{P}_{\alpha}$. 
     \end{itemize}

 For each $i\in [6t+6]$, let $C_i'=\{v_{q+i-1,j}:j\in J'\}$. 
 Let $G'=G[\bigcup_{i\in [6t+6]} C_i']$. We consider the restriction of the contraction sequence $\Pc_n,\dots,\Pc_1$ on $G'$.
 Formally, for each $z\in [n]$, let $\Pc_z'$ be the partition obtained from $\Pc_z$  by replacing every $Q\in \Pc_z$ with $Q\cap V(G')$, and then removing empty set. 
 From $\Pc_n',\dots,\Pc_1'$, if two consecutive partitions are the same, then we remove one of them. Let $\Pc_{r}'',\dots,\Pc_1''$ be the obtained contraction sequence.
 Let $\beta\in [r]$ be such that $\Pc_\beta''$ is obtained from $\Pc_{\alpha}$.
  
  Observe that $G'$ is a $(6t+6, g_2(t))$-mixed grid and $\Pc_\beta''$ is a partition of $G'$ where each part is either a column of $G'$ or the union of two consecutive columns, because
  $\Pc_\alpha$ is a partition of $G$ where each part of $G$ is contained in the union of two consecutive columns, and each $C_i'$ is contained in one part of $\Pc_\alpha$.
  Also, $\Pc_{r}'',\dots,\Pc_1''$ is the contraction sequence of $G'$ containing $\Pc_\beta''$.
  Let $\psi'$ be the restriction of $\psi$ to $V(G')$.
 By Lemma~\ref{lem:mixedgridtwocolumn}, there exists $i\in [r]$ such that the stretch of $\mathcal{P}_i''$ in $(G', \psi')$ is more than $t$. This implies that there exists $i'\in [n]$ such that 
 the stretch of $\mathcal{P}_{i'}$ in $(G, \psi)$ is more than $t$, and proves the theorem.
\end{proof}

\section{Conclusion}

We provide polynomial-time algorithms for \textsc{$\sigma$-Neighborhood} and \textsc{Induced Disjoint Paths} on classes of effectively bounded $\compmaxleaf^\downarrow$. We may consider other problems. Telle and Proskurowski~\cite{TelleP1997} considered a more general problem called the \textsc{$(\sigma,\rho)$-Problem}, in which it is further required that, for every vertex $v$ outside a desired set $S$, the number of neighbors of $v$ in $S$ belongs to~$\rho$. \textsc{Dominating Set} is such a problem, which is (Min) \textsc{$(\mathbb{N},\mathbb{N}\setminus \{0\})$-Problem}. Also, the \textsc{Feedback Vertex Set} and \textsc{$3$-Coloring} problems are natural candidates to consider. We suspect that \textsc{Dominating Set} and \textsc{$3$-Coloring} are NP-hard on some classes of effectively bounded reduced component max-leaf. On the other hand, it seems more plausible that the \textsc{Feedback Vertex Set} problem admits a polynomial-time algorithm on classes of graphs with effectively bounded reduced component max-leaf. They can be solved in polynomial time on classes of effectively bounded mim-width~\cite{BuixuanTV2013,JaffkeKT2020-2,BergougnouxPT2022}.

Gartland and Lokshtanov~\cite[Conjecture 1.4.2]{GartlandThesis} conjectured that \textsc{$3$-Coloring} and $(\operatorname{tw}\le r, \psi)$-\textsc{Max Weighted Induced Subgraph}, which generalizes \textsc{Feedback Vertex Set}, are polynomial-time solvable on classes that exclude a fixed planar graph $H$ as an induced minor. By \Cref{cor:cml-degree-tw}, every class of bounded $\reduced{\compmaxleaf}$ excludes some wall as an induced minor.  Thus, showing that \textsc{3-Coloring} is NP-hard on some class of graphs with effectively bounded reduced component max-leaf would refute the above conjecture for \textsc{3-Coloring}.

The \textsc{Disjoint Paths} problem asks for a set of vertex-disjoint paths connecting given pairs, without requiring that there are no edges between the paths.  Since it is known to be NP-hard on interval graphs~\cite{NatarajanS1996}, we cannot expect to solve it in polynomial time on classes of bounded mim-width. However, the class of interval graphs has unbounded reduced component max-leaf. Thus, it remains open whether the \textsc{Disjoint Paths} problem can be solved in polynomial time on classes of graphs with effectively bounded $\compmaxleaf^\downarrow$.

Our algorithms are all based on assuming that a contraction sequence is given.
We ask whether there is an XP-time (or even an FPT-time) approximation algorithm
that, given a graph $G$ and an integer $d$, either outputs a contraction sequence of $\compmaxleaf^\downarrow (G)\le f(d)$ for some function $f$ or correctly reports that $\compmaxleaf^\downarrow(G)>d$.  For twin-width, Berg\'e, Bonnet, and D\'epr\'es~\cite{BergeBD2022} showed that deciding whether a graph has twin-width at most $4$ is NP-complete, which rules out the possibility of an XP-time exact algorithm for testing whether the twin-width is at most $k$.

\paragraph{Acknowledgments.}
This research was initiated at the Oberwolfach Workshop on Graph Theory in January 2025.

\end{document}